%% file: AAAI.tex
\documentclass[letterpaper]{article} 
\usepackage{aaai25}  
\usepackage{times}  
\usepackage{helvet}  
\usepackage{courier}  
\usepackage[hyphens]{url}  
\usepackage{graphicx} 
\usepackage{natbib}  
\usepackage{caption} 
\usepackage{algorithm}
\usepackage{algorithmic}
\usepackage{color, soul}
\usepackage{subfigure}
\usepackage{appendix}
\usepackage{titletoc}

\newcommand{\Var}{\text{\rm Var}}
\newcommand{\dist}{\text{\rm dist}}
\newcommand{\cd}{\!\cdot\!}

\newcommand{\HL}[1]{}

\input{commands}

\usepackage{newfloat}
\usepackage{listings}
\DeclareCaptionStyle{ruled}{labelfont=normalfont,labelsep=colon,strut=off} 
\floatstyle{ruled}
\newfloat{listing}{tb}{lst}{}
\floatname{listing}{Listing}
\title{Sample Complexity of Linear Regression Models for Opinion Formation in Networks}

\author {
    Haolin Liu\textsuperscript{\rm 1}\thanks{Alphabetical order},
    Rajmohan Rajaraman \textsuperscript{\rm 2},
    Ravi Sundaram\textsuperscript{\rm 2}, \\
    Anil Vullikanti\textsuperscript{\rm 1},
    Omer Wasim\textsuperscript{\rm 2},
    Haifeng Xu\textsuperscript{\rm 3}
}
\affiliations {
    \textsuperscript{\rm 1}University of Virginia \\
     \{srs8rh, asv9v\}@virginia.edu,
    
    \textsuperscript{\rm 2}Northeastern University\\
    \{r.rajaraman, r.sundaram, wasim.o\}@northeastern.edu\\
     \textsuperscript{\rm 3} University of Chicago\\
      haifengxu@uchicago.edu
}

\usepackage{bibentry}

\begin{document}

\maketitle

\begin{abstract}
Consider public health officials aiming to spread awareness about a new vaccine in a community interconnected by a social network. How can they distribute information with minimal resources, so as to avoid polarization and ensure community-wide convergence of opinion? 
To tackle such challenges, we initiate the study of sample complexity of opinion formation in networks. Our framework is built on the recognized \emph{opinion formation game}, where we regard each agent's opinion as a data-derived model, unlike previous works that treat opinions as data-independent scalars. The opinion model for every agent is initially learned from its local samples and evolves game-theoretically as all agents communicate with neighbors and revise their models towards an equilibrium. 
Our focus is on the sample complexity needed to ensure that the opinions converge to an equilibrium such that every agent's final model has low generalization error. 

Our paper has two main technical results.  First, we present a novel
polynomial time optimization framework to quantify the \emph{total
  sample complexity} for arbitrary networks, when the underlying
learning problem is (generalized) linear regression.  Second, we leverage
this optimization to study the \emph{network gain} which measures the improvement
of sample complexity when learning over a network compared to that in
isolation.  Towards this end, we derive network gain bounds for
various network classes including cliques, star graphs, and
random regular graphs.  Additionally, our framework provides a method to study sample distribution within the
network, suggesting that it is sufficient to allocate samples
inversely to the degree.  Empirical results on both synthetic and real-world networks strongly support our theoretical findings.


\end{abstract}

%

\input{Introduction}

\input{New_Formulation}
\input{Network_Affect}
\input{Experiments}

\input{Discussion}

\clearpage
\section*{Acknowledgments}
RR and OW were partially supported by NSF grant CCF-2335187.
AV was partially supported by NSF grants CCF-1918656, IIS-1955797, CNS 2317193 and IIS 2331315. RS was  partially supported by NSF grant D-ISN 2146502.

\bibliography{references}


\newpage
\onecolumn 
\appendix

\appendixpage

{
\startcontents[section]
\printcontents[section]{l}{1}{\setcounter{tocdepth}{2}}
}

\newpage
\input{preliminaries-appendix}
\input{app-2-proof}
\input{app-3.1-proof}
\input{app-3.2-proof}

\input{app-exp-proof}
\input{app-more-exp}
\end{document}

%% file: commands.tex
\usepackage{xspace}
\usepackage{amsmath}
\usepackage{graphicx}
\usepackage{framed}
\usepackage{algorithm}
\usepackage{amsthm}
\usepackage[algo2e, vlined, ruled, linesnumbered]{algorithm2e}
\usepackage{natbib}
\usepackage{mathrsfs}
\usepackage{amssymb}
\usepackage{bm}
\usepackage{makecell}
\usepackage{tabulary}
\usepackage{xcolor}
\usepackage{prettyref}
\usepackage{mathtools}
\usepackage{amsmath}
\usepackage{multirow}
\usepackage{nicefrac}
\usepackage{amssymb}
\usepackage{pifont}

\definecolor{Green}{rgb}{0.13, 0.65, 0.3}
\usepackage{colortbl}
\DeclareMathOperator*{\argmin}{argmin} 

\newcommand{\E}{\mathbb{E}}

\newtheorem{theorem}{Theorem}
\newtheorem{assumption}{Assumption}

\newtheorem{lemma}[theorem]{Lemma}
\newtheorem{corollary}[theorem]{Corollary}
\newtheorem{observation}[theorem]{Observation}
\usepackage{tikz}

\newcommand{\nonl}{\renewcommand{\nl}{\let\nl}}

\usepackage{prettyref}

\newcommand{\savehyperref}[2]{\texorpdfstring{\hyperref[#1]{#2}}{#2}}
\newrefformat{eq}{\savehyperref{#1}{Eq.~\textup{(\ref*{#1})}}}
\newrefformat{eqn}{\savehyperref{#1}{Equation~\ref*{#1}}}
\newrefformat{lem}{\savehyperref{#1}{Lemma~\ref*{#1}}}
\newrefformat{lemma}{\savehyperref{#1}{Lemma~\ref*{#1}}}
\newrefformat{def}{\savehyperref{#1}{Definition~\ref*{#1}}}
\newrefformat{line}{\savehyperref{#1}{Line~\ref*{#1}}}
\newrefformat{thm}{\savehyperref{#1}{Theorem~\ref*{#1}}}
\newrefformat{corr}{\savehyperref{#1}{Corollary~\ref*{#1}}}
\newrefformat{cor}{\savehyperref{#1}{Corollary~\ref*{#1}}}
\newrefformat{sec}{\savehyperref{#1}{Section~\ref*{#1}}}
\newrefformat{subsec}{\savehyperref{#1}{Section~\ref*{#1}}}
\newrefformat{app}{\savehyperref{#1}{Appendix~\ref*{#1}}}
\newrefformat{assum}{\savehyperref{#1}{Assumption~\ref*{#1}}}
\newrefformat{ex}{\savehyperref{#1}{Example~\ref*{#1}}}
\newrefformat{fig}{\savehyperref{#1}{Figure~\ref*{#1}}}
\newrefformat{alg}{\savehyperref{#1}{Algorithm~\ref*{#1}}}
\newrefformat{rem}{\savehyperref{#1}{Remark~\ref*{#1}}}
\newrefformat{conj}{\savehyperref{#1}{Conjecture~\ref*{#1}}}
\newrefformat{prop}{\savehyperref{#1}{Proposition~\ref*{#1}}}
\newrefformat{proto}{\savehyperref{#1}{Protocol~\ref*{#1}}}
\newrefformat{prob}{\savehyperref{#1}{Problem~\ref*{#1}}}
\newrefformat{claim}{\savehyperref{#1}{Claim~\ref*{#1}}}
\newrefformat{que}{\savehyperref{#1}{Question~\ref*{#1}}}
\newrefformat{op}{\savehyperref{#1}{Open Problem~\ref*{#1}}}
\newrefformat{fn}{\savehyperref{#1}{Footnote~\ref*{#1}}}
\newrefformat{tab}{\savehyperref{#1}{Table~\ref*{#1}}}
\newrefformat{fig}{\savehyperref{#1}{Figure~\ref*{#1}}}
\newrefformat{proc}{\savehyperref{#1}{Procedure~\ref*{#1}}}

\usepackage{caption}

\newcommand{\junk}[1]{}
\newcommand{\ngain}{\mu}

%% file: Introduction.tex
\section{Introduction}

In today's interconnected world, rapid dissemination of information plays a pivotal role in shaping public understanding and behavior, especially in areas of critical importance like public health. 
People learn and form opinions through personal investigations and interactions with neighbors. 
The heterogeneity of social networks often means that not every individual requires the same amount of information to form an informed opinion. Some may be heavily influenced by their peers, while others may need more direct information. 
This differential requirement presents both a challenge and an opportunity: How can one distribute information to ensure the entire community is well-informed, without unnecessary redundancies or gaps in knowledge dissemination? 

To answer this question, the first step is to understand the dynamics of opinion formation within social networks. 
Starting from the seminal work of \citet{degroot1974reaching}, one of the predominant models assumes that individuals shape their own beliefs by consistently observing and averaging the opinions of their network peers  \citep{degroot1974reaching, friedkin1990social, demarzo2003persuasion, golub2008homophily, bindel2015bad, haddadan2024optimally}. 
We will adopt this seminal framework to model the opinion formation process among agents. However, these works always encapsulate opinions as single numeric values. While this provides a foundational
understanding, it cannot capture the nuanced ways in which individuals process and integrate new information into their existing belief systems. 
To address this issue, our formulation generalizes the opinions to be a model parameter learned using a machine learning-based approach, as explored in \cite{haghtalab2021belief} to study belief polarization. Such formulation aligns with the human decision-making procedure proposed by \cite{simon2013administrative} based on prior experiences and available data.


Specifically, we adopt the game theoretical formulation of this opinion dynamic process introduced by \cite{bindel2015bad}. Within this framework, an agent's equilibrium opinion emerges as a weighted average of everyone's initial opinions on the network (which are learned form their locally available data), where the weights are unique and determined by the network structure. 
It is not guaranteed, however, that a node's equilibrium opinion would have a small generalization error.
Motivated by the work of~\cite{blum2021one}, we study the conditions under this actually happens (i.e., each node has a model with a small generalization error): specifically, how should samples be distributed across the network such that at the
equilibrium of the opinion formation game, everyone has a
model that is close enough to the ground-truth, and the total number of samples is minimized;
we are also interested in understanding the network gain, i.e., how much does the collaboration on a network help the sample complexity.

We note that when opinions are treated as data-driven model parameters, the equilibrium models share the same mathematical structure as the \textit{fine-grained federated learning} model introduced in \cite{donahue2021model}. However, their study assumes fixed sample sizes without considering networks and focuses on collaboration incentives. Our work focuses on the optimal allocation of samples across networks to ensure that the generalization error of the equilibrium model remains low.

\subsection{Problem formulation}
\label{sec:intro problem}
We consider a fixed set $V = \{1, 2, \cdots, n\}$ of agents (also referred to as nodes) connected by a network $G = (V, E)$.
Let $N(i)$ denote the set of neighbors, and $d_i$ the degree of agent $i$.
We assume every agent $i$
in a given network $G$ has a dataset $S_i = \{(x_{i}^j,
y_{i}^j)\}_{j \in [m_i]}$ allocated by a market-designer, where $m_i$ is the number of samples at
$i$, and $x_i^j \in \mathbb{R}^k, \,\, \forall j \in [m_i]$, and are
independently drawn from a fixed distribution $D$ with dimension $k$. Each agent $i$ learns an
initial model $\Bar{\theta}_i$ (known as internal opinion);
$\Bar{\theta}_i$ need not be a fixed number (as
in~\cite{bindel2015bad}), but can be learned through a machine
learning approach, similar as~\cite{haghtalab2021belief}.
We assume that these datasets are
allocated by a system designer and that $m_i \ge k$ (i.e.,
$S_i$ has at least $k$ samples), which ensures everyone has enough
basic knowledge to form their own opinion before learning over the
social network. Our full paper \cite{liu2023sample} includes a table of all notations.
 
Agent $i$'s goal
is to find a model $\theta_i$ that is close to its internal opinion,
as well as to that of its neighbors, denoted by $N(i)$; i.e. compute
$\theta_i$ that minimizes the loss function $\|\theta_i -
\bar{\theta}_i\|^2 + \sum_{j\in N(i)} v_{ij} \|\theta_i -
\theta_j\|^2$, where $v_{ij} \ge 0$ measures the influence of a
neighboring agent $j$ to agent $i$.  We refer to $\pmb{v}$, where
$[\pmb{v}]_{ij} = v_{ij}$, as the influence factor matrix.
In Lemma~\ref{lemma:potential}, we show that the unique Nash
equilibrium of this game is $(\theta^{eq}_1, \cdots, \theta^{eq}_n)^T
= W^{-1}(\Bar{\theta}_1, \cdots, \Bar{\theta}_n)^T$ where $W =
\mathcal{L} + I$ and $\mathcal{L}$ is the weighted Laplacian matrix of
$G$.


We define the \textbf{\emph{total sample complexity}}, $TSC(G,
\pmb{v}, k, \epsilon)$, as the minimum number of total samples, $\sum_i
m_i$, subject to the constraint $m_i \ge k$, so that the expected
square loss of $\theta_i^{eq}$ is at most $\epsilon$ for all $i$; here
$m_i = |S_i|$.  In the special case where the influence is uniform,
i.e., $v_{ij} = \alpha, \forall i,j$, we use $TSC(G, \alpha,k,
\epsilon)$ to denote the total sample complexity.

Let $\widetilde{M}(k, \epsilon)$ denote the minimum number of samples that any agent
$i$ would need to ensure that the best model $\bar{\theta}_i$ learned
using only their samples ensures that the expected error is at most
some target $\epsilon$; if there is no interaction between the agents,
a total of $n\widetilde{M}(k, \epsilon)$ samples would be needed to ensure that
$\bar{\theta}_i$ has low error for each agent.  Since every agent has
at least $k$ samples, we refer to
\begin{equation}
    \ngain(G, \pmb{v}, k, \epsilon) = \frac{n(\widetilde{M}(k, \epsilon) - k)}{TSC(G, \pmb{v}, k, \epsilon) - nk},
    \label{eq:NetG}
\end{equation}
the ratio of the additional samples needed to achieve error of
$\epsilon$ for every agent under social learning to that needed under
independent learning, as the \textbf{\emph{network gain}} of $G$. We will show that for linear regression $\widetilde{M}(k, \epsilon) = \Theta(\frac{k}{\epsilon})$ (Theorem \ref{thm:single-error}). We assume $\widetilde{M}(k, \epsilon) > k$ and the network gain can be infinite when $TSC(G, \pmb{v}, k, \epsilon) = nk$. 


\subsection{Overview of results}\label{sec:overview}

Our focus here is to 
estimate the $TSC(G, \pmb{v}, k, \epsilon)$, and the
network gain $\mu(G, \pmb{v}, k, \epsilon)$, and
characterize the distribution of the optimal $m_i$s across the
network.
Most proofs are presented in our full paper \cite{liu2023sample}.
\smallskip

\noindent \textbf{Tight bounds on TSC.}
Using the structure of the  Nash equilibrium of the opinion formation game (Lemma~\ref{lemma:potential}), and regression error bounds (Theorem~\ref{thm: one to all}), we derive tight bounds on $TSC(G, \pmb{v}, k, \epsilon)$ for any graph $G$ using a mathematical program (Theorem~\ref{thm:general_optimal}). 
We also show that the TSC can be estimated in polynomial time using second-order cone programming, allowing us to study it empirically.

\vspace{-2mm}
\begin{table}[htb!]
\centering
\begin{scriptsize}
\begin{tabular}{|p{1.0in}|p{1.5in}|}
\hline
\textbf{Graph class} & \textbf{Network gain}\\
\hline
Clique & $\Omega(n)$ \\
\hline
Star  & $O(1)$ \\
\hline
Hypercube &
$\Omega(d^2)$ \\
\hline
Random\ $d$-regular & $\Omega(\min\{d^2, n\})$ \\
\hline
$d$-Expander & $\Omega(\min\{d^2 \tau(G)^4, n\})$ \\
\hline
\end{tabular}    
\end{scriptsize}
\caption{ {\small Network gain $\ngain(G,\alpha, k, \epsilon)$ for different graphs where $\tau(G)$ denotes the
    edge expansion of $G$ (defined in Section~\ref{sec:network}).}
\label{tab:gain_short}
}
\end{table}


\vspace{-3mm}

\noindent \textbf{Impact of graph structure on TSC, network gain, and sample distribution.} 
We begin with the case of uniform influence factor $\alpha$.
We show that $\sum_{i=1}^n\frac{1}{(\alpha d_i + 1)^2} \lesssim
\frac{TSC(G,\alpha, k, \epsilon)}{\widetilde{M}(k, \epsilon)} \lesssim \sum_{i=1}^n\frac{\alpha + 1}{\alpha d_i + 1}$ (informal version of Theorem \ref{thm: degree bound}), where $d_i$ denotes the degree of agent $i$.  
Assigning
$\max\{k, \frac{\alpha + 1}{\alpha d_i + 1}\cd\frac{k}{\epsilon}\}$ samples to agent $i$ can ensure everyone
learns a good model.  
Thus, \emph{\textbf{It is sufficient to solve the TSC problem when the number of samples for an agent is inversely proportional to its degree}}.  
In other words, low-degree nodes need more
samples, in stark contrast to many network
mining problems, e.g., influence maximization, where it suffices to
choose high-degree nodes. Clearly, this result has policy
implications. Building on Theorem \ref{thm: degree bound}, we derive tight bounds on the network gain for different classes of graphs, summarized in Table~\ref{tab:gain_short}.


From Table~\ref{tab:gain_short}, we can see a well-connected network offers a substantial reduction in TSC (high network gain), whereas a star graph provides almost no gain compared to individual learning. We also demonstrate that the lower bounds on network gain in Table~\ref{tab:gain_short} are tight. More detailed results for specific networks can be found in Table~\ref{tab:graph-gain}.



Finally, we consider general influence factors and derive upper and lower bounds on TSC for arbitrary graphs (Theorem~\ref{thm: general tight bound}). We find that these bounds are empirically tight.

\noindent \textbf{Experimental evaluation.}  In our simulation experiments, we compute the total sample
complexity and the near-optimal solutions for a large number of
synthetic and real-world networks.  We begin with the case of uniform influence factors and first validate the findings of Theorem~\ref{thm: degree bound} that sample size in the optimal allocation has negative correlation with degree.  We then experimentally evaluate the $d^2$ network gains for random $d$-regular graphs. Overall, we find that experimentally computed solutions are consistent with our bounds in Table~\ref{tab:gain_short}.
Finally, we consider general influence factors and 
find that our theoretical bounds on TSC in Theorem~\ref{thm: general tight bound} are quite tight empirically.

\subsection{Related Work and Comparisons}
\textbf{Sample complexity of collaborative learning.}  Building on
\citet{blum:neurips17}, a series of
papers~\citep{chen2018tight,nguyen2018improved,blum2021one,haghtalab2022demand}
studied the minimum number of samples to ensure every agent has a
low-error model in collaborative learning. In this setting,
there is a center that can iteratively draw samples from different
mixtures of agents' data distributions, and dispatch the final model
to each agent. In contrast, we use the well-established decentralized
opinion formation framework to describe the model exchange game; the
final model of every agent is given by the equilibrium of this game.

Our formulation is similar to \cite{blum2021one}, which also considers
the sample complexity of equilibrium, ensuring every agent has a
sufficient utility guarantee. However, \cite{blum2021one} study this
problem in an incentive-aware setting without networks, which mainly
focuses on the stability and fairness of the equilibrium. In contrast,
our research is centered on the network effect of the equilibrium
generated by the opinion formation model. Moreover, \cite{blum2021one}
assumes the agents' utility has certain structures that are not
derived from error bounds while we directly minimize the
generalization error of agents' final models.

\cite{haddadan2024optimally} is another related paper which also considers learning on networks under opinion dynamic process. However, they study selecting $K$ agents to correct their prediction to maximize the overall accuracy in the entire network, rather than the sample complexity bound to ensure individual learner has a good model as our paper.

\noindent \textbf{Fully decentralized federated learning. } To reduce
the communication cost in standard federated learning
\citep{mcmahan2017communication}, \citet{lalitha2018fully,
  lalitha2019peer} first studied fully decentralized federated
learning on networks, where they use Bayesian approach to model
agents' belief and give an algorithm that enables agents to learn a
good model by belief exchange with neighbors. This setting can be
regarded as a combination of Bayesian opinion formation models
\citep{banerjee1992simple, bikhchandani1992theory,
  smith2000pathological, acemoglu2006learning} and federated
learning. In the literature regarding opinion formation on networks,
besides those Bayesian models, non-Bayesian models are usually
considered more flexible and practical \citep{degroot1974reaching,
  friedkin1990social, demarzo2003persuasion, golub2008homophily,
  bindel2015bad, haddadan2024optimally}.  Comprehensive surveys of these two kinds of models
can be found in \citet{jackson2010social} and
\citet{acemoglu2011opinion}. Our study makes connections between
non-Bayesian opinion formation models and federated learning for the
first time. Compared with \citet{lalitha2018fully, lalitha2019peer},
we assume each agent can observe the model of neighbors, rather than a
belief function. We do not restrict to specific algorithms but use
game theoretical approaches to find the unique Nash equilibrium and
analyze sample complexity at this equilibrium.

%% file: New_Formulation.tex
\section{The Opinion Formation Game} \label{sec:problem formulation}


As mentioned earlier, we utilize a
variation on the seminal DeGroot model \cite{degroot1974reaching}
proposed by \cite{friedkin1990social}, also studied by~\cite{bindel2015bad}.
Formally, agent $i$ seeks $\theta_i$
which minimizes the loss
\begin{align*}
\|\theta_i - \Bar{\theta}_i\|^2 + \sum_{j\in N(i)} v_{ij} \|\theta_i - \theta_j\|^2
\end{align*}
where $v_{ij} \ge 0,\,\,\forall i,j \in [n], j \in N(i)$ measures the influence of agent $j$ to agent $i$. 
In general, $v_{ij}$'s might not be known, and we will also study the simpler uniform case where $v_{ij} = \alpha \ge 0\,\, \forall
i,j$.
Lemma \ref{lemma:potential} gives the unique equilibrium of this
game (also studied by \citet{bindel2015bad}); its proof is presented in Appendix \ref{app:general setting} of \cite{liu2023sample}.


\begin{lemma}[\textbf{Nash equilibrium of opinion formation}]\label{lemma:potential}
The unique Nash equilibrium $\pmb{\theta^{eq}} = (\theta_1^{eq}, \cdots, \theta_n^{eq})^T$ of the above game is $\pmb{\theta^{eq}} = W^{-1}\pmb{\Bar{\theta}}$ where $W_{ij} = \left\{ 
\begin{array}{ll}
       \sum_{j\in N(i)}v_{ij} + 1 & j=i \\
       -v_{ij}  & j \in N(i)\\
       0  & j \notin N(i), j \neq i
\end{array}
\right.$ and $\pmb{\Bar{\theta}} = (\Bar{\theta}_1, \cdots, \Bar{\theta}_n)$. When all $v_{ij} = \alpha \ge 0$, $W_{ij} = \left\{
\begin{array}{ll}
       \alpha D + 1 & j=i \\
       -\alpha  & j \in N(i)\\
       0  & j\notin N(i), j\neq i
\end{array}\right.$. Furthermore, we have $\sum_{j=1}^n W^{-1}_{ij} = 1$ and $W^{-1}_{ij} \ge 0$ for all $i,j \in [n]$.
\end{lemma}

Thus, from Lemma \ref{lemma:potential}, the equilibrium model of every agent is a convex combination of all the agents' internal opinion on the network.
We show that $\pmb{\theta^{eq}}$ can be computed in a federated manner using the algorithm of~\cite{DBLP:conf/aistats/VanhaesebrouckB17}.

\section{Total Sample Complexity of Opinion Formation}
Recall that each agent $i$'s dataset $S_i = \{(x_{i}^j, y_{i}^j)\}_{j
  \in [m_i]}$ has $m_i \ge k$ samples, where $x_i^j \in \mathbb{R}^k$. We assume there is a  ground-truth
model $\theta^*$ such that for any $x_i^j \sim D$, $y_i^j = (x_i^j)^\top \theta^* + \eta_i^j$ where $\eta_i^j \sim  \eta_i(x_i^j)$ and $\eta_i: \mathbb{R} \rightarrow \Delta(\mathbb{R})$ is an agent-dependent noise function, mapping samples to a noise distribution.  We consider unbiased noise with bounded variance, implying that for every agent $i$,  noise function $\eta_i$ is independently chosen from  $\mathcal{N} = \{\eta:
\E\left[\eta(x)\right] = 0, \Var\left[\eta(x)\right] \le \sigma^2, \forall x \sim D\}$.  Let
$X_i = [x_i^1, \cdots, x_i^{m_i}]$ and $Y_i = [y_i^1,\cdots,
  y_i^{m_i}]^\top$.  Every agent performs ordinary linear regression
to output their initial opinion $\Bar{\theta}_i =
\argmin_{\theta}\sum_{j=1}^{m_i}\left((x_i^{j})^\top \theta -
y_i^j\right)^2$.  For our study, we first make standard Assumption \ref{assum: non-degenerate} on data distributions.


\begin{assumption}[non-degeneracy]
For data distribution $D$ over $\mathbb{R}^k$, if $x$ is drawn from $D$, then for any linear hyperplane $H \subset \mathbb{R}^k$, we have $\mathbb{P}\left(x \in H\right) = 0$.
\label{assum: non-degenerate}
\end{assumption}

Assumption \ref{assum: non-degenerate} is standard to ensure the data
distributions span over the whole $\mathbb{R}^k$ space. If it holds,
from Fact 1 in \cite{mourtada2022exact}, for any $X_i$ with $m_i \ge
k$, $X_i^\top X_i$ is invertible almost surely. This implies the
ordinary linear regression for every agent to form the initial opinion
enjoys the closed-form solution $\Bar{\theta}_i = \left(X_i^\top
X_i\right)^{-1}X_i^\top Y_i$.

\noindent \textbf{Loss measure.} We use the expected square 
loss in the worst case of noise to measure the quality of agents' final opinions $(\theta_1^{eq}, \cdots, \theta_n^{eq})$ at the equilibrium. Namely, for agent $i$, we
consider the loss
\begin{align}
L(\theta_i^{eq}) = \sup_{\eta_{1:n} \in \mathcal{N} }\E_{x \sim D, \forall j, S_j \sim D^{m_j}(\eta_j)}\left[\left(x^\top \theta_i^{eq} - x^\top \theta^* \right)^2\right]
\label{eqn:expected loss}
\end{align}
where $D^{m_i}(\eta_i)$ denotes the joint distribution of  $S_i = \{(x_{i}^j, y_{i}^j)\}_{j \in
  [m_i]}$ given noise function $\eta_i$. We take supremum over all possible noises and take expectation over all dataset because $\theta_i^{eq}$ is related to $\Bar{\theta}_i$ for every agent $i$.



\subsection{Derivation of error bounds} \label{sec:error bounds}
We first define the error for initial opinion as 
\begin{equation*}
   L(\Bar{\theta}_i) = \sup_{\eta_i \in \mathcal{N}}\E_{x \sim D, S_i \sim D^{m_i}(\eta_i)}\left[\left(x^\top \Bar{\theta}_i - x^\top \theta^* \right)^2\right].
\end{equation*}
To quantify the upper bound of $ L(\Bar{\theta}_i) $, we additionally consider the following assumption on data distribution.
\begin{assumption}[small-ball condition] For data distribution $D$,  there exists constant $C_i \ge 1$ and $\alpha_i \in (0,1]$ such that for every hyperplane $H \subset \mathbb{R}^k$ and $t > 0$,  if $x$ is drawn from $D$,  we have $\mathbb{P}\left(\dist (\Sigma_i^{-\frac{1}{2}}x, H) \le t\right) \le (C_it)^{\alpha_i}$  where $\Sigma_i = \mathbb{E}_{x \sim D}\left[x x^\top\right]$.
\label{assum: small-ball}
\end{assumption}

Given Assumption \ref{assum: non-degenerate}, $\Sigma_i$ is guaranteed
to be invertible. Assumption \ref{assum: small-ball} ensures
$\Sigma_i^{-\frac{1}{2}}x$ is not too close to any fixed hyperplane,
which is introduced in \cite{mourtada2022exact} and is a variant of
the small-ball condition in \cite{koltchinskii2015bounding,
  mendelson2015learning, lecue2016performance}. From Proposition 5 in
\cite{mourtada2022exact} (see also Theorem 1.2 in
\cite{rudelson2015small}), if every coordinate of $D$ are
independent and have bounded density ratio, then Assumption
\ref{assum: small-ball} holds. More discussions on this assumption
could be found in Section 3.3 in \cite{mourtada2022exact}.

\junk{
\begin{assumption}[zero-mean Gaussian] Every data distribution $D, \,\, \forall i \in [n]$ is a $k$-dimensional multivariate Gaussian distribution with zero
mean.
\label{assum: gau}
\end{assumption}
}




\begin{theorem} [Theorem 1, Proposition 2 and Equation 17 in \cite{mourtada2022exact}]
For every agent $i$ with $m_i$ samples, ordinary linear regression attains loss $L(\Bar{\theta}_i) = \Theta\left(\frac{k}{m_i}\right)$.
\label{thm:single-error}
\end{theorem}

\junk{
\begin{theorem} [Theorem 1, Proposition 2 and Equation 17 in \cite{mourtada2022exact}]
Given Assumption \ref{assum: non-degenerate} holds, either Assumption \ref{assum: small-ball} or  Assumption \ref{assum: gau} could ensure
    \begin{align*}
        L(\Bar{\theta}_i) = \Theta\left(\frac{k}{m_i}\right)
    \end{align*}
where $\Theta$ hides constant factors that are independent of $k$ and $m_i$.
\label{thm:single-error}
\end{theorem}
}

For certain data distributions, tighter
closed-form bounds have been derived.  For instance, if $D$ is a
$k$-dimensional multivariate Gaussian distribution with zero mean,
then $L(\Bar{\theta}_i) = \frac{\sigma^2 k}{m_i-k-1}$~\cite{anderson1958introduction, breiman1983many,
  donahue2021model}.  For mean estimation,  $L(\Bar{\theta}_i) = \frac{\sigma^2}{m_i}$~\cite{donahue2021model}. Theorem \ref{thm:single-error} shows that
$L(\Bar{\theta}_i)$ scales with $\frac{k}{m_i}$, assuming the data
distributions satisfy non-degeneracy and the small-ball
condition. Thus, $\widetilde{M}(k,\epsilon) = \Theta(\frac{k}{\epsilon})$ samples suffice for a
single agent learning a model with error $\epsilon$.

\junk{
The error bound based on
Assumption \ref{assum: gau} also appears in
\cite{anderson1958introduction, breiman1983many,
  donahue2021model}. The subsequent discussions in the paper are all
grounded in Theorem \ref{thm:single-error}. Therefore, all the
following lemmas and theorems implicitly assume the condition in
Theorem \ref{thm:single-error} holds.
}

 
We next present our main technique for bounding the generalization
error of an equilibrium model, which uses
Theorem~\ref{thm:single-error} to express the error as a function of
the matrix $W$ and the number of samples at each agent.

\begin{theorem}[\textbf{Bound on generalization error}]
    For every agent $i$, we have $ L(\theta_i^{eq}) = \Theta\left( k \sum_{j=1}^n \frac{\left(W_{ij}^{-1}\right)^2}{m_j} \right)$
    where $\theta_i^{eq}$ and matrix $W$ is defined in Lemma \ref{lemma:potential}.
\label{thm: one to all}
\end{theorem}

\noindent \textbf{Remark. } In Appendix \ref{app:general setting} of \cite{liu2023sample}, we show that Theorem \ref{thm: one to all} also holds for generalized linear regression with adapted assumptions, where a mapping function $\phi$ exists such that $\mathbb{E}[y] = \phi(x)^\top \theta^*$ for any possible data $(x,y)$.

\subsection{Total sample complexity}
\label{sec:TSC}
Armed with Theorem~\ref{thm: one to all}, we initiate the study of
total sample complexity of opinion convergence. Recall that we want to
ensure that at the equilibrium, every node on the network has a model
with a small generalization error. Specifically, we want to ensure
$L(\theta_i^{eq}) \le \epsilon$ for any $i \in [n]$ and a given
$\epsilon > 0$.  Theorem~\ref{thm: one to all} precisely gives us the
mechanism to achieve a desired error bound.


\junk{From Theorem \ref{thm:single-error}, using
$\Theta(\frac{k}{\epsilon})$ data samples is enough for an individual
agent to ensure its local error $L(\Bar{\theta}_i) \le
\epsilon$.}

\noindent
\textbf{The Total Sample Complexity (TSC) problem.}
Recall that for a given graph $G$,  influence factor $\pmb{v}$, dimension $k$ and error $\epsilon$, the
total sample complexity $TSC(G, \pmb{v}, k, \epsilon)$ is the minimum value of $\sum_i m_i$, under the constraint $m_i \ge k$, $m_i \in \mathbb{Z}^+$ and $L(\theta_i^{eq}) \le \epsilon$ for every $i \in [n]$. Our central result,
Theorem \ref{thm:general_optimal}, derives near-tight bounds on the
total sample complexity.


\begin{theorem}[\textbf{Bounds on $TSC$}]
  \label{thm:general_optimal}
For any $\epsilon >  0$, let $(m_i^*, i=1,\ldots,n)$ denote an optimal solution of the following optimization problem as a measure of the minimum samples for opinion formation on graph $G$ with influence factor $\pmb{v}$.
\begin{align}
    \min_{m_1, \cdots, m_n} \quad &\sum_{i=1}^n m_i \nonumber\\
    \textrm{s.t.} \quad & \sum_{j=1}^n \frac{(W_{ij}^{-1})^2}{m_j} \le \frac{\epsilon}{k}, \,\, \forall i \label{eq:opt} \\
    &m_i > 0, \,\, \forall i \nonumber
\end{align}
where $W$ is defined in Lemma
\ref{lemma:potential}.  Then, $TSC(G, \pmb{v}, k, \epsilon) = \Theta(\sum_{i=1}^n  m_i^* + nk)$. Assigning $\max\left\{O(\lceil m_i^* \rceil) , k\right\}$ samples to agent $i$ suffices for $L(\theta_i^{eq}) \le \epsilon$ for every $i \in [n]$. Moreover, $\frac{\epsilon m_i^*}{k}$ is a fixed value for any $k$ and $\epsilon$.
\end{theorem}

From Theorem~\ref{thm:general_optimal}, $TSC(G, \pmb{v}, k, \epsilon)$ has order $\sum_{i=1}^n m_i^* + nk$ and assigning $\max\left\{\lceil m_i^* \rceil , k\right\}$ samples to agent $i$ is sufficient to solve the TSC problem up to some constant. Thus, we only need to focus on the solution of Equation \ref{eq:opt}, $m_i^*, \forall i \in [n]$. In the following sections, we will characterize properties of $m_i^*, \forall i \in [n]$ and use it to prove network gain for different graphs.


\junk{Note that every $m_i$ should be an integer that larger than $k$
  by the definition of TSC, but the optimization in Theorem
  \ref{thm:general_optimal} considers a continuous relaxation where
  $m_i > 0$. However, we show that our solution gives an $nk$-additive
  approximation to TSC in general, which is a fair approximation since
  TSC is defined with $m_i \ge k$. The additional $nk$ factor only
  introduce $k$ more samples to every agent, leading to a constant
  multiplicative approximation.}

%% file: Network_Affect.tex
\section{Network Effects on Total Sample Complexity}
\label{sec:network}
In this section, we analyze the network effects on the total sample complexity (TSC) of opinion convergence. We derive bounds on the network gain, showing how the number of samples needed decreases due to network learning. Starting with uniform influence weights, we obtain asymptotically tight bounds for key network classes and characterize how samples should be distributed by node degree to minimize sample complexity. We then extend to arbitrary influence weights and derived bounds for TSC that are empirically tight.

\subsection{Uniform influence factors} \label{sec:special}
\junk{Although the bounds in Theorem \ref{thm: general tight bound} are tight
empirically, they are unintuitive and do not guarantee theoretical
tightness. To solve this issue, we further investigate the natural
case where the influence weights are uniform.}

We model uniform influence factors by setting $v_{ij} = \alpha$ for
every agent $i$ and neighbor $j$ of $i$, for a given real $\alpha \ge
0$.  Let $TSC(G,\alpha, k, \epsilon)$ be the solution of the optimization
in Theorem \ref{thm:general_optimal} for this case. Theorem \ref{thm: degree bound}
provides interpretable bounds for $TSC(G,\alpha, k, \epsilon)$ related to
degree, and serves as the first step to derive tight bounds for specific graph classes. 
The proof of Theorem \ref{thm: degree bound} leverages the dual form of Equation \ref{eq:opt}, together with a careful analysis of matrix $W$ defined in Lemma \ref{lemma:potential}.
Detailed proof can be found in \cite{liu2023sample}. 

\begin{theorem}[\textbf{Sample allocation and degree distribution}] \label{thm: degree bound}
The optimal solution $\{m^*_i\}$ to Equation~\ref{eq:opt} satisfies
\begin{eqnarray*}
\max\left\{\sum_{i=1}^n \frac{1}{(\alpha d_i + 1)^2}, 1\right\} \le \sum_{i \in [n]} m^*_i\cd\frac{\epsilon}{k} \le \sum_{i=1}^n \frac{\alpha + 1}{\alpha d_i + 1}.
\end{eqnarray*}
We have $TSC(G, \alpha, k, \epsilon) = \Theta(\sum_{i \in [n]} m^*_i + nk)$.
Assigning $\max\{ O \left(\lceil \frac{\alpha + 1}{\alpha d_i + 1}\cd\frac{k}{\epsilon} \rceil\right), k\}$ samples to every agent $i$ suffices for $L(\theta^{eq}_i) \le \epsilon,\,\,\forall i \in [n]$.

\end{theorem}
 Theorem \ref{thm: degree
  bound} suggests that it is sufficient to allocate samples
\emph{inversely proportional} to the node degrees. This has
interesting policy implications; in contrast to other social network
models, e.g., the classic work on influence maximization \cite{KKT03},
our result advocates that allocating more resources to
low-connectivity agents benefit the network at large. We provide
empirical validation of this phenomenon in Section \ref{sec:experiment}.

With the help of Theorem \ref{thm: degree bound}, we could derive the bounds of network gain for any graph in Corollary \ref{cor:gain}. 
\begin{corollary}
\!For any graph $G$,\,if $\epsilon \le O(\frac{1}{\alpha \max_i d_i + 1})$, then $\ngain(G,\alpha, k, \epsilon) \ge \Omega\left(\frac{n}{\sum_{i=1}^n \frac{\alpha + 1}{\alpha d_i + 1}}\right)$. If $\epsilon \le O(\frac{1}{ \max_i(\alpha d_i + 1)^2})$, then $\ngain(G,\alpha, k, \epsilon) \le O\left(\min\left\{\frac{n}{\sum_{i \in [n]} \frac{1}{(\alpha d_i + 1)^2}}, n\right\}\right)$.
\label{cor:gain}
\end{corollary}

For small $\epsilon$, Corollary \ref{cor:gain} demonstrates that the network gain is at least $\Omega\left(\frac{n}{\sum_{i \in [n]} \frac{1}{d_i}}\right)$, implying that networks with more high-degree nodes result in higher network gain, but the gain is limited to $O\left(\min\left\{\frac{n}{\sum_{i \in [n]} \frac{1}{d_i^2}}, n\right\}\right)$.

Now we are ready to analyze the tight network gain for special graphs classes. Our main results are shown in Table \ref{tab:graph-gain}, with detailed lemmas and proofs deferred to \cite{liu2023sample}. The third column of Table \ref{tab:graph-gain} gives $\sum_{i \in [n]} m_i^*$, which suffices to characterize total sample complexity (TSC) because $TSC(G, \alpha, k, \epsilon) = \Theta\left(\sum_{i \in [n]} m_i^* + nk\right)$ from Theorem \ref{thm:general_optimal}. The last column indicates the number of samples required for each agent to ensure that all constraints of the TSC problem are satisfied. For each kind of graph, if the sample size for every agent is the maximum of $k$ and the term in the last column, it is sufficient to guarantee $L(\theta^{eq}_i) \le \epsilon,\,\,\forall i$. To prove the results in Table \ref{tab:graph-gain}, we leverage the spectral properties of graph Laplacian $\mathcal{L}$ for different networks, together with the lower bound in Theorem \ref{thm: degree bound}. Note that the lower bound in Table \ref{tab:graph-gain} does not require any constraint on $\epsilon$. This contrasts with the bounds in Corrollary \ref{cor:gain} and requires a more refined analysis. On the other hand, there is no general upper bound for the network gain because it can be infinite.


We now give a more detailed discussion of the results in Table \ref{tab:graph-gain} as follows. \textbf{(1)} For a clique, the network gain is $\Omega(n)$ and $TSC(G, \alpha, k, \epsilon) = \Theta\left(\frac{k}{\epsilon} + nk\right)$. From Corollary \ref{cor:gain}, this is the best lower bound of network gain for any graph. \textbf{(2)} For a star, when $\epsilon$ is less than some constant, the network gain is $O(1)$ and $TSC(G, \alpha, k, \epsilon) = \Theta\left(\frac{nk}{\epsilon}\right)$. This implies a star graph almost cannot get any gain compared with learning individually. \textbf{(3)} For a hypercube, when $\alpha \ge \frac{3}{8}$, the network gain is $\Omega(d^2)$ and $TSC(G, \alpha, k, \epsilon) = \Theta\left(\frac{nk}{d^2\epsilon} + nk\right)$. \textbf{(4)} For a random d-regular graph, with high probability,  the network gain is $\Omega(d^2)$ when $d \le \sqrt{n}$ while it is $\Omega(n)$ when $d > \sqrt{n}$. The total sample complexity is $TSC(G, \alpha, k, \epsilon) = \Theta\left(\frac{nk}{\min\{d^2, n\}\epsilon} + nk\right)$. \textbf{(5)} For a $d$-regular graph $G$, the edge expansion $\tau(G)$ equals $\min
\limits_{|S|: |S|\leq n/2} \frac{|\partial S|}{d|S|}$, where $\partial
S$ is the set $\{(u,v) \in E: u \in S, v \in V \setminus S\}$.  Edge expansion measures graph connectivity and can be viewed as a lower bound on the probability that a randomly chosen edge from any subset of vertices $S$ has one endpoint outside $S$. A $d$-expander denotes a $d$-regular graph with edge expansion $\tau(G)$. For this kind of graph, the network gain is $\Omega(\min\{d^2 \tau(G)^4, n\})$ and  $TSC(G, \alpha, k, \epsilon) = \Theta(\frac{nk}{\min\{d^2 \tau(G)^4, n\}\epsilon} + nk)$. Note that the $\Omega\left(\min\{d^2, n\}\right)$ lower bound is also optimal for these $d-$regular graphs from the upper bound in Corollary \ref{cor:gain}.

\vspace{-2.5mm}
\begin{table}[htb!]
\centering
\begin{scriptsize}
\begin{tabular}{|p{0.5in}|p{0.75in}|p{0.75in}|p{0.75in}|}
\hline
\makecell[c]{\textbf{Graph class}}& \makecell[c]{\textbf{Network gain}} & \makecell[c]{$\sum_{i \in [n]} m_i^*$}& \textbf{Samples per node} (omit $\max$ with $k$)\\
\hline
\makecell[c]{Clique \\ (Lemma \ref{lem: clique}) } & \makecell[c]{$\Omega(n)$} & \makecell[c]{$\Theta(\frac{k}{\epsilon})$} &\makecell[c]{$O(\frac{k}{\epsilon n})$} \\
\hline
\makecell[c]{Star \\ (Lemma \ref{lem: star})} & {\makecell[c]{$O(1)$ \\  \tiny when $\epsilon \le O(1)$}} & \makecell[c]{$\Theta(\frac{nk}{\epsilon})$} & \makecell[c]{$O(\frac{k}{\epsilon})$\quad  (leaf) \\ $O(\frac{k}{n\epsilon})$ \quad (center)} \\
\hline
\makecell[c]{Hypercube \\ (Lemma \ref{lem:hypercube})} &
\makecell[c]{$\Omega(d^2)$ \\ \tiny when $\alpha \ge \frac{3}{8}$} & \makecell[c]{$\Theta(\frac{nk}{d^2\epsilon})$} &\makecell[c]{$O(\frac{k}{d^2\epsilon})$}\\
\hline
\makecell[c]{Random \\ d-regular\\ (Lemma \ref{lem:regular})} & \makecell[c]{$\Omega(\min\{n, d^2\})$ \\ \tiny w.h.p.} & \makecell[c]{$\Theta(\frac{nk}{\min\{n,d^2\}\epsilon})$} & \makecell[c]{$O(\frac{k}{\min\{n,d^2\}\epsilon})$}\\
\hline
\makecell[c]{d-Expander \\ (Lemma \ref{lem:expander})}& $\Omega(\min\{n, d^2 \tau^4\})$ & $\Theta(\frac{nk}{\min\{n, d^2 \tau^4\}\epsilon})$ &$O(\frac{k}{\min\{n,d^2 \tau^4\}\epsilon})$\\
\hline
\end{tabular}    
\end{scriptsize}
\caption{ {\small Network gain $\ngain(G,\alpha, \epsilon)$, $\sum_{i \in [n]} m_i^*$ and
    distribution of samples in different network classes and constant $\alpha$. The number of samples at a node is the maximum of $k$ and the term in the
    third column. In the last row, $\tau = \tau(G)$ is the edge expansion.}
\label{tab:graph-gain}
}
\end{table}

\vspace{-2mm}
It follows from Table \ref{tab:graph-gain} and Corollary \ref{cor:gain} that
both random and expander $d$-regular graphs (with constant $\tau(G)$)
have the optimal lower bound for network gain. Interestingly, this property also
holds for the hypercube, which has regular degree $d = \log n$, even though its expansion is
$O(\frac{1}{\log n})=o(1)$.  A natural open question is to
characterize other degree-bounded network families which also achieve
near-optimal network gains.

\junk{
Thus for regular graphs with constant expansion,

have that
$TSC=\Theta(\frac{nM(\epsilon)}{\alpha^2d^2})$ for $d <
O(\frac{\sqrt{n}}{\alpha})$ and $TSC = \Theta(M(\epsilon))$ for $d >
\Omega(\frac{\sqrt{n}}{\alpha})$.  The bounds given in
Lemma \ref{lem:regular} and Lemma \ref{lem:expander} suggest that graph
expansion is inversely related to the total sample
complexity. However, note that the hypercube has expansion
 but admits a tight sample complexity of
$\Theta(\frac{nM(\epsilon)}{d^2})$, in contrast to the upper bound for
general expanders. }

\subsection{General influence factors} \label{sec:general}

Going beyond uniform influence factors, we investigate the total sample
complexity of opinion convergence with general
influence factors.  The non-uniformity of the influence factors
implies that the sample complexity is not just dependent on the
network structure, but also on how each individual agent weighs the
influence of its neighbors.  A trivial bound is $\Omega\left(\frac{k}{\epsilon}\right) \le
TSC(G, \pmb{v}, k, \epsilon) \le O\left(\frac{nk}{\epsilon}\right)$ which means learning
through the game is always more beneficial than learning individually
but needs at least the samples for one agent to learn a good model.



To derive more sophisticated and tighter bounds, we need to analyze
the matrix $W$ with general weights, which captures both the network topology and influence
factors. Since $W^{-1}$ is positive-definite, from the Schur product theorem,
$(W^{-1}) \circ (W^{-1})$ is also positive-definite where $\circ$ is
the Hadamard product (element-wise product). Define $B = ((W^{-1})
\circ (W^{-1}))^{-1}$ and $[B]_{ij} = b_{ij}$. Theorem \ref{thm: general
  tight bound} establishes a bound for graphs with general influence
factors, which is empirically tight from our experiments in Section \ref{sec:experiment}. The proof for the upper bound utilizes a thorough analysis of the property of $b_{ij}$s, and the lower bound is derived based on the dual form of Equation \ref{eq:opt} with general weights. We refer the reader to \cite{liu2023sample} for the whole proof.

\begin{theorem}[\textbf{TSC under general influence factors}]
  \label{thm: general tight bound}
Let $\gamma_i = \max\{0, \sum_{j=1}^n \frac{b_{ij}}{(\sum_{k=1}^n b_{jk})^2}\}$ for every $i \in [n]$, then the optimal solution $\{m^*_i\}$ to Equation~\ref{eq:opt} satisfies
\begin{align*}
\sum_{i=1}^n \left(2\sqrt{\sum_{j=1}^n \gamma_j (W^{-1}_{ij})^2} - \gamma_i\right) \le \sum_{i} m_i^*\cd\frac{\epsilon}{k} \le \sum_{i=1}^n \frac{1}{\sum_{j=1}^n b_{ij}}   
\end{align*}
and $TSC(G, \pmb{v}, k, \epsilon)$ is $\Theta(\sum_{i} m_i^* + nk)$.
Assigning 
$\max\left\{k, O\left(\lceil \frac{k}{\epsilon\sum_{j=1}^n b_{ij}}\rceil \right)\right\}$ samples to agent $i$ suffices for $L(\theta_i^{eq}) \le \epsilon$ for every $i \in [n]$.
\end{theorem}

Although the above two bounds give good empirical estimations to the TSC, deriving closed-form solution for Equation \ref{eq:opt} is still interesting, and is investigated in Lemma \ref{lem:closed-form}  under a certain condition.
\begin{lemma}\label{lem:closed-form}
If $\sum_{j=1}^n \frac{b_{ij}}{(\sum_{k=1}^n b_{jk})^2} \ge 0$ for all $i \in [n]$,  $ \sum_{i \in [n]} m_i^* = \sum_{i=1}^n\frac{1}{\sum_{j=1}^n b_{ij}}\cd\frac{k}{\epsilon}$ where the optimal $m_i^* = \frac{1}{\sum_{j=1}^n b_{ij}}\cd\frac{k}{\epsilon}$.
\end{lemma}
\HL{I add this lemma here because we have a conjecture based on this in experiments (see below). But I remove the solution of dual variables, only leaving $m_i^*$.}

It is easy to verify that the condition in Lemma \ref{lem:closed-form} holds for cliques, 
but beyond cliques, we have yet to identify other instances where it holds. Characterizing the graphs that meet this condition remains an interesting open problem



%% file: Experiments.tex
\section{Experiments}
\label{sec:experiment}


We use experiments on both synthetic and real-world networks to further understand the distribution of samples (Theorem \ref{thm: degree bound}), the quality of our bounds for general influence weights (Theorem \ref{thm: general tight bound}), and the network gains for $d$-regular graph (Table \ref{tab:graph-gain} and Lemma \ref{lem:regular}). We observe good agreement with our theoretical bounds. We only list part of results here  and more experiments and be found in  \cite{liu2023sample}.

From Theorem \ref{thm:general_optimal}, given a network and influence factors, $\frac{\epsilon m_i^*}{k}$ is fixed for any $k$ and $\epsilon$. This value can be solved directly by reformulating Equation \ref{eq:opt}. Since $\widetilde{M}(k,\epsilon) = \Theta(\frac{k}{\epsilon})$,  $\frac{\epsilon m_i^*}{k}$ indicates the percentage of samples one agent needs at the equilibrium compared to the number of samples required to learn independently, which is the main factor we are interested in. Thus, we only consider $\frac{\epsilon m_i^*}{k}$ for experiments of Theorem \ref{thm: degree bound} and Theorem \ref{thm: general tight bound} without specific $\epsilon$ and $k$.


To solve $m_i^\star$ or $\frac{\epsilon m_i^*}{k}$ for every $i \in [n]$ exactly, we reformulate Equation \ref{eq:opt} to a second-order cone programming (SOCP) and then use the solver CVXOPT (\cite{andersen2013cvxopt}) to solve it. More details are in \cite{liu2023sample}.  We now describe the networks used in our experiments.


\emph{\textbf{Synthetic networks.}} We use three types of synthetic networks: 
scale-free (\textbf{SF}) networks~\cite{barabasi1999emergence}, random d-regular graphs (\textbf{RR}), and Erd\"os-Renyi random graphs (\textbf{ER}). The exact parameters for these graphs used in different experiments are given in \cite{liu2023sample}.

\emph{\textbf{Real-world networks. }} The networks we use (labeled as \textbf{RN}) are shown in Table \ref{real_network_table} with their features ($err_u/err_l$ will be defined later).
We use the 130bit network, the econ-mahindas network (\cite{nr}),  the ego-Facebook network (\cite{leskovec2012learning}), and the email-Eu-core temporal network (\cite{paranjape2017motifs}).  The last two networks are also in \citet{snapnets}.

\vspace{-2mm}
\begin{table}[htbp]
\begin{footnotesize}
\centering
\begin{tabular}{| p{6em} | p{3em} | p{3em}| p{3em} | p{3em}|}
    \hline
    \textbf{Network} & \textbf{Nodes} & \textbf{Edges}  & \textbf{$err_u$} & \textbf{$err_l$}  \\ \hline
    ego-Facebook & 4039 & 88234 & 27\% & 24\%\\ \hline
    Econ & 1258 & 7620 & 6.6\% & 1.3\%\\ \hline
    Email-Eu & 986 & 16064 &  30\% & 21\% \\ \hline
    130bit & 584 & 6067 & 41\% & 36\%\\ \hline
\end{tabular} 
\caption{Real-world networks and bound performance}
\label{real_network_table}
\end{footnotesize}
\end{table}


\vspace{-4mm}

Our results are summarized below.

\textbf{Distribution of samples (Theorem \ref{thm: degree bound}). } We set the uniform influence factor $\alpha  = 0.1$, showing how the sample assignment is related to degrees at the solution of Equation \ref{eq:opt} on different kinds of networks: scale-free (Figures \ref{PL_AVG} and \ref{PL_VAR}), Erdos-Renyi random networks (Figures \ref{ER_AVG} and \ref{ER_VAR}), and real-world networks (Figures \ref{Real_AVG}  and  \ref{Real_VAR}).   Figures \ref{PL_AVG}, \ref{ER_AVG}, \ref{Real_AVG} show the average number of samples (divided by $\frac{k}{\epsilon}$) for each degree;  Figure \ref{PL_VAR}, \ref{ER_VAR} and \ref{Real_VAR} show the variance of sample number for each degree.  

Our main observations are: 
(1) The samples assigned to nodes with the same degree tend to be almost the same (i.e., the variance is small), 
(2) Fewer samples tend to be assigned to high-degree agents.
These observations are consistent with Theorem \ref{thm: degree bound}, which states that sample assignment has negative correlation with degree.

\textbf{Tightness of bounds (Theorem \ref{thm: general tight bound}). }
We empirically show the performance of the bounds in Theorem \ref{thm: general tight bound}. 
Here, all influence factors $v_{ij}$s are generated randomly from $[0,1]$.
For each kind of synthetic network, we generate 150 different instances with different number of nodes $n$ and $v_{ij}$s. We meausure the performance of bounds in Theorem \ref{thm: general tight bound} by relative error (i.e $\frac{|U \text{ or } L - \sum_{i \in [n]} m_i^*|}{\sum_{i \in [n]} m_i^*}$ where $U$ is the upper bound in Theorem \ref{thm: general tight bound} and $L$ is the max of $1$ (trivial lower bound) and the lower bound in Theorem \ref{thm: general tight bound}).
We visualize the relative error through frequency distributions of generated networks. Our observation is that the bounds are very tight for scale-free graphs (Figure \ref{T_PL}),  Erdos-Renyi random graphs (Figure~\ref{T_ER}), and random d-regular graphs (Figure~\ref{T_RR}).

\vspace{-3mm}

\begin{figure}[H]
\begin{minipage}[t]{\linewidth}
    \subfigure[SF Average]{ \label{PL_AVG} \includegraphics[width = 0.48\linewidth]{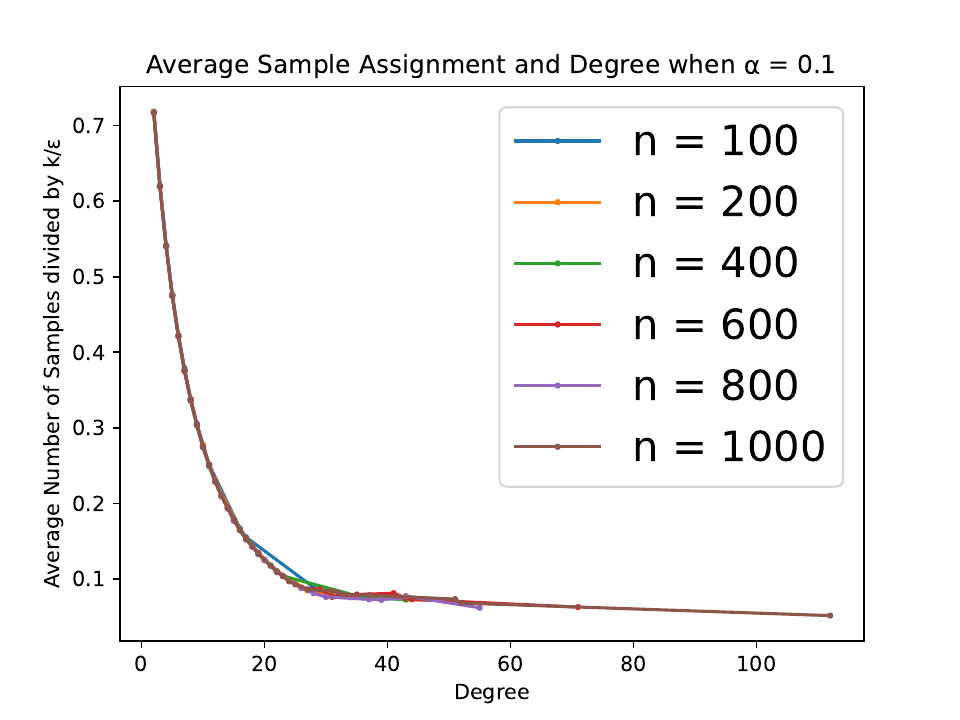}}
    \subfigure[SF Variance]{ \label{PL_VAR} \includegraphics[width = 0.48\linewidth]{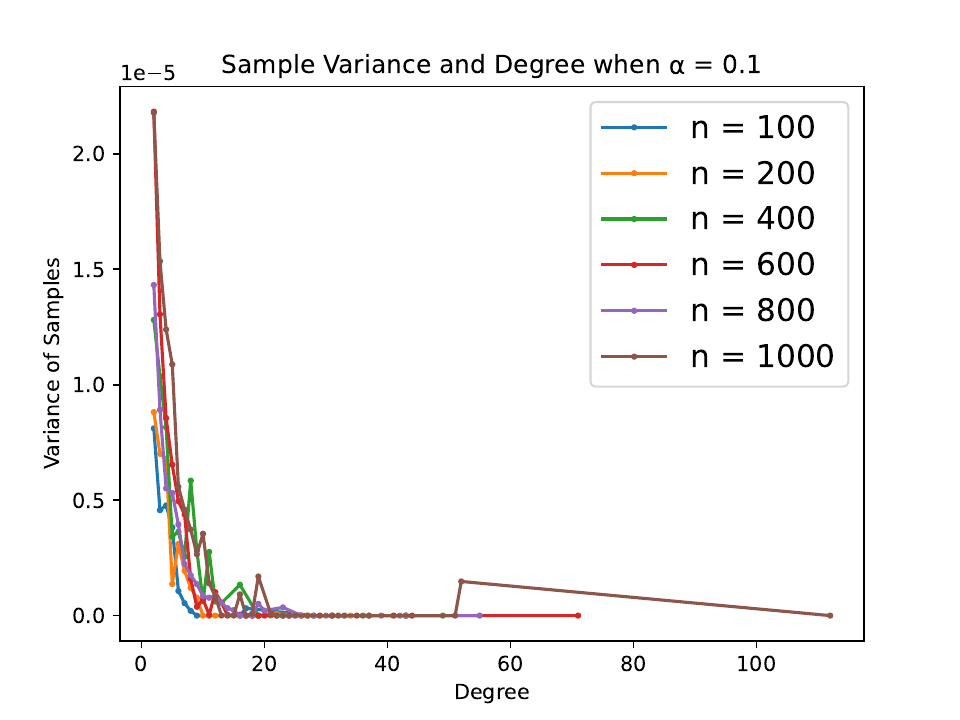}}
\end{minipage}

\begin{minipage}[t]{\linewidth}
    \subfigure[ER Average]{ \label{ER_AVG} \includegraphics[width = 0.48\linewidth]{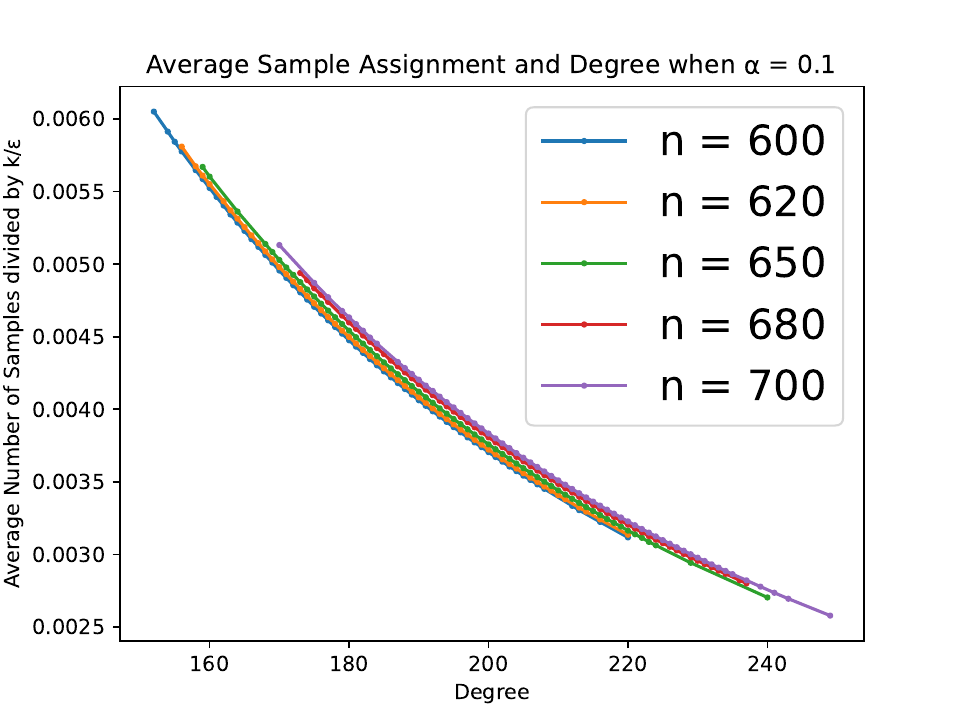}}
    \subfigure[ER Variance]{ \label{ER_VAR} \includegraphics[width = 0.48\linewidth]{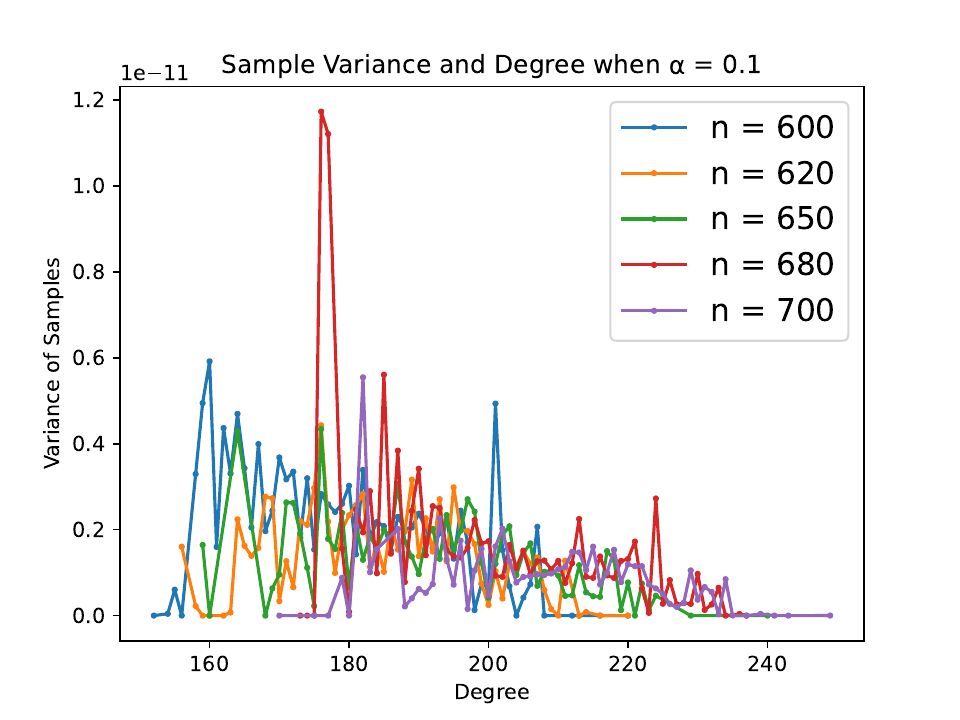}}
    \label{ER_TSC}
\end{minipage}
\begin{minipage}[t]{\linewidth}
     \subfigure[RN Average]{ \label{Real_AVG} \includegraphics[width = 0.48\linewidth]{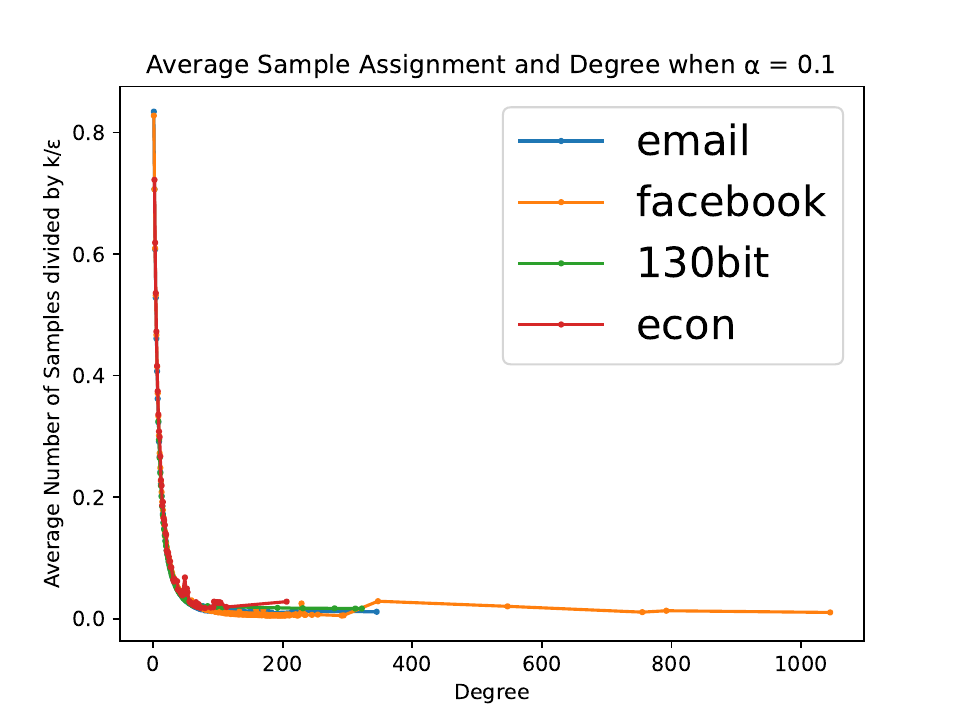}}
    \subfigure[RN Variance]{ \label{Real_VAR} \includegraphics[width = 0.48\linewidth]{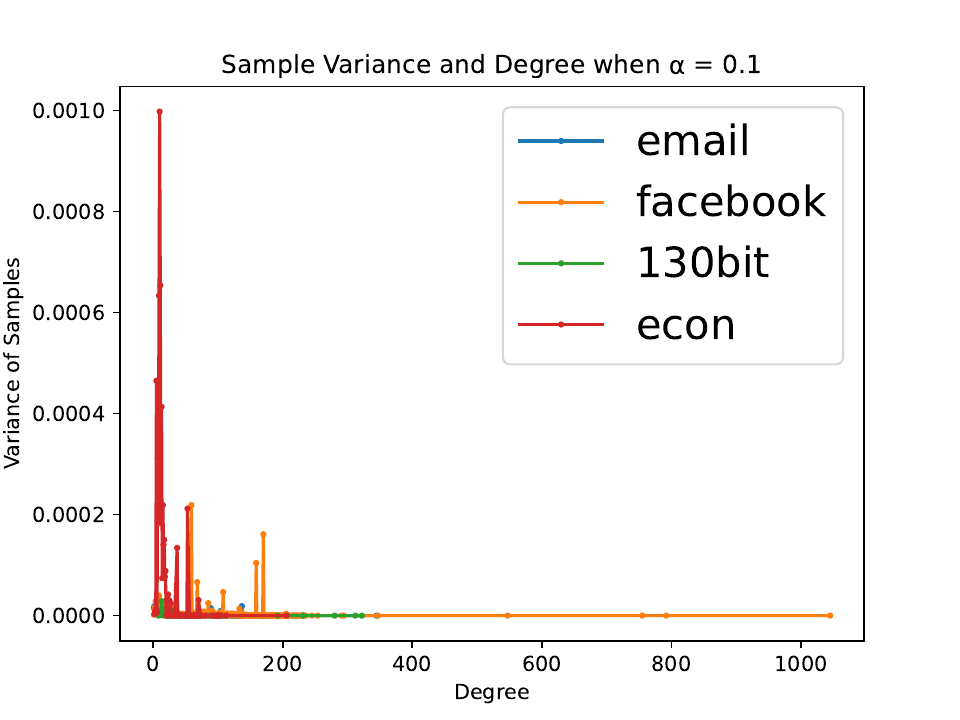}}
    \label{real_TSC}
\end{minipage}
\vspace{-2mm}

\caption{
\small
Relationship between sample distribution and degree for different networks when $\alpha = 1$.
The x-axis is node degree and the y-axis is the average or variance of $N^d = \{\frac{\epsilon m_i^*}{k}: d_i = d, i \in [n]\}$ where $m_i^*, \forall i \in [n]$ is the solution of Equation \ref{eq:opt}.
}
\label{fig:All figures, sample/degree}
\end{figure}
\vspace{-3.5mm}

\vspace{-3.5mm}
\begin{figure}[H]

\begin{minipage}[t]{0.3333\linewidth}
    \subfigure[\scriptsize SF with random $v_{ij}$s]{\label{T_PL}  \includegraphics[width =\linewidth]{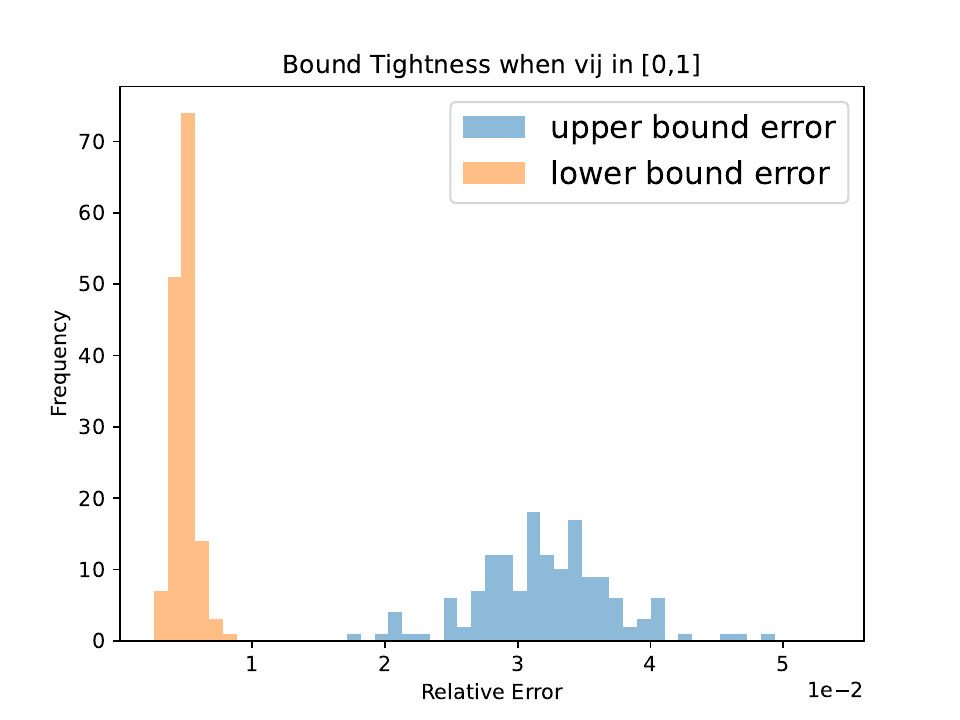}}
\end{minipage}%
\begin{minipage}[t]{0.3333\linewidth}
     \subfigure[\scriptsize ER with random $v_{ij}$s]{\label{T_ER} \includegraphics[width = \linewidth]{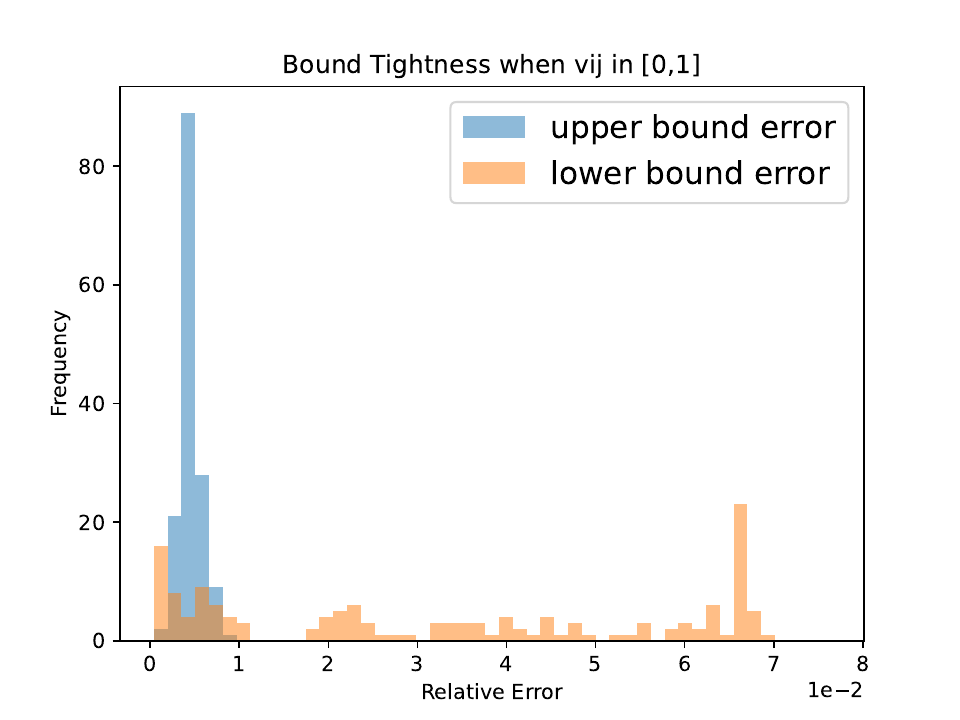}}
\end{minipage}%
\begin{minipage}[t]{0.3333\linewidth}
     \subfigure[\scriptsize RR with random $v_{ij}$s]{ \label{T_RR} \includegraphics[width = \linewidth]{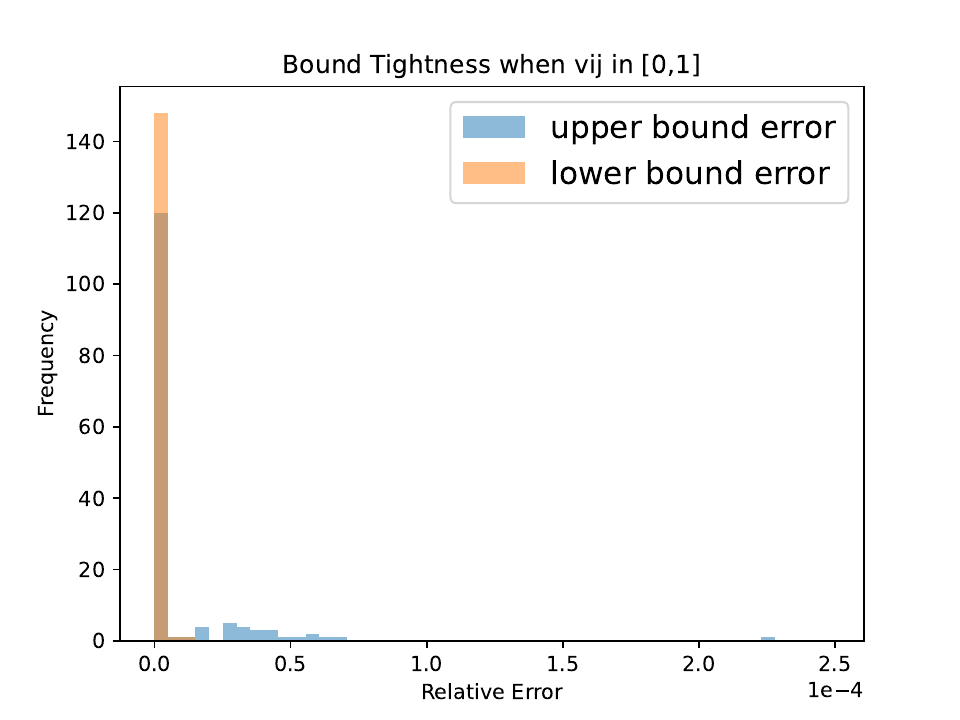}}
\end{minipage}

\vspace{-2mm}
\caption{\small Bound tightness for synthetic networks for random $v_{ij}$s. The x-axis is the relative errors of upper/lower bounds and the y-axis is the frequency.}
\label{Tight_special}
\end{figure}



\vspace{-3.5mm}

For real-world networks, we generate 50 different $v_{ij} \in [0,1]$ and denote the maximum relative error of upper/lower bounds as $err_u/err_l$. From Table \ref{real_network_table}, the relative errors are still relatively small for real-world networks, showing the tightness of our bounds.

\textbf{Network Gain (Table \ref{tab:graph-gain} and Lemma \ref{lem:regular}). } We empirically demonstrate the network gain results for random $d$-regular graphs (Lemma \ref{lem:regular}). To calculate the network gain, we replace the $TSC(G, \pmb{v}, k, \epsilon)$ in network gain (Equation \ref{eq:NetG}) by its upper bound $\sum_{i=1}^n m_i^* + nk + n$. This allows us to calculate a lower bound of the network gain, which aligns with the trend of lower bound in Lemma \ref{lem:regular}. We set $k = 5$ and $\epsilon = 0.01$, making $\frac{k}{\epsilon} = 500$. Additionally, we set $\alpha = 1$ and consider degree ranging from $2$ to $10$. We choose different number of nodes $n$ such that degree $d \le \sqrt{n}$, which implies a $d^2$ gain in theory (4$^{th}$ row of Table \ref{tab:graph-gain} and Lemma \ref{lem:regular}). The empirical result, shown in Figure \ref{fig: RR_Gain}, exactly confirms the gain is $\Theta(d^2)$ in our setup. Moreover, we also found that when degree $d = 11$, the gain becomes infinity. This is because $\frac{k}{d^2\epsilon} = \frac{500}{121} < k = 5$. According to the 4$^{th}$ row of Table \ref{tab:graph-gain} (Lemma \ref{lem:regular}), in this scenario, it is sufficient for each node to have $k$ samples, resulting in an infinite network gain.

\vspace{-3.5mm}
\begin{figure}[H]
\centering
\includegraphics[width = 0.75\linewidth]{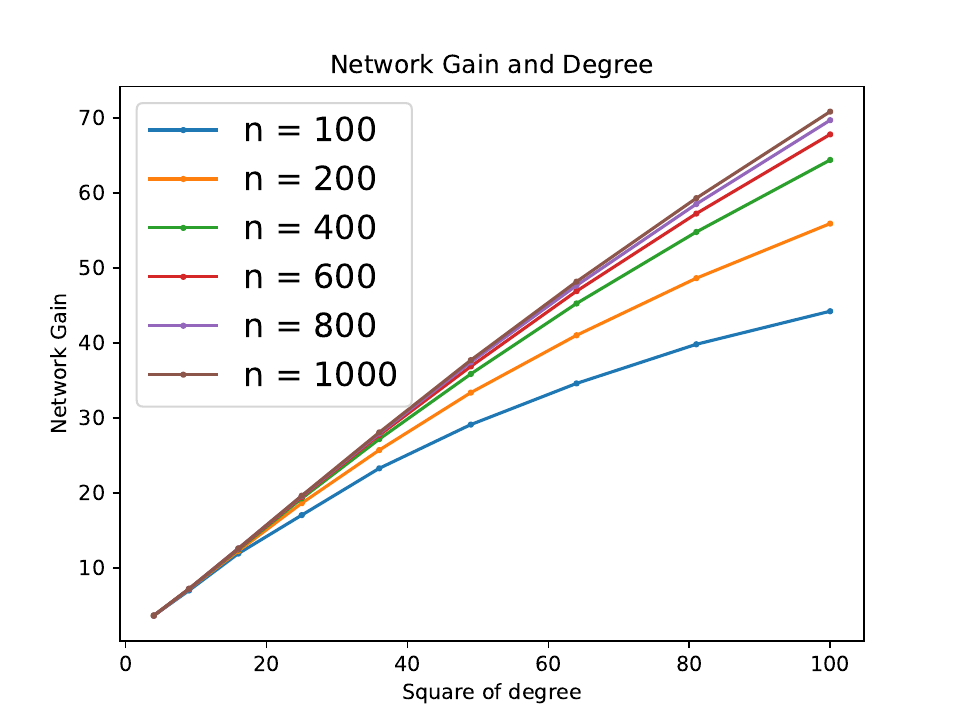}
\vspace{-2mm}

\caption{\small Network Gain of RR. The x-axis is the square of degree and the y-axis is a lower bound of network gain.}
\label{fig: RR_Gain}
\end{figure}
\vspace{-4mm}

\textbf{Impact of influence factors.}
When $\alpha$ is small, sample assignment is mostly independent of $n$; however, for large $\alpha$, increasing $n$ leads to fewer samples for nodes of the same degree. Additionally, the bounds in Theorem \ref{thm: general tight bound} get worse as influence factors increase. See \cite{liu2023sample} for details.




%% file: Discussion.tex
\vspace{-3mm}
\section{Discussion and future work}
\label{sec:discussion}

We initiate a study on the total sample complexity (TSC) required for opinion formation over networks. By extending the standard opinion formation game into a machine-learning framework, we characterize the minimal samples needed to ensure low generalization error for equilibrium models. We provide tight bounds on TSC and analyze sample distributions, showing that network structure significantly reduces TSC, with an inverse relationship to node degree being an interesting result.

Our formulation opens several research directions. We currently focus on (generalized) linear regression, but it’s worth investigating if other problems like kernel methods or soft SVMs. Moreover, it turns out that our problem formulation is similar to the best arm identification problem for multi-armed bandits with fixed confidence~\citep{karnin2013almost}. It is interesting to extend our results analog to the fixed budget setting. Namely, assuming the number of samples is fixed, how to assign them to minimize the total error? Another future direction is to consider individual-level incentives for sample collection, building on the results of~\cite{blum2021one}.

%% file: preliminaries-appendix.tex
\section{Additional details on preliminaries}
\label{sec:prelim-appendix}

\begin{table}[htb!]
\centering
\begin{tabular}{|p{1in}|p{5in}|}
\hline
$D$ & Data distribution\\
\hline
$k$ & Dimension of $D$\\ \hline
$S_i$ & Set of samples for agent $i$; $S_i= \{(x_{i}^j,
y_{i}^j)\}_{j \in [m_i]}$ \\ \hline
$m_i$ & Number of samples for agent $i$\\ \hline
$\Bar{\theta}_i$ & Initial model for agent $i$ (also known as ``internal opinion'')\\ \hline
$\theta_i$ & Model learned by agent $i$ in opinion formation game\\ \hline
$N(i), d_i$ & Set of neighbors and degree of node $i$, respectively\\ \hline
$\theta^{eq}_i$ & Model for agent $i$ in the Nash equilibrium\\ \hline
\makecell[l]{ \\ $\widetilde{M}(k, \epsilon)$} & Minimum number of samples that any agent
$i$ would need to ensure that the best model $\bar{\theta}_i$ learned
using only their samples ensures that the expected error is at most
some target $\epsilon$ (bounded in Theorem~\ref{thm:single-error} for linear regression)\\ \hline
$a_{ij}$ & We denote $a_{ij} = (W^{-1})_{ij}$ for all $i,j \in [n]$ to simplify notations. \\ \hline
$b_{ij}$ &  We denote $b_{ij} = [(W^{-1} \circ W^{-1})^{-1}]_{ij}$ for all $i,j \in [n]$.\\
\hline
\end{tabular}
\caption{Notation used in the paper}
\label{tab:notation-appendix}
\end{table}

%% file: app-2-proof.tex
\section{Missing proof in Section \ref{sec:problem formulation}}
\label{app:formulation proof}

\subsection{Missing proof in Section \ref{sec:problem formulation}}
\label{app:OF}
\textbf{Lemma \ref{lemma:potential}.} The unique Nash equilibrium $\pmb{\theta^{eq}} = (\theta_1^{eq}, \cdots, \theta_n^{eq})^\top$ of the above game is $\pmb{\theta^{eq}} = W^{-1}\pmb{\Bar{\theta}}$ where $W_{ij} = \left\{ 
\begin{array}{ll}
       \sum_{j\in N(i)}v_{ij} + 1 & j=i \\
       -v_{ij}  & j \in N(i)\\
       0  & j \notin N(i), j \neq i
\end{array}
\right.$ and $\pmb{\Bar{\theta}} = (\Bar{\theta}_1, \cdots, \Bar{\theta}_n)$. When all $v_{ij} = \alpha \ge 0$, $W_{ij} = \left\{ 
\begin{array}{ll}
       \alpha D + 1 & j=i \\
       -\alpha  & j \in N(i)\\
       0  & j\notin N(i), j\neq i
\end{array}\right.$. Furthermore, we have $\sum_{j=1}^n W^{-1}_{ij} = 1$ and $W^{-1}_{ij} \ge 0$ for all $i,j \in [n]$. 

\begin{proof}
For any agent $i$, given $\theta_j(j \in N(i))$, by setting the gradient of the loss function to zero, we have the best response for $i$ is $\theta_i = \frac{\Bar{\theta}_i + \sum_{j \in N(i)} v_{ij} \theta_j}{1 + \sum_{j \in N(i)}v_{ij}}$. 
An action profile $\pmb{\theta^{eq}} = (\theta_1^{eq}, \cdots, \theta_n^{eq})^\top$ is a Nash equilibrium if and only if for any $i$, $\theta_i^{eq} = \frac{\Bar{\theta}_i + \sum_{j \in N(i)} v_{ij} \theta_j^{eq}}{1 + \sum_{j \in N(i)}v_{ij}}$, which is equivalent to $(\sum_{j \in N(i)}v_{ij} + 1)\theta_i^{eq} -  \sum_{j \in N(i)} v_{ij}\theta_j^{eq} = \Bar{\theta}_i $ for any $i$.  Combining all these equations, we have $W \pmb{\theta}^{eq} = \pmb{\Bar{\theta}}$ where $W_{ij} = \left\{ 
\begin{array}{ll}
       \sum_{j \in N(i)}v_{ij} + 1 & j=i \\
       -v_{ij}  & j \in N(i)\\
       0  & j\notin N(i), j\neq i
\end{array}
\right.$ Note that $W =  \mathcal{L} + I_n$ where $\mathcal{L}$ is the weighted Laplacian matrix of the graph $G$ and $I_n$ is the $n \times n$ identity matrix. Since $\mathcal{L} $ is positive semi-definite, $W$ is positive definite. Thus $W$ is invertible and $\pmb{\theta}^{eq} = W^{-1}\pmb{\bar{\theta}}$. Given the graph structure, $\pmb{\theta}^{eq}$ is unique and thus the Nash equilibrium is unique.  When all $v_{ij} = \alpha$, we can simplify $W_{ij} = \left\{ 
\begin{array}{ll}
       \alpha D + 1 & j=i \\
       -\alpha  & j \in N(i)\\
       0  & j\notin N(i), j\neq i
\end{array}\right.$ by direct calculation.

Since $v_{ij} > 0$, every non-diagonal element in $W$ is negative. On the other hand, since $W =  \mathcal{L} + I_n$ where $\mathcal{L}$ is the weighted Laplacian matrix of $G$, $W$ is symmetric and positive definite. Thus $W$ is a Stieltjes matrix. From \citet{young2014iterative}, the inverse of any Stieltjes matrix is a non-negative matrix. Thus, $W^{-1}$ is a non-negative matrix.

Let $\bar{1}$ be the vector in which all elements are $1$.  Since $W\bar{1} = \bar{1}$, we have $\bar{1} = W^{-1}W \bar{1} = W^{-1}\bar{1}$.
Thus, $\sum_{j = 1}^n (W^{-1})_{ij} = 1$ for any $i$ and $(W^{-1})_{ij} \in [0,1]$.
Since the best response update for $i$ is $\theta_i = \frac{\Bar{\theta}_i + \sum_{j \in N(i)} v_{ij} \theta_j}{1 + \sum_{j \in N(i)}v_{ij}}$, the final weight $(W_{ij})^{-1} = 0$ if and only if $i$ and $j$ are not connected. Since we consider a connected network(or we can consider different connected components separately), we have $W_{ij}^{-1} > 0$. Thus, $W_{ij}^{-1}$s are normalized weights.
\end{proof}

\subsection{Missing proofs in Section \ref{sec:error bounds}}
\label{app:general setting}

We consider a slightly general case here, where 
every agent $i$'s dataset $S_i = \{(x_{i}^j, y_{i}^j)\}_{j \in
  [m_i]}$ has $m_i \ge k$ samples, where $x_i^j \in \mathbb{R}^k, \,\, \forall j
\in [m_i]$ are independently drawn from distribution $D$. We assume there is a
global ground-truth model $\theta^*$ such that $y_i^j = f(\theta^*, x_i^j) + \eta_i$ where $\eta_i \in H$ is an unbiased noise with bounded variance $\sigma^2$. We further assume $f$ has linear structure over $\theta$, which means $f(\lambda_1\theta_1 + \lambda_2\theta_2, x) = \lambda_1 f(\theta_1, x) + \lambda_2 f(\theta_2, x)$. This capture the case of generalized linear regression \cite{audibert2010linear} where $f(\theta, x) = \phi(x)^\top \theta$ and $\phi$ is an arbitrary mapping function.

We further assume every agent learns a unbiased estimator $\Bar{\theta}_i$ from their dataset $S_i$ such that $\E_{S_i}\left[\Bar{\theta}_i\right] = \theta^*$. Note that in ordinary linear regression $(X_i^\top X_i)^{-1}X_i^\top Y_i$ is an unbiased estimator. In generalized linear regression, $(\Phi_i^\top \Phi_i)^{-1}\Phi_i^\top Y_i$ is unbiased where $\Phi_i = [\phi(x_i^1), \cdots, \phi(x_i^{m_i})]$.

We extend the definition of $L(\cdot)$ as follows.
\begin{align*}
    &L(\theta_i^{eq}) = \sup_{\eta_1, \cdots, \eta_n \in \mathcal{N} }\E_{x \sim D, \forall j, S_j \sim D^{m_j}(\eta_j)}\left[\left(f(\theta_i^{eq}, x) - f(\theta_i^*, x) \right)^2\right]\\
     &L(\Bar{\theta}_i) = \sup_{\eta_i \in \mathcal{N}}\E_{x \sim D, S_i \sim D^{m_i}(\eta_i)}\left[\left(f(\Bar{\theta}_i, x) - f(\theta_i^*, x) \right)^2\right].
\end{align*}

\noindent \textbf{Theorem \ref{thm: one to all}. } If Assumption \ref{assum: non-degenerate} holds, then for every agent $i$, 
\[
L(\theta_i^{eq}) = \sum_{j=1}^n \left(W_{ij}^{-1}\right)^2 L(\Bar{\theta}_j)
\]
where $\theta_i^{eq}$ and matrix $W$ is defined in Lemma \ref{lemma:potential}. Moreover, if we consider linear regression and the condition in Theorem \ref{thm:single-error} holds, we have
\begin{align*}
    L(\theta_i^{eq}) = \Theta\left( k \sum_{j=1}^n \frac{\left(W_{ij}^{-1}\right)^2}{m_j} \right).
\end{align*}

\begin{proof}
To simplify notation, we define $w_{ij} = W_{ij}^{-1}$ in our proof. By Lemma \ref{lemma:potential}, we have $\sum_{j=1}^n w_{ij} = 1$ and $w_{ij} \ge 0$ for all $i,j \in [n]$. For any agent $i$ and any noise $\eta_1, \cdots, \eta_n \in \mathcal{N}$, we have

\begin{align*}
&\E_{x \sim D, \forall j, S_j \sim D^{m_j}(\eta_j)}\left[\left(f\left(\theta_i^{eq}, x\right) - f(\theta^*, x)\right)^2\right]
\\&=\E_{x \sim D, \forall j, S_j \sim D^{m_j}(\eta_j)}\left[\left(f\left(\sum_{j=1}^n w_{ij}\Bar{\theta}_j, x\right) - f(\theta^*, x)\right)^2\right]
\\&= \E_{x \sim D, \forall j, S_j \sim D^{m_j}(\eta_j)}\left[\left(\sum_{j=1}^n w_{ij}f\left(\Bar{\theta}_j, x\right) - f(\theta^*, x)\right)^2\right] \tag{$f(\theta,x)$ is linear for $\theta$}
\\&= \E_{x \sim D, \forall j, S_j \sim D^{m_j}(\eta_j)}\left[\left(\sum_{j=1}^n w_{ij}\left(f\left(\Bar{\theta}_j, x\right) - f\left(\theta^*, x\right)\right)\right)^2\right] \tag{$\sum_{j=1}^n w_{ij} = 1$}
\\&=\mathbb{E}_{x \sim D, \forall j, S_j \sim D^{m_j}(\eta_j)} \left[ \sum_{j=1}^n w_{ij}^2 \left(f\left(\Bar{\theta}_j, x\right) - f\left(\theta^*, x\right)\right)^2 \right] 
\\ &\qquad \quad + \underbrace{\sum_{j=1}^n \sum_{k\neq i} w_{ij}w_{ik} \mathbb{E}_{x \sim D, \forall j, S_j \sim D^{m_j}(\eta_j)}\left[\left(f\left(\Bar{\theta}_j, x\right) - f\left(\theta^*, x\right)\right)\left(f\left(\Bar{\theta}_k, x\right) - f\left(\theta^*, x\right)\right)\right]}_{=0}
\\&= \sum_{j=1}^n w_{ij}^2\mathbb{E}_{x \sim D,  S_j \sim D^{m_j}(\eta_j)} \left[ \left(f\left(\Bar{\theta}_j, x\right) - f\left(\theta^*, x\right)\right)^2 \right] \tag{$\Bar{\theta}_i, \,\, \forall i \in [n]$ are unbiased estimators}
\end{align*}
Thus, we have
\begin{align*}
    L(\theta_i^{eq}) &= \sup_{\eta_1, \cdots, \eta_n \in \mathcal{N}} \E_{x \sim D, \forall j, S_j \sim D^{m_j}(\eta_j)}\left[\left(f(\theta_i^{eq}, x) - f(\theta^*, x) \right)^2\right]
    \\&= \sup_{\eta_1, \cdots, \eta_n \in \mathcal{N}} \sum_{j=1}^n w_{ij}^2\mathbb{E}_{x \sim D,  S_j \sim D^{m_j}(\eta_j)} \left[ \left(f\left(\Bar{\theta}_j, x\right) - f\left(\theta^*, x\right)\right)^2 \right] 
    \\&= \sum_{j=1}^n w_{ij}^2 \sup_{\eta_j \in \mathcal{N}} \mathbb{E}_{x \sim D,  S_j \sim D^{m_j}(\eta_j)} \left[ \left(f\left(\Bar{\theta}_j, x\right) - f\left(\theta^*, x\right)\right)^2 \right]  \tag{independece of noises}
    \\&= \sum_{j=1}^n w_{ij}^2 L(\Bar{\theta}_j).
\end{align*}
For linear regression and the given condition in Theorem \ref{thm:single-error}, we directly have
\begin{align*}
    L(\theta_i^{eq}) = \Theta\left( k \sum_{j=1}^n \frac{\left(W_{ij}^{-1}\right)^2}{m_j} \right).
\end{align*}
Note that we can also get the counterpart of Theorem \ref{thm:single-error} for generalized linear regression. The only thing we need to change is that 
Assumption \ref{assum: non-degenerate} and Assumption \ref{assum: small-ball} should be taken for the distribution of $\phi(x)$ rather than $x$.
\end{proof}
\noindent \textbf{Theorem \ref{thm:general_optimal}. }
For any $\epsilon >  0$, let $(m_i^*, i=1,\ldots,n)$ denote an optimal solution of the following optimization problem as a measure of the minimum samples for opinion formation on graph $G$ with influence factor $\pmb{v}$.

\begin{align*}
    \min_{m_1, \cdots, m_n} \quad &\sum_{i=1}^n m_i \nonumber\\
    \textrm{s.t.} \quad & \sum_{j=1}^n \frac{(W_{ij}^{-1})^2}{m_j} \le \frac{\epsilon}{k}, \,\, \forall i \label{eq:opt} \\
    &m_i > 0, \,\, \forall i \nonumber
\end{align*}
where $W$ is defined in Lemma \ref{lemma:potential}.
Then, $TSC(G, \pmb{v}, k, \epsilon) = \Theta(\sum_{i=1}^n m_i^* + nk)$. Moreover, given $G$ and  $\pmb{v}$, $\frac{\epsilon m_i^\star}{k}$ is a fixed value for any $k$ and $\epsilon$. If $\frac{\epsilon m_i^\star}{k} \ge \epsilon$ for every $i \in [n]$,
then $TSC(G, \pmb{v}, k, \epsilon) = \Theta(\sum_{i=1}^n m_i^*)$.

\begin{proof}
We first show a tight characterization of TSC. By Theorem \ref{thm: one to all}, there exists constant $c_1,c_2$ such that
\begin{equation} 
   c_1 k\sum_{j=1}^n \frac{\left(W_{ij}^{-1}\right)^2}{m_i} \le   L(\theta_i^{eq}) \le c_2 k\sum_{j=1}^n \frac{\left(W_{ij}^{-1}\right)^2}{m_i}\label{eq:c1c2eq}
\end{equation}
By definition, $TSC(G, \pmb{v}, k, \epsilon)$ is the solution of Equation \ref{eq:real TSC}
\begin{align}
    \min_{m_1, \cdots, m_n} \quad &\sum_{i=1}^n m_i \nonumber\\
    \textrm{s.t.} \quad &  L(\theta_i^{eq}) \le \epsilon , \,\, \forall i \label{eq:real TSC}\\
    &m_i \ge k, \,\, \forall i \nonumber \\
    &m_i \in \mathbb{Z}^+ \nonumber
\end{align}
We consider the following two optimizations. Denote the solution of Equation \ref{eq:TSC_L} as $\text{Opt-L}(G, \pmb{v}, k, \epsilon)$ and the solution of Equation \ref{eq:TSC_U} as $\text{Opt-U}(G, \pmb{v}, k, \epsilon)$.

\begin{minipage}{0.5\linewidth}
\begin{align}
    \min_{m_1, \cdots, m_n} \quad &\sum_{i=1}^n m_i \nonumber\\
    \textrm{s.t.} \quad & c_1 k\sum_{j=1}^n \frac{\left(W_{ij}^{-1}\right)^2}{m_i}  \le \epsilon , \,\, \forall i \label{eq:TSC_L}\\
    &m_i \ge k, \,\, \forall i \nonumber \\
    &m_i \in \mathbb{Z}^+ \nonumber
\end{align}
\end{minipage}%
\begin{minipage}{0.5\linewidth}
\begin{align}
    \min_{m_1, \cdots, m_n} \quad &\sum_{i=1}^n m_i \nonumber\\
    \textrm{s.t.} \quad &  c_2 k\sum_{j=1}^n \frac{\left(W_{ij}^{-1}\right)^2}{m_i}  \le \epsilon , \,\, \forall i \label{eq:TSC_U}\\
    &m_i \ge k, \,\, \forall i \nonumber \\
    &m_i \in \mathbb{Z}^+ \nonumber
\end{align}
\end{minipage}
From Equation \ref{eq:c1c2eq}, we have
\begin{align}
    \text{Opt-L}(G, \pmb{v}, k, \epsilon) \le TSC(G, \pmb{v}, k, \epsilon) \le \text{Opt-U}(G, \pmb{v}, k, \epsilon).
\label{eq: TSC-LG}
\end{align}

\noindent We restate Equation \ref{eq:opt} here and denote the solution as $m_i^\star$ for every $i \in [n]$.

\begin{align}
    \min_{m_1, \cdots, m_n} \quad &\sum_{i=1}^n m_i \nonumber\\
    \textrm{s.t.} \quad & \sum_{j=1}^n \frac{(W_{ij}^{-1})^2}{m_j} \le \frac{\epsilon}{k}, \,\, \forall i \nonumber \\
    &m_i > 0, \,\, \forall i \nonumber
\end{align}
By defining $m_i' = \frac{\epsilon m_i}{k}$, Equation \ref{eq:opt} could be reformulated to 
\begin{align}
    \min_{m_1', \cdots, m_n'} \quad &\sum_{i=1}^n m_i' \nonumber\\
    \textrm{s.t.} \quad & \sum_{j=1}^n \frac{(W_{ij}^{-1})^2}{m_j'} \le 1, \,\, \forall i \label{eq:first-re}\\
    &m_i' > 0, \,\, \forall i \nonumber
\end{align}
Thus, $\frac{\epsilon m_i^*}{k}$ is a fixed value for any $k$, $\epsilon$. This also implies that if $\frac{\epsilon}{k}$ is replaced by $\frac{1}{c}\cd\frac{\epsilon}{k}$ for some constant $c$ in Equation \ref{eq:opt}, the solution will become $cm_i^*$. Thus, 
\begin{equation}
     \text{Opt-L}(G, \pmb{v}, k, \epsilon) \ge c_1 \sum_{i \in [n]} m_i^*
\label{eq:OPT-LL}
\end{equation}
Since $\text{Opt-L}(G, \pmb{v}, k, \epsilon) \ge nk$, we have
\begin{align}
    \text{Opt-L}(G, \pmb{v}, k, \epsilon) \ge \frac{c_1}{2}\sum_{i \in [n]} m_i^* + \frac{nk}{2}
\label{eq:OPT-LLm}
\end{align}
On the other hand, since $\max\left\{\lceil c_2 m_i^\star \rceil, k \right\}, \,\, \forall i \in [n]$ is always a solution of Equation \ref{eq:TSC_U}, we have
\begin{equation}
     \text{Opt-U}(G, \pmb{v}, k, \epsilon) \le \sum_{i \in [n]} \max\left\{\lceil c_2 m_i^\star \rceil, k \right\} \le  c_2 \sum_{i \in [n]} m_i^* + nk + n
\label{eq:OPT-UU}
\end{equation}
Combing Equation \ref{eq: TSC-LG}, Equation \ref{eq:OPT-LL}, and Equation \ref{eq:OPT-LLm}, we have $TSC(G, \pmb{v}, k, \epsilon) = \Theta(\sum_{i \in [n]} m_i^\star + nk)$. Assigning $\max\left\{\lceil c_2 m_i^\star \rceil, k \right\}$ samples to every agent $i \in [n]$ ensures $L(\theta^{eq}) \le \epsilon$.

If $\frac{\epsilon m_i^\star}{k} \ge \epsilon$ for every $i \in [n]$, then  $m_i^* \ge k$ and $\sum_{i \in [n]}m_i^* \ge nk$. Thus, $TSC(G, \pmb{v}, k, \epsilon) = \Theta(\sum_{i=1}^n m_i^*)$. From experiments in Section \ref{sec:experiment}, we can observe that for many networks,  $\frac{\epsilon m_i^\star}{k} \ge 0.001$, thus, setting $\epsilon = 0.001$ suffices for $TSC(G, \pmb{v}, k, \epsilon) = \Theta(\sum_{i=1}^n m_i^*)$ in these networks.

\end{proof}

%% file: app-3.1-proof.tex
\section{Missing proof in Section \ref{sec:special}}
\label{app:special proof}

\subsection{Dual form of Equation \ref{eq:opt}}
\begin{lemma}
The dual form of Equation \ref{eq:opt} is 
\begin{align}
    \max_{\lambda_1, \cdots, \lambda_n} \quad &2\sum_{i=1}^n \sqrt{\sum_{j=1}^n \lambda_j (W_{ij}^{-1})^2} - \frac{\epsilon}{k}\sum_{i=1}^n \lambda_i \label{eq:dual-eq}\\
    \textrm{s.t.} \quad & \lambda_i \ge 0, \forall i. \nonumber
\end{align}
Moreover, strong duality holds, which implies the solution of Equation \ref{eq:opt} is equal to Equation \ref{eq:dual-eq}.
\label{lem:dual}
\end{lemma}

\begin{proof}
The Lagrangian of optimization Equation \ref{eq:opt} is 
\begin{equation*}
    L = \sum_{i=1}^n m_i + \sum_{i = 1}^n \lambda_i \left(\sum_{j=1}^n \frac{w_{ij}^2}{m_j} - \frac{\epsilon}{k}\right) \qquad (\lambda_i \ge 0) 
\end{equation*} in domain $m_i > 0, \forall i$. When $\lambda_i = 0$ for all $i$, we have $\inf\limits_{m_1, \cdots, m_n} L = 0$. When $\exists \lambda_i \neq 0$, from $\frac{\partial L}{\partial m_i} = 1 - \sum_{j=1}^n \lambda_j \frac{w_{ji}^2}{m_i^2}$, the Lagrangian function gets mininum value when $m_i = \sqrt{\sum_{j=1}^n \lambda_j w_{ij}^2} > 0$. Now \begin{align*}
\min\limits_{m_1, \cdots, m_n} L &= \sum_{i=1}^n \sqrt{\sum_{j=1}^n \lambda_j w_{ij}^2} + \sum_{i=1}^n \lambda_i (\sum_{j=1}^n \frac{w_{ij}^2}{\sqrt{\sum_{k=1}^n \lambda_k w_{kj}^2}} - \frac{\epsilon}{k})
\\&= \sum_{i=1}^n \sqrt{\sum_{j=1}^n \lambda_j w_{ij}^2} + \sum_{j=1}^n(\frac{\sum_{i=1}^n \lambda_iw_{ij}^2}{\sqrt{\sum_{k=1}^n \lambda_k w_{kj}^2}}) - \frac{\epsilon}{k}\sum_{i=1}^n \lambda_i
\\&= 2\sum_{i=1}^n \sqrt{\sum_{j=1}^n \lambda_j w_{ij}^2} - \frac{\epsilon}{k}\sum_{i=1}^n \lambda_i
\end{align*}
Thus when $\exists \lambda_i \neq 0$, the maxmin problem of Lagrangian is 
\begin{align*}
    \max_{\lambda_1, \cdots, \lambda_n} \quad &2\sum_{i=1}^n \sqrt{\sum_{j=1}^n \lambda_j w_{ij}^2} - \frac{\epsilon}{k}\sum_{i=1}^n \lambda_i\\
    \textrm{s.t.} \quad & \lambda_i \ge 0, \forall i  \quad \exists \lambda_i \neq 0
\end{align*}
Note that when all $\lambda_i = 0$, this optimization is $0$ which matches the infimum of the Lagrangian function when all $\lambda_i = 0$. Thus we can remove the constraint $\exists \lambda_i \neq 0$. The dual problem of optimization Equation \ref{eq:opt} is:
\begin{align*}
    \max_{\lambda_1, \cdots, \lambda_n} \quad &2\sum_{i=1}^n \sqrt{\sum_{j=1}^n \lambda_j w_{ij}^2} - \frac{\epsilon}{k}\sum_{i=1}^n \lambda_i\\
    \textrm{s.t.} \quad & \lambda_i \ge 0, \forall i
\end{align*}
Since both the objective function and the feasible region of Equation \ref{eq:opt} is convex, following the Slater condition \citep{boyd2004convex}, strong duality holds because there is a point that strictly satisfies all constraints. 
\end{proof}

\subsection{Proof of Theorem \ref{thm: degree bound}}
To simplify notations, assume $a_{ij} = (W^{-1})_{ij}$ for all $i,j \in [n]$. 
We will first list more observations related to $a_{ij}$.

\begin{observation}
$\frac{1}{\alpha d_i + 1}\le a_{ii} \le \frac{1}{\frac{\alpha}{\alpha + 1}d_i + 1}$.
\label{aii}
\end{observation}
\begin{proof}
From $WW^{-1} = I$, we have $(d_i\alpha + 1)a_{ii} - \alpha \sum_{j \in N(i)} a_{ij} = 1$. From $a_{ii} + \sum_{j\neq i}a_{ij} = 1$, we have $(d_i\alpha + 1)a_{ii} - \alpha (1-a_{ii}) \le 1$, which is equivalent to $a_{ii} \le \frac{1}{\frac{\alpha}{\alpha + 1}d + 1}$. On the other hand, $(d_i\alpha + 1)a_{ii} \ge 1$,  $a_{ii} \ge \frac{1}{d_i\alpha + 1}$. Thus $\frac{1}{\alpha d_i + 1}\le a_{ii} \le \frac{1}{\frac{\alpha}{\alpha + 1}d_i + 1}$
\end{proof}

\begin{observation}
$a_{ij} \le \frac{\alpha}{\alpha \max\{d_{i},d_{j}\} + 1}$ for $i \neq j$. 
\label{aij}
\end{observation} 

\begin{proof}
From $WW^{-1} = I$, we have $(d_i\alpha + 1)a_{ij} - \alpha \sum_{k\in {N(i)}} a_{kj}  = 0$. Combining with $\sum_{k \in N(i)} a_{kj} = \sum_{k \in N(i)} a_{jk} \le 1$, we have $a_{ij} \le \frac{\alpha}{\alpha d_{i} + 1}$. Similarly, we get $a_{ji} \le \frac{\alpha}{\alpha d_{j} + 1}$. Since $a_{ij} = a_{ji}$, we have $a_{ij} \le \frac{\alpha}{\alpha \max\{d_{i},d_{j}\} + 1}$.
\end{proof}

We restate Theorem \ref{thm: degree bound} and prove it here.

\noindent \textbf{Theorem \ref{thm: degree bound}.}
The optimal solution $\{m^*_i\}$ to Equation~\ref{eq:opt} satisfies
\begin{equation*}
\max\left\{\sum_{i=1}^n \frac{1}{(\alpha d_i + 1)^2}, 1\right\} \cd \frac{k}{\epsilon} \le \sum_{i \in [n]} m_i^* \le \sum_{i=1}^n \frac{\alpha + 1}{\alpha d_i + 1} \cd \frac{k}{\epsilon}
\end{equation*}
Assigning $\frac{\alpha + 1}{\alpha d_i + 1}\cd \frac{k}{\epsilon}$ samples to node $i$ can make sure at the equilibrium, every agent $i \in [n]$ has error $L(\theta^{eq}_i) \le \epsilon$. Assigning $\max\left\{O\left(\frac{\alpha + 1}{\alpha d_i + 1}\cd \frac{k}{\epsilon}\right), k\right\}$ samples to node $i$ is sufficent to solve the TSC problem.

\begin{proof}
(1) For the upper bound, let $m_i = \frac{|S|}{\alpha d_i + 1}$. Equation \ref{eq:opt} then becomes 
\begin{align*}
    \min_{|S|} \quad &|S| \sum_{i=1}^n \frac{1}{\alpha d_i + 1}\\
    \textrm{s.t.} \quad & \frac{1}{|S|}\sum_{j=1}^n a_{ij}^2 (\alpha d_j + 1) \le \frac{\epsilon}{k}, \quad \forall i\\
    &|S| > 0
\end{align*}
The optimal solution is $p^*_{inv} = M(\epsilon)(\sum_{i=1}^n \frac{1}{\alpha d_i + 1})\max\limits_{i}\{\sum_{j=1}^n a_{ij}^2 (\alpha d_j + 1)\}$ when $|S| = \frac{k}{\epsilon} \cd\max\limits_{i}\{\sum_{j=1}^n a_{ij}^2 (d_j + 1)\}$. From Observation \ref{aii} and Observation \ref{aij}, 
\begin{align*}
    \sum_{j=1}^n a_{ij}^2 (\alpha d_j + 1) &= a_{ii}^2(\alpha d_i + 1) + \sum_{j \neq i} a_{ij}^2 (\alpha d_j + 1)
    \\&\le a_{ii}\frac{\alpha + 1}{\alpha d_i + \alpha + 1} (\alpha d_{i} + 1) + \sum_{j \neq i} a_{ij}\frac{\alpha}{\alpha \max\{d_i, d_j\} + 1 } (\alpha d_j + 1)
    \\&\le  a_{ii}\frac{\alpha + 1}{\alpha d_i + 1} (\alpha d_{i} + 1) + \sum_{j \neq i} a_{ij}\frac{\alpha + 1}{\alpha d_j + 1 } (\alpha d_j + 1)
    \\&= (\alpha + 1) \sum_{j=1}^n a_{ij}
    \\&= (\alpha + 1)
\end{align*}
Note that this bound holds for all $i$, thus we have
\begin{align*}
p^*_{inv} &= \frac{k}{\epsilon}\cd(\sum_{i=1}^n \frac{1}{\alpha d_i + 1})\max\limits_{i}\{\sum_{j=1}^n a_{ij}^2 (\alpha d_j + 1)\}
\\& \le \frac{k}{\epsilon}\cd(\sum_{i=1}^n \frac{\alpha + 1}{\alpha d_i + 1}) 
\end{align*}

Thus, $\sum_{i \in [n]} m_i^* \le p^*_{inv} \le (\sum_{i=1}^n \frac{\alpha + 1}{\alpha d_i + 1}) \cd \frac{k}{\epsilon}$. Since $m_i = \frac{\alpha + 1}{\alpha d_i + 1}\cd \frac{k}{\epsilon}$ is a solution of Equation \ref{eq:opt}, assign $m_i$ samples to to node $i$ can make sure at the equilibrium, every agent $i \in [n]$ has error $L(\theta^{eq}_i) \le \epsilon$. Assigning $\max\left\{O\left(\frac{\alpha + 1}{\alpha d_i + 1}\cd \frac{k}{\epsilon}\right), k\right\}$ samples to node $i$ is sufficent to solve the TSC problem.

(2) For the lower bound, we utilize the dual form Equation \ref{eq:dual-eq}. Let $\lambda_i = \lambda (\alpha d_i + 1)^x > 0$ where $x \in R$ will be optimized later. The dual optimization then becomes 
\begin{align*}
    \max_{\lambda} \quad &2\sqrt{\lambda}\sum_{i=1}^n \sqrt{\sum_{j=1}^n a_{ij}^2 (\alpha d_j + 1)^x} - \frac{\lambda \epsilon}{k}\sum_{i=1}^n (\alpha d_i + 1)^x\\
    \textrm{s.t.} \quad & \lambda > 0
\end{align*}
Let $d^*_{power}(x)$ be the solution of this optimization given $x \in R$. Since the objective function is a quadratic function, the function gets maximal value when 
\begin{equation*}
\sqrt{\lambda} = \frac{k}{\epsilon}\cd\frac{\sum_{i=1}^n \sqrt{\sum_{j=1}^n a_{ij}^2 (\alpha d_j + 1)^x}}{\sum_{i=1}^n (\alpha d_i + 1)^x} > 0
\end{equation*}
Thus, 
\begin{equation*}
d^*_{power}(x) = \frac{k}{\epsilon}\cd\frac{(\sum_{i=1}^n \sqrt{\sum_{j=1}^n a_{ij}^2 (\alpha d_j + 1)^x})^2}{\sum_{i=1}^n (\alpha d_i + 1)^x} 
\end{equation*}
From Observation \ref{aii}, we have $a_{ii}^2 \ge \frac{1}{(d_i\alpha + 1)^2}$, thus 
\begin{align*}
d^*_{power}(x) &= \frac{k}{\epsilon}\cd\frac{(\sum_{i=1}^n \sqrt{\sum_{j=1}^n a_{ij}^2 (\alpha d_j + 1)^x})^2}{\sum_{i=1}^n (\alpha d_i + 1)^x} 
\\ &\ge \frac{k}{\epsilon}\cd\frac{(\sum_{i=1}^n (\alpha d_j + 1)^{\frac{x-2}{2}})^2}{\sum_{i=1}^n (\alpha d_i + 1)^x} 
\end{align*}
Note that this holds for all $x \in R$, to get the tightest bound, we should find the maximal value of the last term. By  Cauchy–Schwarz inequality, we have
\begin{align*}
\frac{(\sum_{i=1}^n (\alpha d_j + 1)^{\frac{x-2}{2}})^2}{\sum_{i=1}^n (\alpha d_i + 1)^x} &\le \sum_{i=1}^n \frac{((\alpha d_j + 1)^{\frac{x-2}{2}})^2}{(\alpha d_i + 1)^x}
\\ &= \sum_{i=1}^n \frac{1}{(\alpha d_i + 1)^2}
\end{align*}
This maximum value is achieved when $x = -2$. Thus $ \sum_{i \in [n]} m_i^* \ge d^*_{power}(x) \ge \sum_{i=1}^n \frac{1}{(\alpha d_i + 1)^2}\cd\frac{k}{\epsilon}$.

(3) In the dual form Equation \ref{eq:dual-eq}, $\lambda_i = \frac{k^2}{n\epsilon^2}$ for any $i$ is a feasible solution. From strong duality, we have
\begin{align*}
   \sum_{i \in [n]} m_i^* &\ge 2\frac{k}{\epsilon}\cd\sum_{i=1}^n\sqrt{\frac{1}{n}\sum_{j=1}^n(W_{ij}^{-1})^2} - \frac{k}{\epsilon}
    \\&\ge 2\frac{k}{\epsilon}\sum_{i=1}^n\sqrt{\frac{1}{n}
    \left(\frac{\sum_{j=1}^nW_{ij}^{-1}}{\sqrt{n}}\right)^2} - \frac{k}{\epsilon} \tag{Cauchy-Schwarz}
    \\&\ge \frac{k}{\epsilon} \tag{$\sum_{j=1}^nW_{ij}^{-1}$ for all $i$}
\end{align*}
\end{proof}

\subsection{Bounds for the Network Gain}
\label{subsec: gain}
From Theorem \ref{thm:single-error}, for every agent $i \in [n]$, we have 
\begin{align}
   c_1 \frac{k}{m_i}  \le L(\Bar{\theta}_i) \le c_2 \frac{k}{m_i}
\label{eq:single_c1c2}
\end{align}
Recall that $\widetilde{M}(k, \epsilon)$ is the minimal number of samples one agent needs to learn a model with $\epsilon$ error. We have
\begin{align}
  c_1 \frac{k}{\epsilon} \le \widetilde{M}(k, \epsilon) \le  c_2 \frac{k}{\epsilon}
\label{eq:real_M_c1c2}
\end{align}
By Theorem \ref{thm: one to all}, we have
\begin{align}
    c_1 k\sum_{j=1}^n \frac{\left(W_{ij}^{-1}\right)^2}{m_i} \le   L(\theta_i^{eq}) \le c_2 k\sum_{j=1}^n \frac{\left(W_{ij}^{-1}\right)^2}{m_i}.
\label{eq:c1c2eq-2}
\end{align}

\noindent \textbf{Corollary \ref{cor:gain}.} For any network $G$, if $\epsilon \le \frac{c_2(\alpha + 1)}{\alpha \max_i d_i + 1}$, then $\ngain(G,\alpha, k, \epsilon) \ge \Omega\left(\frac{n}{\sum_{i=1}^n \frac{\alpha + 1}{\alpha d_i + 1}}\right)$. If $\epsilon \le \frac{c_1}{2 \max_i(\alpha d_i + 1)^2}$, then $\ngain(G,\alpha, k, \epsilon) \le O\left(\min\left\{\frac{n}{\sum_{i \in [n]} \frac{1}{(\alpha d_i + 1)^2}}, n\right\}\right)$.

\begin{proof}
We first derive the upper bound of the network gain. From Equation \ref{eq:c1c2eq-2}, Equation \ref{eq: TSC-LG}, and Equation \ref{eq:OPT-UU}, we have
\begin{align*}
    TSC(G, \pmb{v}, k, \epsilon) \le   c_2 \sum_{i \in [n]} m_i^* + nk + n.
\end{align*}
This implies
\begin{align*}
     \ngain(G, \pmb{v}, k, \epsilon) = \frac{n (\widetilde{M}(k, \epsilon)- k)}{TSC(G, \pmb{v}, k, \epsilon) - nk} \ge \frac{c_1 n \frac{k}{\epsilon} - nk}{ c_2 \sum_{i \in [n]} m_i^* + n} \ge \frac{c_1 n \frac{k}{\epsilon} - nk}{ c_2 \sum_{i=1}^n \frac{\alpha + 1}{\alpha d_i + 1}\cd\frac{k}{\epsilon} + n}
\end{align*}
where the last step applies the upper bound in Theorem \ref{thm: degree bound}. The assumption $\epsilon \le \frac{c_2(\alpha + 1)}{\alpha \max_i d_i + 1}$ implies for every $i \in [n]$,  $c_2 \frac{\alpha + 1}{\alpha d_i + 1}\cd\frac{k}{\epsilon} \ge k$ and  $c_2 \sum_{i=1}^n \frac{\alpha + 1}{\alpha d_i + 1}\cd\frac{k}{\epsilon} \ge nk$. We have
\begin{align*}
    \ngain(G, \pmb{v}, k, \epsilon) \ge \frac{c_1 n \frac{k}{\epsilon} - nk}{ 2c_2 \sum_{i=1}^n \frac{\alpha + 1}{\alpha d_i + 1}\cd\frac{k}{\epsilon}} = \frac{c_1 n}{ 2c_2 \sum_{i=1}^n \frac{\alpha + 1}{\alpha d_i + 1}} - \frac{nk}{ 2c_2 \sum_{i=1}^n \frac{\alpha + 1}{\alpha d_i + 1}\cd\frac{k}{\epsilon}} \ge \frac{c_1 n}{ 2c_2 \sum_{i=1}^n \frac{\alpha + 1}{\alpha d_i + 1}} - \frac{1}{2} \ge \Omega\left(\frac{n}{\sum_{i=1}^n \frac{\alpha + 1}{\alpha d_i + 1}}\right).
\end{align*}
We now derive a general upper bound of the network gain. From Equation \ref{eq:c1c2eq-2}, Equation \ref{eq: TSC-LG}, and Equation \ref{eq:OPT-LL}, we have
\begin{align*}
    TSC(G, \pmb{v}, k, \epsilon) \ge c_1 \sum_{i \in [n]} m_i^*.
\end{align*}
This implies
\begin{align}
     \ngain(G, \pmb{v}, k, \epsilon) = \frac{n (\widetilde{M}(k, \epsilon)- k)}{TSC(G, \pmb{v}, k, \epsilon) - nk} \le \frac{c_2 n \frac{k}{\epsilon} - nk}{ c_1 \sum_{i \in [n]} m_i^* - nk} \le \frac{c_2 n \frac{k}{\epsilon} - nk}{ c_1 \sum_{i=1}^n \frac{1}{(\alpha d_i + 1)^2}\cd\frac{k}{\epsilon} - nk}
\label{eq:gain_U}
\end{align}
where the last step applies the lower bound in Theorem \ref{thm: degree bound}. The assumption $\epsilon \le \frac{c_1}{2 \max_i(\alpha d_i + 1)^2}$ is equivalent to $\frac{c_1}{2(\alpha d_i + 1)^2}\cd\frac{k}{\epsilon} \ge k$ for every $i \in [n]$. This implies $\frac{c_1}{2} \sum_{i=1}^n \frac{1}{(\alpha d_i + 1)^2}\cd\frac{k}{\epsilon} \ge nk$ which leads to
\begin{align*}
      \ngain(G, \pmb{v}, k, \epsilon)  \le \frac{2c_2 n \cd\frac{k}{\epsilon} - 2nk}{ c_1 \sum_{i=1}^n \frac{1}{(\alpha d_i + 1)^2}\cd\frac{k}{\epsilon}} \le   \frac{2c_2 n}{ c_1 \sum_{i=1}^n \frac{1}{(\alpha d_i + 1)^2}}  \le O\left(\frac{n}{ \sum_{i=1}^n \frac{1}{(\alpha d_i + 1)^2}}\right)
\end{align*}
When $\max_i d_i \le \frac{\sqrt{n} - 1}{\alpha}$, we have $\sum_{i=1}^n \frac{1}{(\alpha d_i + 1)^2} \ge 1$, thus, $\ngain(G, \pmb{v}, k, \epsilon) \le O\left(\min\left\{\frac{n}{ \sum_{i=1}^n \frac{1}{(\alpha d_i + 1)^2}}, n \right\}\right)$. When $\max_i d_i > \frac{\sqrt{n} - 1}{\alpha}$, $\epsilon \le \frac{c_1}{2 \max_i(\alpha d_i + 1)^2} < \frac{c_1}{2n}$, this implies $\frac{c_1}{2}\cd\frac{k}{\epsilon} \ge nk$. We have
\begin{align*}
      \ngain(G, \pmb{v}, k, \epsilon) \le \frac{c_2 n \frac{k}{\epsilon} - nk}{ c_1 \sum_{i \in [n]} m_i^* - nk} \le \frac{c_2 n \frac{k}{\epsilon} - nk}{ c_1\frac{k}{\epsilon} - nk} \le \frac{2c_2}{c_1} n \le O(n)
\end{align*}
where we use the trivial lower bound in Theorem \ref{thm: degree bound}. Thus, $\ngain(G, \pmb{v}, k, \epsilon) \le O\left(\min\left\{\frac{n}{ \sum_{i=1}^n \frac{1}{(\alpha d_i + 1)^2}}, n \right\}\right)$. 
\end{proof}

We then consider another lower bound of network gain, which will be applied to prove the network gain for special networks.
\begin{lemma}
Let $m_i = m, \forall i \in [n]$ be a feasible solution of Equation \ref{eq:opt}. If $\lceil c_2 m \rceil \le k$, the network gain is infinite. If $\lceil c_2 m \rceil > k$, we have

\begin{equation*}
    \ngain(G, \pmb{v}, k, \epsilon) \ge \frac{c_1k}{ 2c_2 m \epsilon} - 1 .
\end{equation*}
\label{lem:general gain}
\end{lemma}

\begin{proof}
From Equation \ref{eq:c1c2eq-2}, Equation \ref{eq: TSC-LG}, Equation \ref{eq:OPT-LL}, we have
\begin{align*}
   nk\le TSC(G, \pmb{v}, k, \epsilon) \le  \sum_{i \in [n]} \max\left\{\lceil c_2 m \rceil, k \right\} 
\end{align*}
If $\lceil c_2 m \rceil \le k$, then $TSC(G, \pmb{v}, k, \epsilon) = nk$ and the network gain is infinite. If $\lceil c_2 m \rceil > k$, then 
\begin{align*}
   TSC(G, \pmb{v}, k, \epsilon) \le  \sum_{i \in [n]} \max\left\{\lceil c_2 m \rceil, k \right\} \le  c_2 n m + nk + n
\end{align*}
$\lceil c_2 m \rceil > k$ implies $c_2 m > \max\{1, k-1\}$ because $k \ge 1$. Combining these two bounds with Equation \ref{eq:real_M_c1c2}, we have
\begin{equation*}
    \ngain(G, \pmb{v}, k, \epsilon) = \frac{n (\widetilde{M}(k, \epsilon)- k)}{TSC(G, \pmb{v}, k, \epsilon) - nk} \ge \frac{c_1 n \frac{k}{\epsilon} - nk}{ c_2 n m + n} \ge \frac{c_1 n \frac{k}{\epsilon} - nk}{ 2c_2 n m} \ge \frac{c_1 k}{ 2c_2 m\epsilon} - \frac{k}{2c_2m} \ge \frac{c_1 k}{ 2c_2 m \epsilon} - 1 .
\end{equation*}
\end{proof}

\subsection{Proof of Network Gain for Special Networks}

Again, To simplify notations, assume $a_{ij} = (W^{-1})_{ij}$ for all $i,j \in [n]$. 

\begin{lemma}
If the network $G$ is an $n$-clique,  $\sum_{i \in [n]} m_i^* = (1 + \frac{n-1}{(n\alpha + 1)^2})\cd\frac{k}{\epsilon}$, where $m_i^* = \frac{(\alpha+1)^2+ (n-1)\alpha^2}{(n\alpha + 1)^2}\cd\frac{k}{\epsilon}$ for every $i \in [n]$. The network gain $\ngain(G, \pmb{v}, k, \epsilon) = \Omega(n)$    
\label{lem: clique}
\end{lemma}

\begin{proof}
For n-clique, we have $\begin{bmatrix}
1+(n-1)\alpha&-\alpha& \cdots &-\alpha
\\-\alpha&1+(n-1)\alpha &\cdots &-\alpha
\\ \vdots & \ddots & \ddots & \vdots
\\-\alpha & -\alpha & \cdots & 1+(n-1)\alpha
\end{bmatrix} 
\begin{bmatrix} 
\theta_1^{eq}
\\ \theta_2^{eq}
\\ \vdots
\\ \theta_n^{eq}
\end{bmatrix}
= 
\begin{bmatrix}
\Bar{\theta}_1
\\ \Bar{\theta}_2
\\ \vdots
\\ \Bar{\theta}_n
\end{bmatrix}$. The solution to this system is $\hat{\theta}_i =  \frac{(\alpha +1)\Bar{\theta}_i + \alpha \sum_{j \neq i} \Bar{\theta}_j}{n\alpha + 1}$ for any $i \in [n]$. Thus for any $i$, $a_{ii} = \frac{\alpha + 1}{n\alpha + 1}, a_{ij} = \frac{\alpha}{n\alpha + 1}(j \neq i)$. Since the constraints of Equation \ref{eq:opt} are symmetric now, there is an optimal solution such that $m_1 = m_2 = \cdots = m_n = |S|$. The constraint now becomes $|S| \ge \frac{k}{\epsilon}\cd\sum_{j=1}^n a_{ij}^2 = \frac{(\alpha + 1)^2 + (n-1) \alpha^2}{(n\alpha +1)^2}\cd\frac{k}{\epsilon} > 0$. Thus the unique optimal solution is $\frac{n(\alpha+1)^2+ n(n-1)\alpha^2}{(n\alpha + 1)^2}\cd\frac{k}{\epsilon}  = (1 + \frac{n-1}{(n\alpha + 1)^2})\cd\frac{k}{\epsilon}$ when $|S| = \frac{(\alpha + 1)^2 + (n-1) \alpha^2}{(n\alpha +1)^2}\cd\frac{k}{\epsilon}$.

From Lemma \ref{lem:general gain}, if $\lceil c_2 m_i^\star \rceil \le k$, the network gain is infinite. If $\lceil c_2 m_i^\star \rceil > k$, then
\begin{align*}
    \ngain(G, \pmb{v}, k, \epsilon) \ge  \frac{c_1k}{ 2c_2 m_i^* \epsilon} - 1  = \frac{c_1 n }{2 c_2 \left(1 + \frac{n-1}{(n\alpha + 1)^2}\right)} - 1 \ge  \frac{c_1 n }{2 c_2 \left(1 + \frac{1}{2\alpha}\right)} - 1 = \Omega(n)
\end{align*}
\end{proof}

\begin{lemma}
For an $n$-agent star graph, $\frac{n-1}{(\alpha + 1)^2}\cd\frac{k}{\epsilon}  \le \sum_{i \in [n]} m_i^* \le n\cd\frac{k}{\epsilon} $. If $\epsilon \le \frac{c_1}{4(\alpha + 1)^2}$, the network gain $\ngain(G, \pmb{v}, k, \epsilon) = O(1)$.
\label{lem: star}
\end{lemma}
\begin{proof}

We can directly get the bounds of $\sum_{i \in [n]} m_i^* $ from the trivial bound and the lower bound of Theorem \ref{thm: degree bound}. Let node $j$ be the center of the star, since every node should have at least $k$ samples, from Equation \ref{eq:gain_U} and the lower bound $\frac{n-1}{(\alpha + 1)^2}\cd\frac{k}{\epsilon} $, we have
\begin{align*}
\ngain(G, \pmb{v}, k, \epsilon) \le \frac{c_2 n\cd\frac{k}{\epsilon}  - nk}{ c_1 \sum_{i=1}^n \frac{1}{(\alpha d_i + 1)^2} \cd\frac{k}{\epsilon}  - nk} \le  \frac{c_2 n \cd\frac{k}{\epsilon} - nk}{ c_1 \frac{n-1}{(\alpha + 1)^2}\cd\frac{k}{\epsilon} - nk}
\end{align*}
From the assumption of $\epsilon$ and $n \ge 2$, we have $\epsilon \le \frac{c_1}{4(\alpha + 1)^2} \le \frac{c_1}{2(\alpha + 1)^2}(1-\frac{1}{n})$. This implies $\frac{c_1}{2} \frac{n-1}{(\alpha + 1)^2}\cd\frac{k}{\epsilon}  \ge nk$, which leads to
\begin{align*}
    \ngain(G, \pmb{v}, k, \epsilon) \le \frac{2c_2 n \cd\frac{k}{\epsilon}  - 2nk}{ c_1 \frac{n-1}{(\alpha + 1)^2}\cd\frac{k}{\epsilon}} \le \frac{2c_2 n}{ c_1 \frac{n-1}{(\alpha + 1)^2}} \le \frac{4c_2(\alpha + 1)^2}{c_1} = O(1).
\end{align*}
\end{proof}

Recall the optimization problem in Theorem \ref{thm:general_optimal}.
To get a lower bound on any $m_i$, we upper bound the quantity $\sum_{j=1}^n (W^{-1})_{ij}^2$ for any $i$. The key idea is to investigate the structure of $W^{-1}$. Since $W=\alpha \mathcal{L}+I$ is a symmetric positive definite matrix when the influence factors have uniform value $\alpha$, it can be written in terms of its eigen-decomposition as $W=PDP^T$ where $P$ is an orthonormal matrix whose columns $p_1,p_2,...,p_n$ are unit eigenvectors of $W$ s.t. for any $i\neq j$, $p_i^\top p_j=0$, and $D$ is the diagonal matrix of the corresponding eigenvalues of $W$. Let $\lambda_1, \lambda_2,...,\lambda_n$ denote the eigenvalues of the Laplacian matrix of $G$ such that $\lambda_1\leq \lambda_2\leq .... \leq \lambda_n$. By definition of $W$, it follows that the eigenvalues of $W$ can be expressed as $1+\alpha\lambda_1, 1+\alpha\lambda_2,...,1+\alpha\lambda_n$. It is well known that for any graph, $\lambda_1=0$ and for $d$-regular graphs, the eigenvalues $\lambda_i\in [0,2d]$ $\forall i\in [n]$. Furthermore, note that $P^TP=PP^T=I$ so that $P^{-1}=P^T$. Another observation we make is that the eigenvectors $p_1, p_2,..,p_n$ of $W$ are simply the eigenvectors of the Laplacian $\mathcal{L}$.

Similarly, we express $W^{-1}$ in terms of its eigen-decomposition. Since $W=PDP^T$, it follows that $W^{-1}=PD^{-1}P^T$ where $D^{-1}$ is the diagonal matrix such that $(D^{-1})_{ii}=\frac{1}{1+\alpha\lambda_i}$ for all $i\in [n]$. To lower bound any entry of $W^{-1}$, it suffices to compute the eigenvalues and eigenvectors of the Laplacian matrix $L$. Based on the above observations, we have, 
\begin{align*}
    (W^{-1})_{ij}=\sum_{k=1}^n \frac{1}{1+\alpha\lambda_k}P_{ik}P_{jk}=\sum_{k=1}^n \frac{1}{1+\alpha\lambda_k}(p_k)_i(p_k)_j
    \end{align*}
where $(p_k)_i$ denotes the $i^{th}$-entry of the eigenvector $p_k$. This implies that
\begin{align*}
    \sum_{j=1}^n (W^{-1})_{ij}^2=\sum_{j=1}^n\sum_{k=1}^n \frac{1}{1+\alpha\lambda_k}(p_k)_i(p_k)_j)^2=(W^{-2})_{ii}
\end{align*}
where $W^{-2}$ denotes the matrix $(W^{-1})^2$. The last equality follows from the fact that $W$ is symmetric. Now, note that $W^{-2}=P(D^{-1})^2P^T$ so that, $(W^{-2})_{ii}=\sum_{k=1}^n \frac{1}{(1+\alpha\lambda_k)^2} P_{ik}^2=\sum_{k=1}^n \frac{1}{(1+\alpha\lambda_k)^2} (p_k)_i^2$. 

Effectively, to bound $\sum_{i \in [n]} m_i^*$, we upper bound the quantity $(W^{-2})_{ii}$ which depends only on $\alpha$, eigenvalues and eigenvectors of the graph Laplacian $L$, summarized in the following Lemma.

\begin{lemma}
    $\sum_{i \in [n]} m_i^* \le \frac{nk}{\epsilon}\cd\max\limits_{i}\{(W^{-2})_{ii}^2\} = \frac{nk}{\epsilon}\cd\max\limits_{i} \{\sum_{k=1}^n \frac{1}{(\alpha \lambda_k + 1)}(p_k)_i^2\}$. Moreover, $m_i = m = \frac{k}{\epsilon}\cd\max\limits_{i} \{\sum_{k=1}^n \frac{1}{(\alpha \lambda_k + 1)}(p_k)_i^2\}$ for every $i \in [n]$ is a feasible solution to Equation \ref{eq:opt}.
\label{lem:eigen_lemma}
\end{lemma}

\begin{proof} 
At the end of the proof of the upper bound in Theorem \ref{thm: degree bound}, by setting $m_1 = m_2 = \cdots = m_n$, we show that $\sum_{i \in [n]} m_i^* \le \frac{nk}{\epsilon}\cd\max\limits_{i}\{\sum_{j=1}^n (W_{ij}^{-1})^2\}$. Thus $\sum_{i \in [n]} m_i^* \le \frac{nk}{\epsilon}\cd\max\limits_{i} \{\sum_{k=1}^n \frac{1}{(\alpha \lambda_k + 1)}(p_k)_i^2\}$. The solution is achieved by setting $m_i = \frac{k}{\epsilon}\cd\max\limits_{i} \{\sum_{k=1}^n \frac{1}{(\alpha \lambda_k + 1)}(p_k)_i^2\}$, which is then a solution to Equation \ref{eq:opt}.
\end{proof}


\begin{lemma}
 Consider a graph $G$ which is a $d$-dimensional hypercube where $d=\log n$, and $n$ is a power of $2$. Then, when $\alpha$ is a constant with $\alpha \geq \frac{3}{8}$,
$\sum_{i \in [n]} m_i^* = \Theta(\frac{nk}{d^2\epsilon})$. The network gain $\ngain(G, \pmb{v}, k, \epsilon) \ge \Omega\left(d^2\right)$. Assigning $\max\{ O \left(\lceil \frac{k}{d^2\epsilon}\rceil\right), k\}$ samples to every agent suffices to solve the TSC problem.
\label{lem:hypercube}
\end{lemma}


\begin{proof}
It is known that for the hypercube, it holds that for all eigenvectors $p_k$, and all $i\in [n]$, $(p_k)_i\in \{\frac{-1}{\sqrt{n}}, \frac{1}{\sqrt{n}}\}$ and eigenvalues of $L$ are given by $2i$ for $i\in \{0,1,...,d\}$ with multiplicity $\binom{d}{i}$ respectively \citep{spielman}. Based on this, we have that for any $i$,
\begin{align*}
(W^{-2})_{ii}&=\sum_{k=1}^n \frac{1}{(1+\alpha\lambda_k)^2} (p_k)_i^2 \\
            &= \frac{1}{n}\sum_{k=1}^n \frac{1}{(1+\alpha\lambda_k)^2} \\
            &= \frac{1}{n}\sum_{i=0}^d \binom{d}{i}\frac{1}{(2i\alpha+1)^2} \\
\end{align*}

We bound the quantity, $\sum_{i=0}^d \binom{d}{i}\frac{1}{(2i\alpha+1)^2}$. For this purpose, we look at the quantity $z(1+z^2)^d=\sum_{i=0}^d\binom{d}{i}z^{2i\alpha+1}$ where $z$ will be set later. Now,
\begin{align*}
    z(1+z^2)^d&=\sum_{i=0}^d \binom{d}{i}z^{2i+1}      &\iff  \\
  \int (1+z^2)^d z\, dz &= \sum_{i=0}^d\binom{d}{i}\int z^{2i} z\,dz    &\iff \\
  \frac{(1+z^2)^{d+1}}{2(d+1)} &= \sum_{i=0}^d\binom{d}{i} \frac{z^{2i+2}}{2i+2}  &\iff \\
  \frac{(1+z^2)^{d+1}}{2(d+1)}z &= \sum_{i=0}^d\binom{d}{i} \frac{z^{2i+2}}{2i+2}z  &\iff \\
  \int \frac{(1+z^2)^{d+1}}{2(d+1)}z \,dz &= \sum_{i=0}^d\binom{d}{i} \int \frac{z^{2i+2}}{2i+2} z \,dz &\iff \\
  \frac{(1+z^2)^{d+2}}{4(d+1)(d+2)} &= \sum_{i=0}^d\binom{d}{i}\frac{z^{2i+4}}{(2i+2)(2i+4)} \\
\end{align*}
Integrating over $[0,1]$ and setting $d=\log n$ we get that, 
\begin{align*}
\sum_{i=0}^d\binom{d}{i}\frac{1}{(2i+2)(2i+4)}=\frac{4n}{4(d+1)(d+2)}=\frac{n}{(d+1)(d+2)}.    
\end{align*}

Assuming $\alpha\geq \frac{3}{8}$, we get,
\begin{align*}
    \sum_{i=0}^d\binom{d}{i}\frac{1}{(2i\alpha+1)^2}&=    \sum_{i=0}^d\binom{d}{i}\frac{1}{4\alpha^2i^2+4\alpha i+1} \\
    &=\sum_{i=0}^d\binom{d}{i}\frac{8}{8(4\alpha^2i^2+4\alpha i+1)}\\
    &\leq \sum_{i=0}^d\binom{d}{i}\frac{8}{4i^2+12i+8} \\
    &=\sum_{i=0}^d\binom{d}{i}\frac{8}{(2i+2)(2i+4)} \\
    \end{align*}
    It follows that,
\begin{align*}
    \sum_{i=0}^d\binom{d}{i}\frac{1}{(2i\alpha+1)^2}
    \leq 8 \sum_{i=0}^d\binom{d}{i}\frac{1}{(2i+2)(2i+4)}=\frac{8n}{(d+1)(d+2)}
\end{align*}
so that, 
\begin{align*}
(W^{-2})_{ii} &= \frac{1}{n}\sum_{i=0}^d \binom{d}{i}\frac{1}{(2i\alpha+1)^2} \\
             &\leq \frac{1}{n}\frac{8n}{(d+1)(d+2)} \\
             &=O(\frac{1}{d^2}) \\
\end{align*}
Note that this holds for all $i \in [n]$. Thus, by Lemma \ref{lem:eigen_lemma}, we get that $m_i = m =
 O\left(\frac{k}{d^2\epsilon}\right)$ is a feasible solution to Equation \ref{eq:opt}, and $\sum_{i \in [n]} m_i^* \le O\left(\frac{nk}{d^2\epsilon}\right)$. On the other hand, from the lower bound in Theorem \ref{thm: degree bound}, we have $\sum_{i \in [n]} m_i^* \ge \Omega\left(\frac{n k}{d^2\epsilon}\right)$.
Since $d = \log{n} < \sqrt{n}$, we have $\sum_{i \in [n]} m_i^* = \Theta(\frac{nk}{d^2\epsilon})$. Since $m_i = m = O\left(\frac{k}{d^2\epsilon}\right)$ is a feasible solution to Equation \ref{eq:opt}. Assigning $\max\{ O \left(\lceil \frac{k}{d^2\epsilon}\rceil\right), k\}$ samples to every agent is sufficent to solve the TSC problem. By Lemma \ref{lem:general gain}, we have 
\begin{equation*}
    \ngain(G, \pmb{v}, k, \epsilon) \ge \frac{c_1 k}{ 2c_2 m\epsilon} - 1 = \Omega\left(d^2\right) .
\end{equation*}
\end{proof}

\begin{lemma}
For a random $d$-regular graph, with probability of at least $1- O(n^{-\tau})$ where $\tau = \lceil \frac{\sqrt{d-1} + 1}{2}\rceil - 1$

(1) If $d \ge \frac{\sqrt{n}}{\alpha}$,  $\sum_{i \in [n]} m_i^*  = \Theta(\frac{k}{\epsilon})$. Assigning $\max\{ O \left(\lceil \frac{k}{n\epsilon}\rceil\right), k\}$ samples to every agent $i$ suffices to solve the TSC problem. The network gain $\ngain(G, \pmb{v}, k, \epsilon) \ge \Omega\left(n\right)$.

(2) If $d < \frac{\sqrt{n}}{\alpha}$,  $\sum_{i \in [n]} m_i^* = \Theta(\frac{nk}{\alpha^2 d^2\epsilon})$. Assigning $\max\{ O \left(\lceil \frac{k}{\alpha^2 d^2 \epsilon}\rceil\right), k\}$ suffices to solve the TSC problem. The network gain $\ngain(G, \pmb{v}, k, \epsilon) \ge \Omega\left(d^2\right)$.
\label{lem:regular}
\end{lemma}

\begin{proof}
For a random $d$-regular graph, it is known that the largest eigenvalue of the adjacency matrix is $d$ and with the probability of at least $1- O(n^{-\tau})$ where $\tau = \lceil \frac{\sqrt{d-1} + 1}{2}\rceil - 1$, the rest of the eigenvalues are in the range $[-2\sqrt{d}-o(1), 2\sqrt{d}+o(1)]$\citep{friedman2003proof}. This implies that for the Laplacian the smallest eigenvalue is $0$ and the rest of the eigenvalues are $\Theta(d)$. Since in the Laplacian matrix, $(p_1)_i = \frac{1}{\sqrt{n}}$ for all $i\in[n]$. We have that, for every $i$, \begin{align*}
    (W^{-2})_{ii}&=\sum_{k=1}^n \frac{1}{(1+\alpha\lambda_k)^2} (p_k)_i^2 \\
                &=\frac{1}{n}+\frac{1}{\Omega(\alpha^2 d^2)}\sum_{k=2}^n (p_k)_i^2 \\ 
                &=\frac{1}{n}+\frac{1}{\Omega(\alpha^2 d^2)}(1-\frac{1}{n})
                \\&\leq \frac{1}{n}+\frac{1}{\Omega(\alpha^2 d^2)}
\end{align*}

From Lemma \ref{lem:eigen_lemma}, we have $\sum_{i \in [n]} m_i^*  \le \max\limits_{i}\{(W^{-2})_{ii}\}n\cd\frac{k}{\epsilon} \le (\frac{1}{n}+\frac{1}{\Omega(\alpha^2 d^2)})n\cd\frac{k}{\epsilon}$. Thus, when $d \ge \frac{\sqrt{n}}{\alpha}$,  $\sum_{i \in [n]} m_i^*  \le \frac{nk}{\Omega(n)\epsilon}= O(\frac{k}{\epsilon})$, and $m_i = m =O(\frac{k}{n\epsilon})$ is a feasible solution to Equation \ref{eq:opt}. Assigning $\max\{ O \left(\lceil \frac{k}{n\epsilon}\rceil\right), k\}$ samples to every agent suffices to solve the TSC problem. On the other hand, from the trivial lower bound, we have $\sum_{i \in [n]} m_i^*  \ge \frac{k}{\epsilon}$, implying $\sum_{i \in [n]} m_i^* = \Theta(\frac{k}{\epsilon})$. By Lemma \ref{lem:general gain}, we have 
\begin{equation*}
    \ngain(G, \pmb{v}, k, \epsilon) \ge \frac{c_1 k}{ 2c_2 m \epsilon} - 1 = \Omega\left(n\right) .
\end{equation*}

When $d< \frac{\sqrt{n}}{\alpha}$,  $\sum_{i \in [n]} m_i^*  \le \frac{nk}{\Omega(\alpha^2 d^2\epsilon)}= O(\frac{nk}{\alpha^2 d^2 \epsilon})$, and $m_i = m =  O(\frac{k}{\alpha^2 d^2 \epsilon})$ is a feasible solution to Equation \ref{eq:opt}. From the trivial bound and the lower bound of Theorem \ref{thm: degree bound}, we have $\sum_{i \in [n]} m_i^*  \ge \frac{nk}{(\alpha d + 1)^2\epsilon}$. Since $d< \frac{\sqrt{n}}{\alpha}$,  we have $\sum_{i \in [n]} m_i^*  = \Theta(\frac{nk}{\alpha^2 d^2 \epsilon})$. Assigning $\max\{ O \left(\lceil \frac{k}{\alpha^2 d^2 \epsilon}\rceil\right), k\}$ samples to every agent suffices to solve the TSC problem.
By Lemma \ref{lem:general gain}, we have 
\begin{equation*}
    \ngain(G, \pmb{v}, k, \epsilon) \ge \frac{c_1 k}{ 2c_2 m \epsilon} - 1 = \Omega\left(d^2\right) .
\end{equation*}
\end{proof}


\begin{lemma}
Consider a graph $G$ which is a $d$-regular expander with expansion $\tau(G)$. Then, we have that

(1) When $d \le \frac{\sqrt{n}}{\alpha \tau(G)^2}$, then $\sum_{i \in [n]} m_i^*   \le O(\frac{nk}{\alpha^2 d^2 \tau(G)^4\epsilon})$. Assigning $\max\{ O \left(\lceil \frac{k}{\alpha^2 d^2 \tau(G)^4 \epsilon} \rceil\right), k\}$ samples to every agent $i$ suffices to solve the TSC problem. The network gain $\ngain(G, \pmb{v}, k, \epsilon) \ge \Omega\left(d^2\tau(G)^4\right)$.


 (2) When $d > \frac{\sqrt{n}}{\alpha \tau(G)^2}$, then $\sum_{i \in [n]} m_i^*    \le O(\frac{k}{\epsilon})$. Assigning $\max\{ O \left(\lceil \frac{k}{n\epsilon} \rceil\right), k\}$ samples to every agent $i$ suffices to solve the TSC problem. The network gain $\ngain(G, \pmb{v}, k, \epsilon) \ge \Omega\left(n\right)$.
\label{lem:expander}
\end{lemma}

\begin{proof}
Since $\lambda_1=0$ for any graph, we know from Cheeger's inequality that the second smallest eigenvalue $\lambda_2$ satisfies, 
\begin{align*}
    \frac{\lambda_2}{2d} \leq \tau(G) \leq \sqrt{2\frac{\lambda_2}{d}}
\end{align*}
from which it follows that $\lambda_2\geq \frac{d\tau(G)^2}{2}$. Then, we have that
\begin{align*}
    (W^{-2})_{ii}&=\sum_{k=1}^n \frac{1}{(1+\alpha\lambda_k)^2} (p_k)_i^2 \\
                &\leq \frac{1}{n}+\frac{4}{\alpha^2d^2\tau(G)^4}\sum_{k=2}^n (p_k)_i^2 \\ 
                &\leq \frac{1}{n} +\frac{4}{\alpha^2d^2\tau(G)^4}\\
\end{align*}
Note that this holds for all $i \in [n]$. Thus, by Lemma \ref{lem:eigen_lemma}, $\sum_{i \in [n]} m_i^*  \le (\frac{1}{n} +\frac{4}{\alpha^2d^2\tau(G)^4})n\cd\frac{k}{\epsilon}$, and $m_i = m = (\frac{1}{n} +\frac{4}{\alpha^2d^2\tau(G)^4})\cd\frac{k}{\epsilon}$ is a feasible solution to Equation \ref{eq:opt}. Thus, when $d \ge \frac{\sqrt{n}}{\alpha \tau(G)^2}$, $\sum_{i \in [n]} m_i^*   \le O(\frac{k}{\epsilon})$. Since $m = O(\frac{k}{n\epsilon})$ now, assigning $\max\{ O \left(\lceil \frac{k}{n\epsilon} \rceil\right), k\}$ samples to every agent $i$ suffices to solve the TSC problem. The network gain $\ngain(G, \pmb{v}, k, \epsilon) \ge \Omega\left(n\right)$.

When $d < \frac{\sqrt{n}}{\alpha \tau(G)^2}$, $\sum_{i \in [n]} m_i^*   \le  O(\frac{nk}{\alpha^2 d^2 \tau(G)^4 \epsilon})$. Since $m = O \left( \frac{k}{\alpha^2 d^2 \tau(G)^4 \epsilon}\right)$ now, assigning $\max\{ O \left(\lceil \frac{k}{\alpha^2 d^2 \tau(G)^4 \epsilon}\rceil\right), k\}$ samples to every agent $i$ suffices to solve the TSC problem. The network gain $\ngain(G, \pmb{v}, k, \epsilon) \ge \Omega\left(d^2 \tau(G)^4\right)$.
\end{proof}

%% file: app-3.2-proof.tex
\section{Missing proof in Section \ref{sec:general}}
\label{app:general proof}
To simplify notations, assume $a_{ij} = (W^{-1})_{ij}$ for all $i,j \in [n]$. Since $W^{-1}$ is positive-definite, from the Schur product theorem, $(W^{-1}) \circ (W^{-1})$ is also positive-definite where $\circ$ is the Hadamard product (element-wise product). Let $B = ((W^{-1}) \circ (W^{-1}))^{-1}$ and $[B]_{ij} = b_{ij}$.

\begin{lemma} \label{b>0}
 $\sum\limits_{j=1}^n b_{ij} > 0$ for all $i \in [n]$
\end{lemma}

\begin{proof}
Let $W_2 = W^{-1} \circ W^{-1}$, we have $W_2^{-1} = B$. Let $W_2^*, (W^{-1})^*$ be the adjoint matrix of $W_2, W^{-1}$ respectively. We have $B = W_2^{-1} = \frac{1}{\det(W_2)}W_2^*$. Consider the sum of the $n^{th}$ row of $B$, to prove it is positive, since $\det(W_2)$ is positive, we only need to prove the sum of the $n^{th}$ row of $W_2^*$ is positive. The $n^{th}$-row sum of $W_2^*$ is $\det(W_2(n))$ where $W_2(n)$ is generated by replacing the $n^{th}$ row of $W_2$ to all one vector. Similarly, we can define $W^{-1}(n)= \begin{bmatrix}
    a_{11} & \cdots & a_{1n}
\\ \vdots & \ddots & \vdots
\\ a_{(n-1)1} & \cdots & a_{(n-1)n}
\\ 1 & \cdots & 1
\end{bmatrix}$. Since $W^{-1}$ is positive definite, the first $(n-1)$ leading principal minors of $W^{-1}(n) $, which are exactly the same as $W^{-1}$, are positive. On the other hand, $\det(W^{-1}(n)) = \sum_{j=1}^n (W^{-1})^*_{nj} = \det(W^{-1})\sum_{j=1}^n W_{nj} = \det(W^{-1}) > 0 $. Thus, $W^{-1}(n)$ is positive definite. Note that $W_2(n) = W^{-1}(n) \circ W^{-1}(n)$, from Schur product theorem, $W_2(n)$ is positive definite. Thus $\det(W_2(n)) > 0$, which means $\sum_{j=1}^n b_{nj} > 0$. By relabeling each node to $n$, we can get $\sum\limits_{j=1}^n b_{ij} > 0$ for all $i \in [n]$.
\end{proof}

We first prove Lemma \ref{lem:closed-form} which will guide us to derive the Theorem \ref{thm: general tight bound}.

\noindent \textbf{Lemma \ref{lem:closed-form}. } \emph{When  $\sum_{j=1}^n \frac{b_{ij}}{(\sum_{k=1}^n b_{jk})^2} \ge 0$ for all $i \in [n]$, the closed-form solution of Equation \ref{eq:opt} is $\sum_{i=1}^n\frac{1}{\sum_{j=1}^n b_{ij}}\cd\frac{k}{\epsilon}$ when $\lambda_i = \sum_{j=1}^n \frac{b_{ij}}{(\sum_{k=1}^n b_{jk})^2}\cd\frac{k^2}{\epsilon^2}$ and $m_i = \frac{1}{\sum_{j=1}^n b_{ij}}\cd\frac{k}{\epsilon}$ for all $i \in [n]$}

\begin{proof}
From Lemma \ref{lem:dual}, strong duality holds. Thus, we can only consider the solution of dual form.
\begin{align*}
    \max_{\lambda_1, \cdots, \lambda_n} \quad &2\sum_{i=1}^n \sqrt{\sum_{j=1}^n \lambda_j a_{ij}^2} - \frac{\epsilon}{k}\sum_{i=1}^n \lambda_i\\
    \textrm{s.t.} \quad & \lambda_i \ge 0, \forall i
\end{align*}

Let $g(\lambda)$ be the objective function of Equation \ref{eq:dual-eq} Setting gradient to zero, we have $\frac{\partial g}{\partial \lambda_k} = \sum_{i=1}^n \frac{a_{ik}^2}{\sqrt{\sum_{j=1}^n \lambda_j a_{ij}^2}} - \frac{\epsilon}{k} =  \sum_{i=1}^n \frac{a_{ki}^2}{\sqrt{\sum_{j=1}^n \lambda_j a_{ij}^2}} - \frac{\epsilon}{k} = 0$ for any $i$. Let $\frac{1}{\sqrt{\sum_{j=1}^n \lambda_j a_{ij}^2}} = p(i)$, we have 
$\begin{bmatrix}
a_{11}^2 &\cdots& a_{1n}^2
\\ a_{21}^2 &\cdots &a_{2n}^2
\\ \vdots & \ddots& \vdots
\\ a_{n1}^2 & \cdots & a_{nn}^2
\end{bmatrix}
\begin{bmatrix}
p(1) \\ \vdots \\ p(n)
\end{bmatrix}
= \begin{bmatrix}
\frac{\epsilon}{k} \\ \vdots \\ \frac{\epsilon}{k}
\end{bmatrix}$. The matrix here is equivalent to $(W^{-1}) \circ (W^{-1})$ where $\circ$ denotes the Hadamard product of two matrices. Since $W$ is positive definite, $W^{-1}$ is also positive definite. From Schur product theorem, $(W^{-1}) \circ (W^{-1})$ is positive definite and then is invertible. Since $B = ((W^{-1}) \circ (W^{-1}))^{-1}$ and $[B]_{ij} = b_{ij}$, we have $p(i) = \frac{1}{M(\epsilon)} \sum_{j=1}^n b_{ij}$. In Lemma \ref{b>0}, we have shown that $\sum_{j=1}^n b_{ij} > 0$, thus, this solution is valid. Considering $\lambda_i$, we have $\sum_{j=1}^n \lambda_j a_{ij}^2 = \frac{1}{p^2(i)}$ for any $i$. Thus 
$\begin{bmatrix}
a_{11}^2 &\cdots& a_{1n}^2
\\ a_{21}^2 &\cdots &a_{2n}^2
\\ \vdots & \ddots& \vdots
\\ a_{n1}^2 & \cdots & a_{nn}^2
\end{bmatrix}
\begin{bmatrix}
\lambda_1 \\ \vdots \\ \lambda_n
\end{bmatrix}
= \begin{bmatrix}
\frac{1}{p^2(1)} \\ \vdots \\ \frac{1}{p^2(n)}
\end{bmatrix}$. We have $\lambda_i = \sum_{j =1}^n \frac{b_{ij}}{p^2(j)} = \sum_{j=1}^n \frac{b_{ij}}{(\sum_{k=1}^n b_{jk})^2}\cd\frac{k^2}{\epsilon^2}$ (Note that it can be negative). This solution is valid if and only if $\sum_{j=1}^n \frac{b_{ij}}{(\sum_{k=1}^n b_{jk})^2} \ge 0$ for any $i \in [n]$. If this solution is valid, the optimal solution is:
\begin{align*}
&2\sum_{i=1}^n \sqrt{\sum_{j=1}^n \lambda_j a_{ij}^2} - \frac{\epsilon}{k}\sum_{i=1}^n \lambda_i
\\&= 2\sum_{i=1}^n \frac{1}{p(i)} - \frac{\epsilon}{k}\sum_{i=1}^n\sum_{j=1}^n \frac{b_{ij}}{(\sum_{k=1}^n b_{jk})^2}\cd\frac{k^2}{\epsilon^2}
\\&= \frac{2k}{\epsilon}\sum_{i=1}^n\frac{1}{\sum_{j=1}^n b_{ij}} - \frac{k}{\epsilon} \sum_{j=1}^n \frac{\sum_{i=1}^nb_{ij}}{(\sum_{k=1}^n b_{kj})^2}
\\&= \frac{2k}{\epsilon}\sum_{i=1}^n\frac{1}{\sum_{j=1}^n b_{ij}} - \frac{k}{\epsilon} \sum_{i=1}^n \frac{1}{\sum_{j=1}^n b_{ij}}
\\&= \frac{k}{\epsilon}\sum_{i=1}^n\frac{1}{\sum_{j=1}^n b_{ij}}
\end{align*}
From KKT conditions, setting the gradient of the Lagrangian function of (P) to zero, we have 
\begin{align*}
m_i &= \sqrt{\sum_{k=1}^n \lambda_k a_{ik}^2}
\\&= \frac{k}{\epsilon} \sqrt{\sum_{k=1}^n \sum_{t=1}^n\frac{b_{kt}}{(\sum_{j=1}^n b_{tj})^2} a_{ik}^2}
\\&= \frac{k}{\epsilon} \sqrt{\sum_{t=1}^n \frac{1}{(\sum_{j=1}^n b_{tj})^2} \sum_{k=1}^n a_{ik}^2b_{kt}}
\\&= \frac{k}{\epsilon} \sqrt{\sum_{t=1}^n \frac{1}{(\sum_{j=1}^n b_{tj})^2}  I_{it}}
\\&= \frac{1}{\sum_{j=1}^n b_{ij}}\cd\frac{k}{\epsilon}
\end{align*}
where $I_{it}$ is the element on the $i$-th row and $t$-th column of identity matrix $I$. The fourth equality comes from $B = (W^{-1} \circ W^{-1})^{-1}$ where the $ (W^{-1} \circ W^{-1})_{ij} = a_{ij}^2$.
\end{proof}

\noindent \textbf{Theorem \ref{thm: general tight bound}. }  Let $\gamma_i = \max\{0, \sum_{j=1}^n \frac{b_{ij}}{(\sum_{k=1}^n b_{jk})^2}\}$ for every $i \in [n]$, then we have\begin{equation*}
\max\{2\sum_{i=1}^n \sqrt{\sum_{j=1}^n \gamma_j (W^{-1}_{ij})^2} - \sum_{i=1}^n \gamma_i, 1\}\cd\frac{k}{\epsilon} \le \sum_{i \in [n]} m_i^* \le \min\{\sum_{i=1}^n\frac{1}{\sum_{j=1}^n b_{ij}}, n \} \cd\frac{k}{\epsilon}\end{equation*}

Assigning $\frac{1}{\sum_{j=1}^n b_{ij}}\cd\frac{k}{\epsilon}$ samples to node $i$ suffices to solve the TSC problem.

\begin{proof}
(1) $\sum_{i \in [n]} m_i^*  \le \frac{nk}{\epsilon}$ is trivial because $m_i = \frac{k}{\epsilon}$ for any $i$ is a feasible solution for the Equation \ref{eq:opt} because $\sum_{j=1}^n(W_{ij}^{-1})^2 \le 1$ for any $i \in [n]$

(2) Let $\bar{m} = (\frac{1}{m_1}, \cdots, \frac{1}{m_n})^\top$, $\bar{M} = (\frac{\epsilon}{k}, \cdots, \frac{\epsilon}{k})^\top$. In Equation \ref{eq:opt}, setting the first constraints to equality for every $i \in [n]$, we have $(W^{-1} \circ W^{-1})\bar{m} = \bar{M}$. From Lemma \ref{b>0}, the solution $\Bar{m} = B\bar{M}$ is feasible ($m_i > 0, \forall i$). Now $\frac{1}{m_i}  = \frac{\epsilon\sum_{j=1}^n b_{ij}}{k}$, $m_i = \frac{1}{\sum_{j=1}^n b_{ij}}\cd\frac{k}{\epsilon}$, and $\sum_{i=1}^n m_i = \sum_{i=1}^n\frac{1}{\sum_{j=1}^n b_{ij}}\cd\frac{k}{\epsilon}$. Since $\bar{m}$ is in the feasible region, $\sum_{i \in [n]} m_i^*  \le \sum_{i=1}^n\frac{1}{\sum_{j=1}^n b_{ij}}\cd\frac{k}{\epsilon}$. Thus,  $\sum_{i \in [n]} m_i^*  \le \sum_{i=1}^n\frac{1}{\sum_{j=1}^n b_{ij}}\cd\frac{k}{\epsilon}$. Assigning $\frac{1}{\sum_{j=1}^n b_{ij}}\cd\frac{k}{\epsilon}$ samples to node $i$ suffies to solve the TSC problem.

(3) In the dual form Equation \ref{eq:dual-eq}, $\lambda_i = \frac{k^2}{n\epsilon^2}$ for any $i$ is a feasible solution. From strong duality, we have
\begin{align*}
   \sum_{i \in [n]} m_i^* &\ge \frac{2k}{\epsilon}\sum_{i=1}^n\sqrt{\frac{1}{n}\sum_{j=1}^n(W_{ij}^{-1})^2} - \frac{k}{\epsilon}
    \\&\ge \frac{2k}{\epsilon}\sum_{i=1}^n\sqrt{\frac{1}{n}
    \left(\frac{\sum_{j=1}^nW_{ij}^{-1}}{\sqrt{n}}\right)^2} - \frac{k}{\epsilon} \tag{Cauchy-Schwarz}
    \\&\ge \frac{k}{\epsilon}\tag{$\sum_{j=1}^nW_{ij}^{-1}$ for all $i$}
\end{align*}

(4) The other lower bound is inspired by Lemma \ref{lem:closed-form} and truncate all $\lambda_i$ to
$\max\{0, \sum_{j=1}^n \frac{b_{ij}}{(\sum_{k=1}^n b_{jk})^2}\}$ to ensure they are in the feasible set of the dual form Equation \ref{eq:dual-eq}.
\end{proof}

%% file: app-exp-proof.tex
\section{Missing proof in Section \ref{sec:experiment}}
\label{app:exp-setup}
\begin{lemma}
The optimization Equation \ref{eq:opt} is equivalent to the following second-order cone programming.    

\begin{align*}
    \min_{t_1, \cdots, t_n, z_1, \cdots, z_n} \quad &\sum_{i=1}^n z_i\\
    \textrm{s.t.} \quad &  \sum_{j=1}^n  a_{ij}^2 t_i \le \frac{\epsilon}{k}, \quad \forall i\\
    &||[t_i - z_i , 2]||_2 \le t_i + z_i,  \quad \forall i
    \\&t_i > 0, \quad \forall i
\end{align*}
The optimal solution $(z_1^*, \cdots, z_n^*)$ is equal to the optimal solution $(m_1^*, \cdots, m_n^*)$ in the optimization (P).
\end{lemma}

\begin{proof}
The original optimization Equation \ref{eq:opt} is :
\begin{align*}
    \min_{m_1, \cdots, m_n} \quad &\sum_{i=1}^n m_i\\
    \textrm{s.t.} \quad &  \sum_{j=1}^n  \frac{a_{ij}^2}{m_j} \le \frac{\epsilon}{k}, \quad \forall i\\
    &m_i > 0, \quad \forall i
\end{align*}
Let $t_i = \frac{1}{m_i}$ and $z_i \ge \frac{1}{t_i}$, the programming is equivalent to
\begin{align*}
    \min_{t_1, \cdots, t_n, z_1, \cdots, z_n} \quad &\sum_{i=1}^n z_i\\
    \textrm{s.t.} \quad &  \sum_{j=1}^n  a_{ij}^2 t_i \le \frac{\epsilon}{k}, \quad \forall i \\
    &z_it_i \ge 1, \quad \forall i\\
    &t_i > 0, \quad \forall i
\end{align*}
The constraint $z_it_i \ge 1$ is equivalent to $||[t_i - z_i , 2]||_2 \le t_i + z_i$. Thus, this programming is equivalent to the following second-order cone programming.
\begin{align*}
    \min_{t_1, \cdots, t_n, z_1, \cdots, z_n} \quad &\sum_{i=1}^n z_i\\
    \textrm{s.t.} \quad &  \sum_{j=1}^n  a_{ij}^2 t_i \le \frac{\epsilon}{k}, \quad \forall i\\
    &||[t_i - z_i , 2]||_2 \le t_i + z_i,  \quad \forall i
    \\&t_i > 0, \quad \forall i
\end{align*}
\end{proof}

%% file: app-more-exp.tex
\section{More Experiments}
\label{app:more_experiments}
In this section, we will give more detailed experiments.  Our code can be found in \url{https://github.com/liuhl2000/Sample-Complexity-of-Opinion-Formation-on-Networks}. These experiments are done with seven 8-core Intel CPUs. The optimization is solved by running the solver until the convergence error is less than $1e-6$. All optimizations for our chosen networks could be finished in 150 iterations. The most time-consuming part is generating random weights 50 times and running corresponding experiments in Section \ref{subsec:tightness bounds}. Actually, there is little difference among these random experiments and in order to reproduce our experiments quickly, you can simply set the ``Iter num'' parameter to $1$ in our code.

\subsection{Distribution of Samples at the Optimal Point}
To study the distribution of samples, we assume all $v_{ij} = \alpha \ge 0$. In our experiments, we choose $\alpha = 0.01, 0.1, 1, 5, 10$. For figures in this section, the x-axis is node degree and the y-axis is the average or variance of $N^d = \{\frac{\epsilon m_i^*}{k}: d_i = d, i \in [n]\}$ where $m_i^*$ is the sample assigned to agent $i$ at the optimal point of Equation \ref{eq:opt}.

We generate scale-free networks using Barabasi–Albert preferential attachment \cite{barabasi1999emergence} with attachment parameter $2$ (Incoming node will attach to two existing nodes) for $n = 100, 200, 400, 600, 800, 1000$. Figure \ref{ALL_PL_S} and Figure \ref{ALL_PL_V} show the average sample number and variance for different $\alpha$.

\begin{figure}[htbp]
    \centering
    \subfigure[$\alpha = 0.01$]{ \includegraphics[width = 0.188\linewidth]{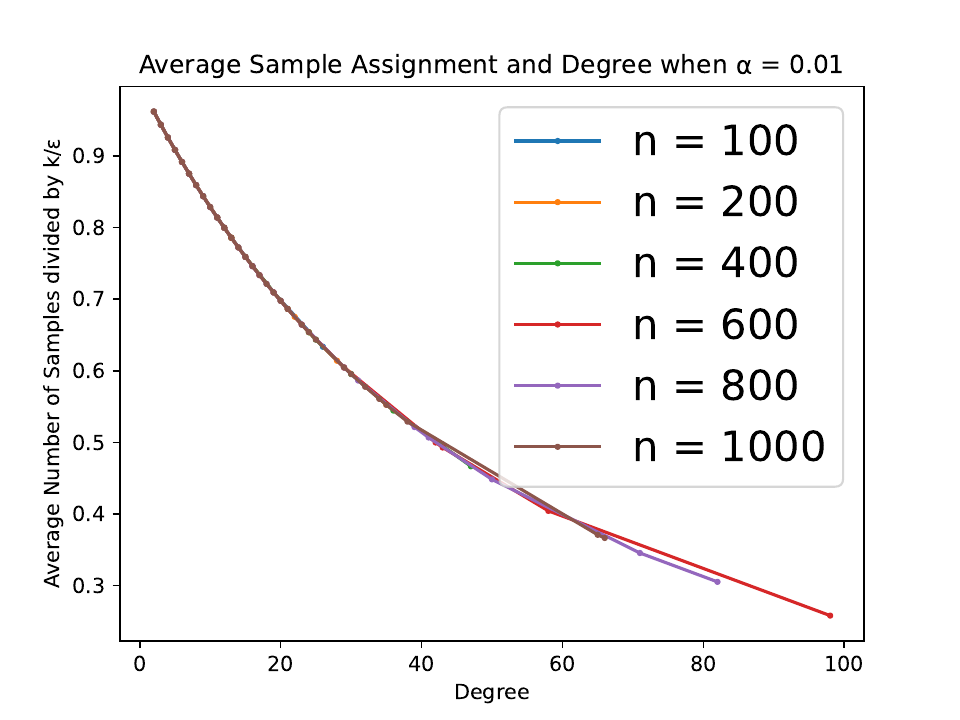}}
    \subfigure[$\alpha = 0.1$]{  \includegraphics[width = 0.188\linewidth]{figures/Degree_and_TSC/PL/sample_0.1.pdf}}
    \subfigure[$\alpha = 1$]{  \includegraphics[width = 0.188\linewidth]{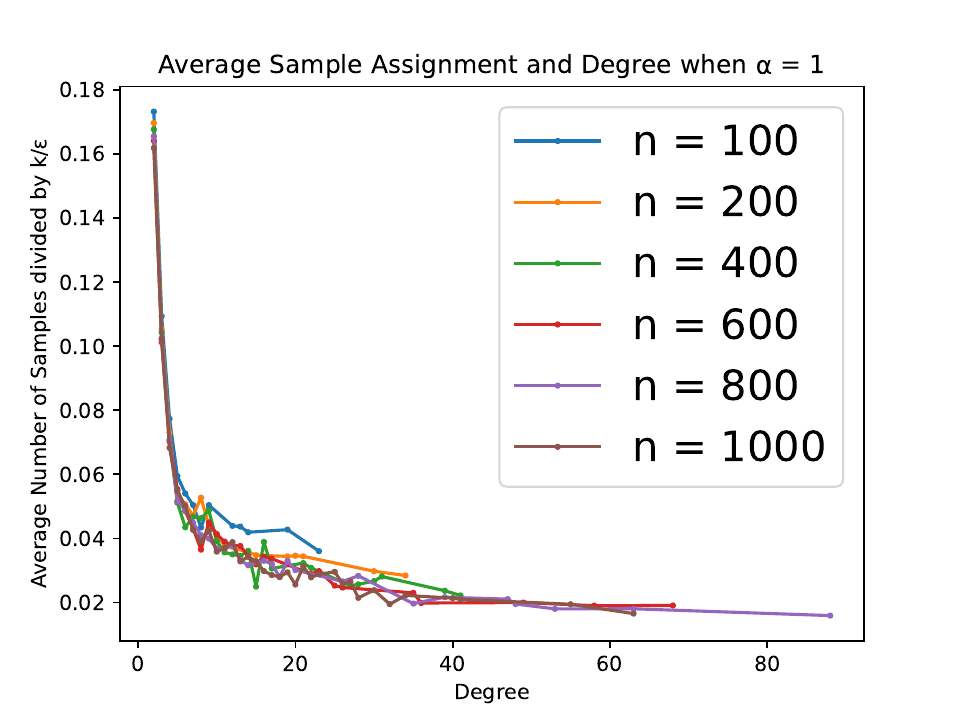}}
    \subfigure[$\alpha = 5$]{ \includegraphics[width = 0.188\linewidth]{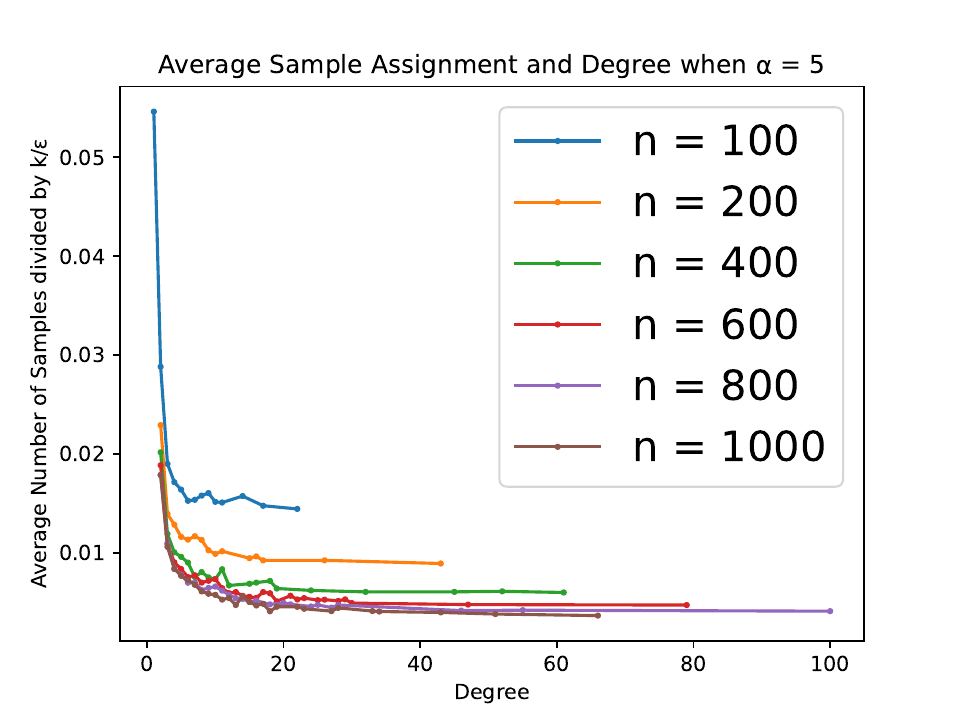}}
    \subfigure[$\alpha = 10$]{ \includegraphics[width = 0.188\linewidth]{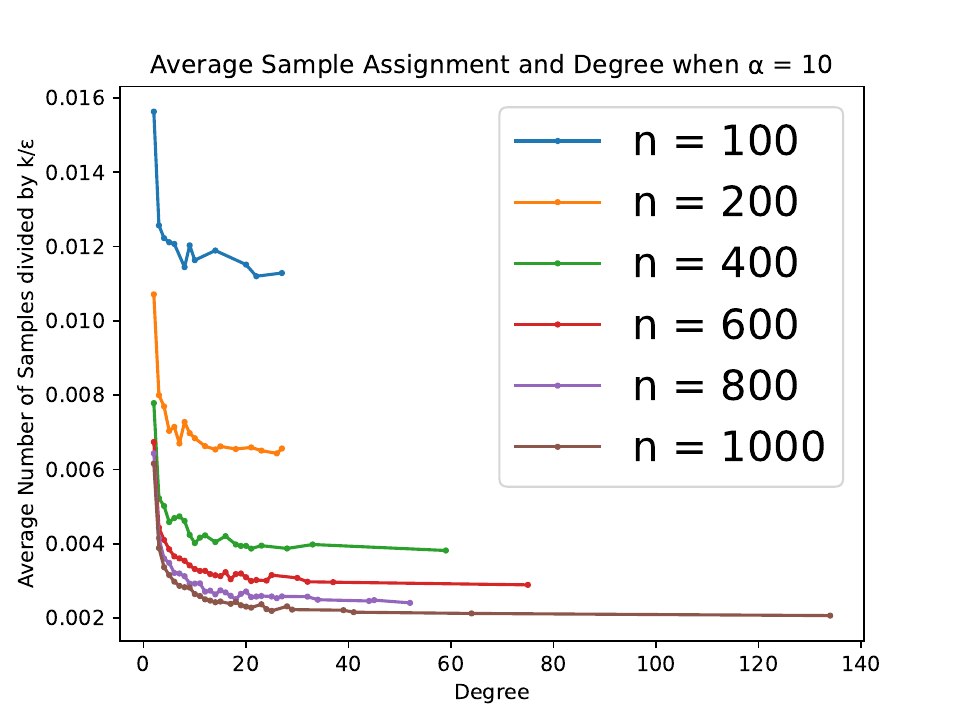}}
    \caption{\small
Relationship between average samples and degree of SF for different $\alpha$.}
    \label{ALL_PL_S}
\end{figure}

\begin{figure}[htbp]
    \centering
    \subfigure[$\alpha = 0.01$]{ \includegraphics[width = 0.188\linewidth]{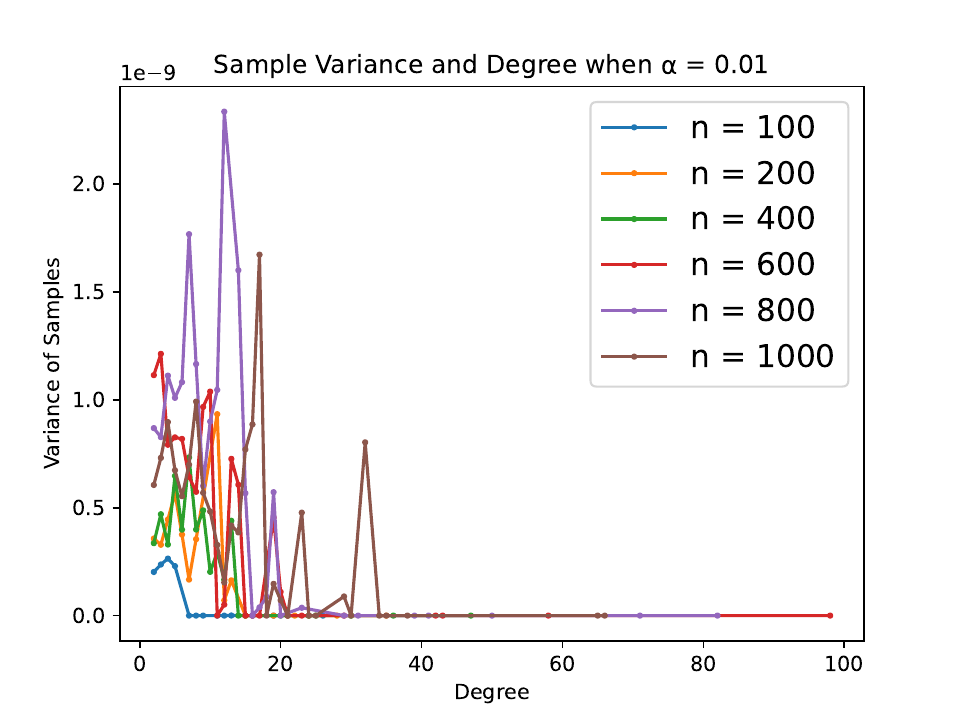}}
    \subfigure[$\alpha = 0.1$]{  \includegraphics[width = 0.188\linewidth]{figures/Degree_and_TSC/PL/variance_0.1.pdf}}
    \subfigure[$\alpha = 1$]{  \includegraphics[width = 0.188\linewidth]{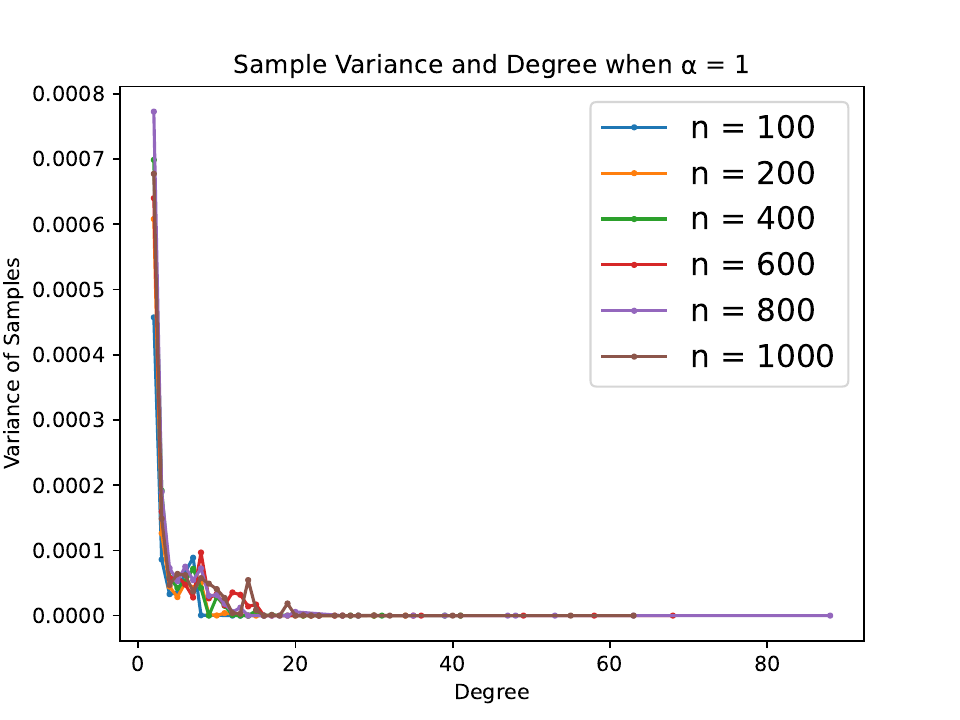}}
    \subfigure[$\alpha = 5$]{ \includegraphics[width = 0.188\linewidth]{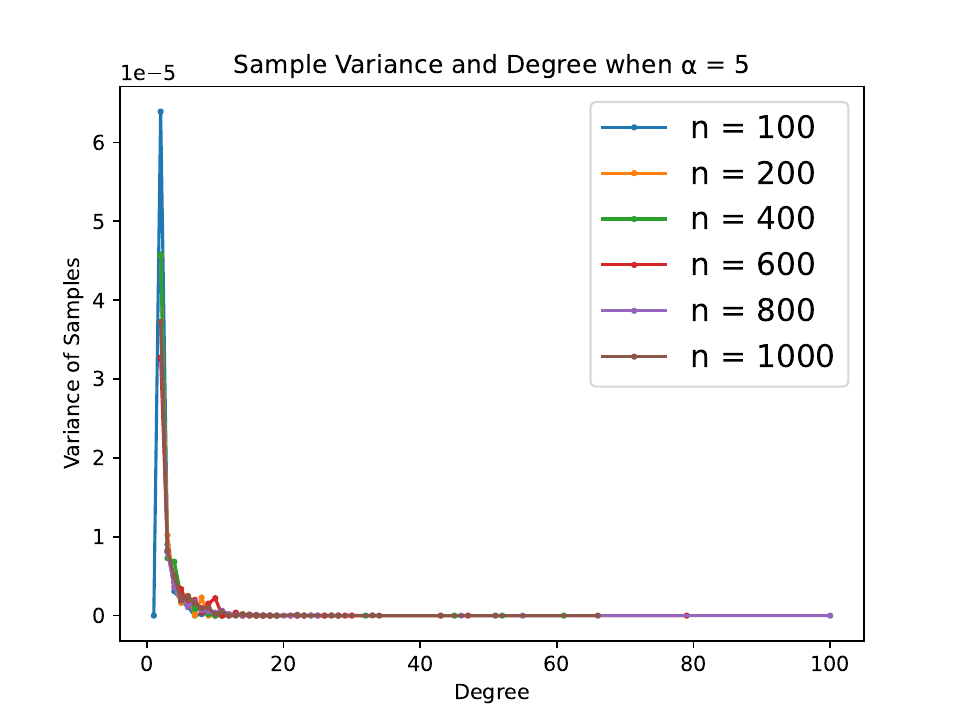}}
    \subfigure[$\alpha = 10$]{ \includegraphics[width = 0.188\linewidth]{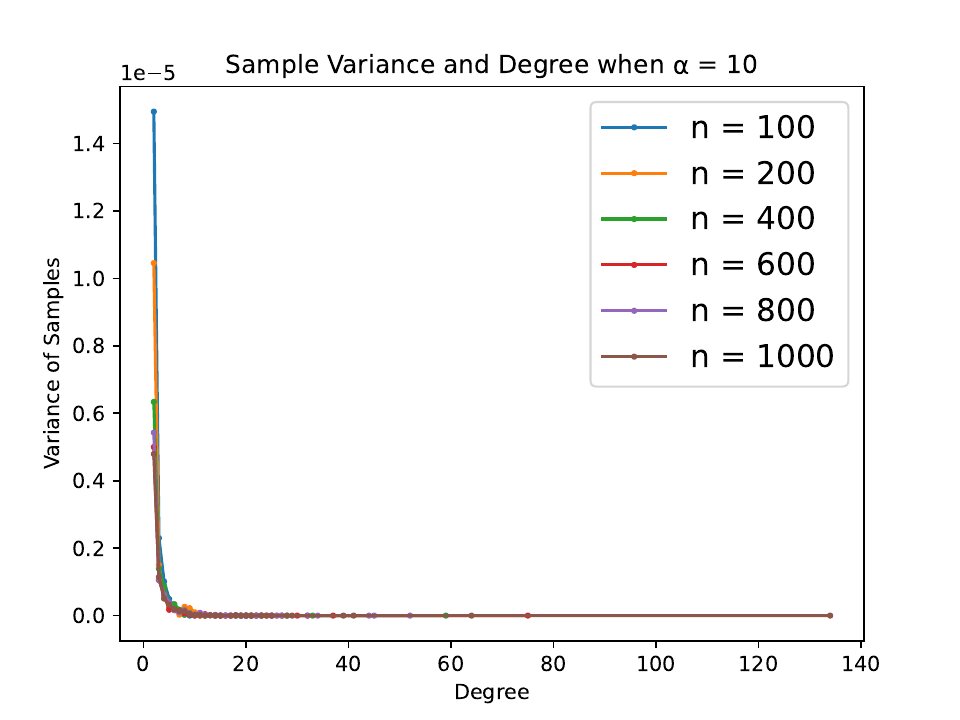}}
    \caption{Relationship between sample variance and degree of SF for different $\alpha$.}
    \label{ALL_PL_V}
\end{figure}

We use Erdos-Renyi random networks with probability $0.3$(Two nodes has the probability $0.3$ to have one edge) for $n = 600, 620, 650, 680, 700$. Figure \ref{ALL_ER_S} and Figure \ref{ALL_ER_V} show the average sample number and variance for different $\alpha$. 

\begin{figure}[htbp]
    \centering
    \subfigure[$\alpha = 0.01$]{ \includegraphics[width = 0.188\linewidth]{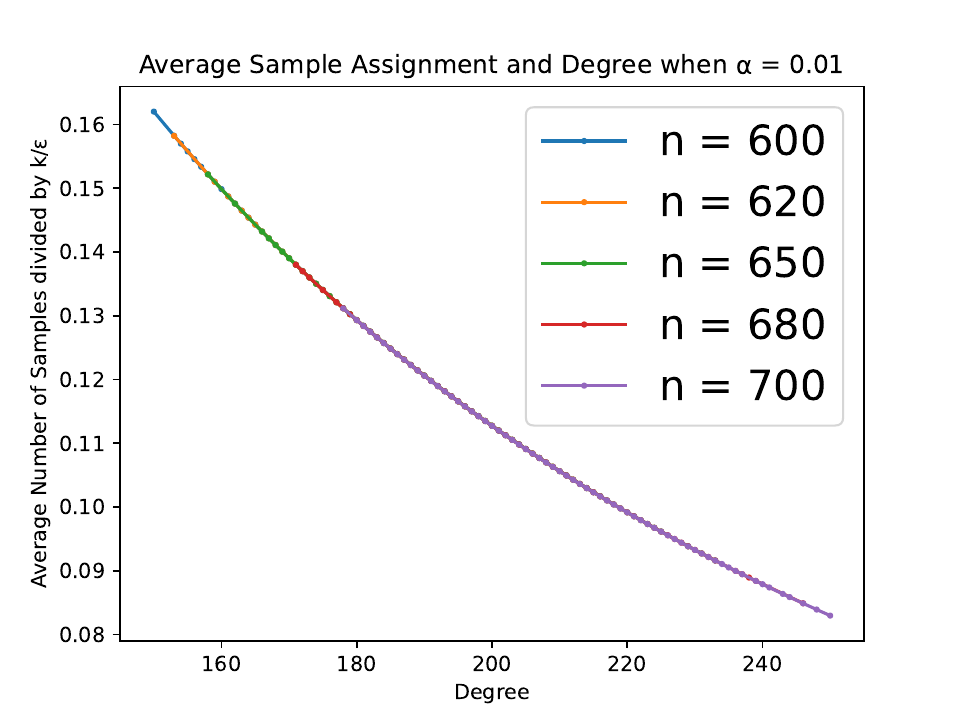}}
    \subfigure[$\alpha = 0.1$]{  \includegraphics[width = 0.188\linewidth]{figures/Degree_and_TSC/ER/sample_0.1.pdf}}
    \subfigure[$\alpha = 1$]{  \includegraphics[width = 0.188\linewidth]{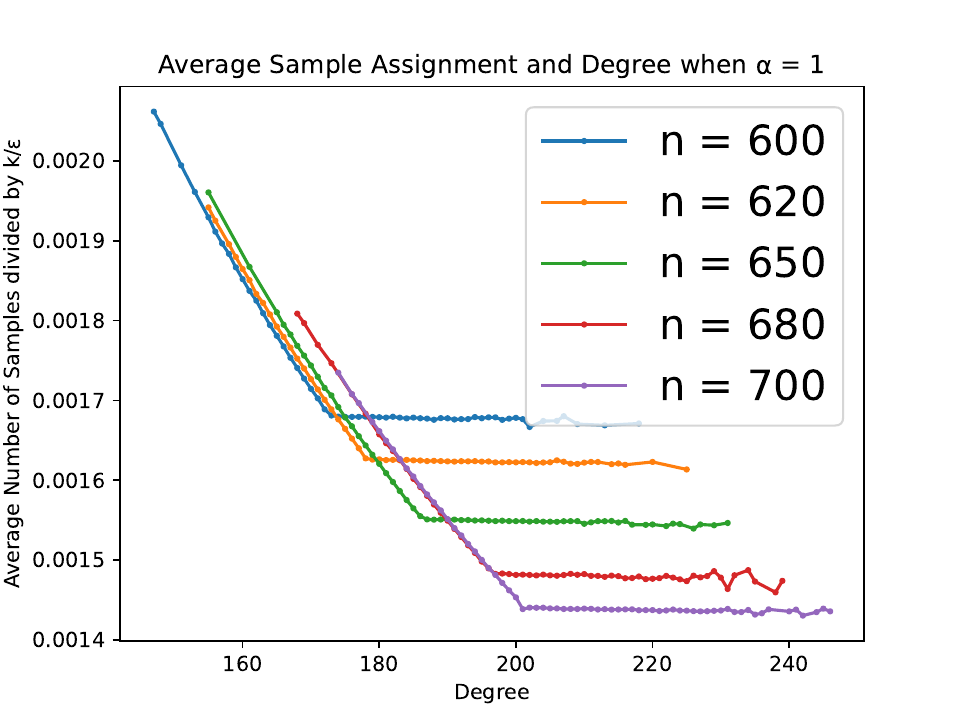}}
    \subfigure[$\alpha = 5$]{ \includegraphics[width = 0.188\linewidth]{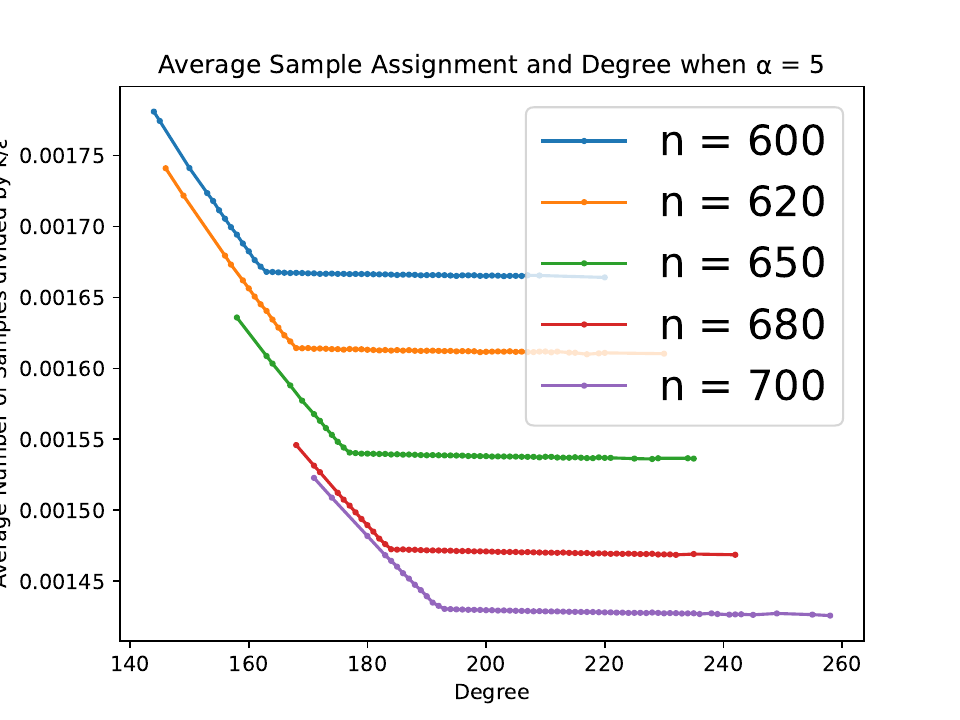}}
    \subfigure[$\alpha = 10$]{ \includegraphics[width = 0.188\linewidth]{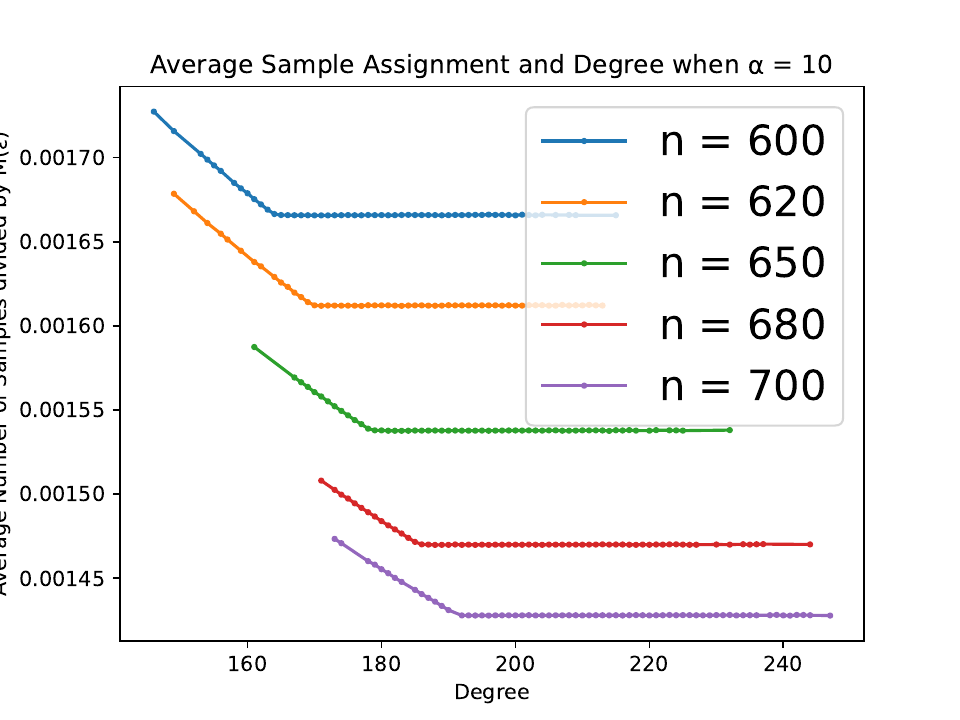}}
    \caption{Relationship between average samples and degree of ER for different $\alpha$.}
    \label{ALL_ER_S}
\end{figure}

\begin{figure}[htbp]
    \centering
    \subfigure[$\alpha = 0.01$]{ \includegraphics[width = 0.188\linewidth]{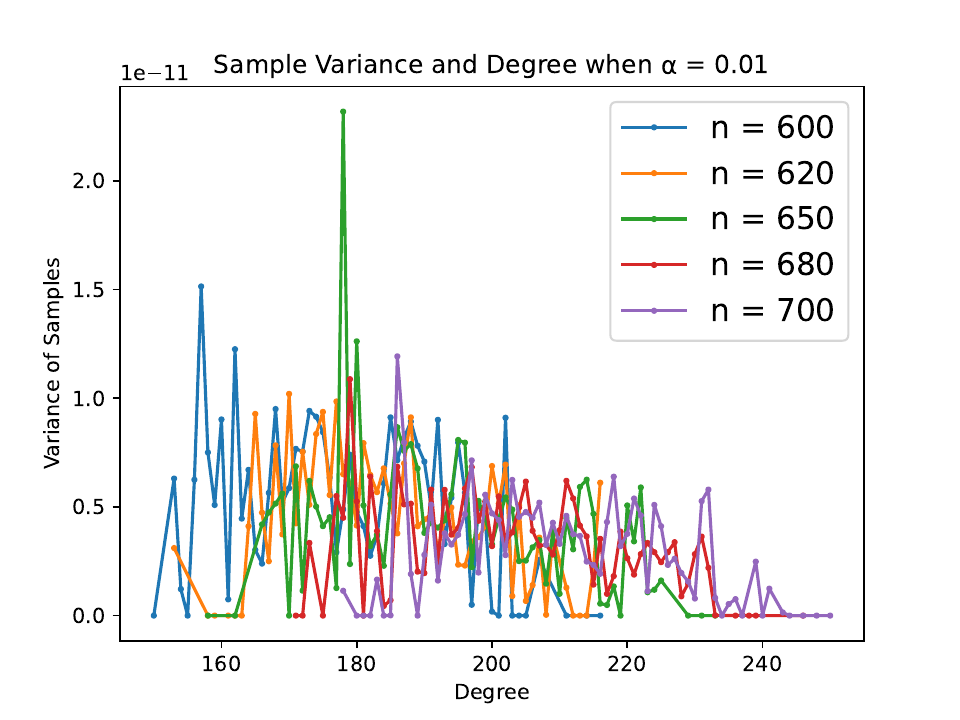}}
    \subfigure[$\alpha = 0.1$]{  \includegraphics[width = 0.188\linewidth]{figures/Degree_and_TSC/ER/variance_0.1.pdf}}
    \subfigure[$\alpha = 1$]{  \includegraphics[width = 0.188\linewidth]{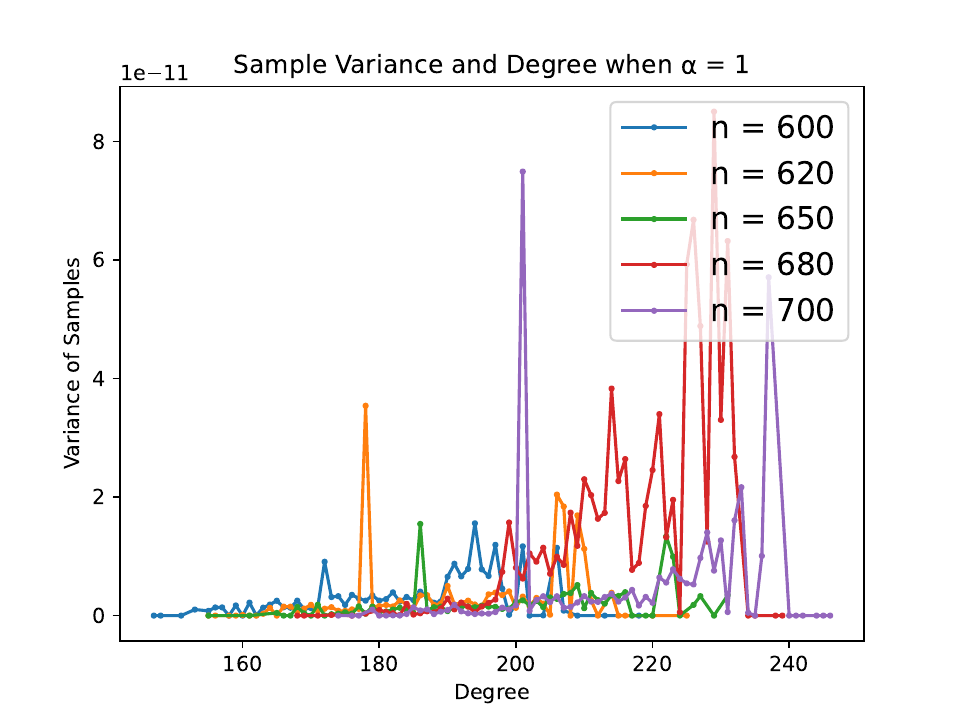}}
    \subfigure[$\alpha = 5$]{ \includegraphics[width = 0.188\linewidth]{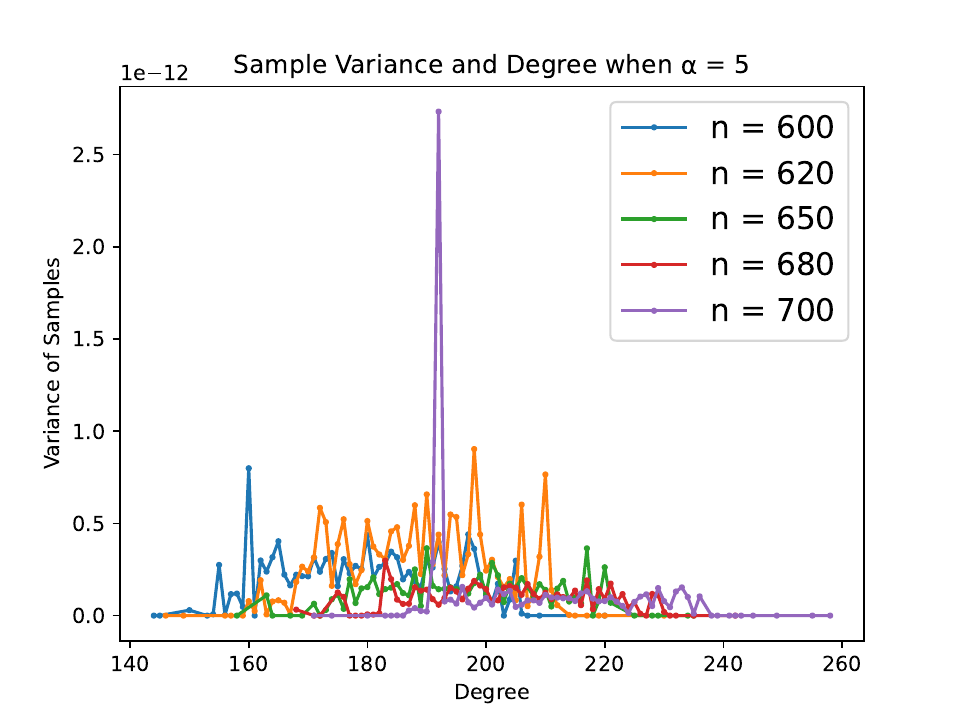}}
    \subfigure[$\alpha = 10$]{ \includegraphics[width = 0.188\linewidth]{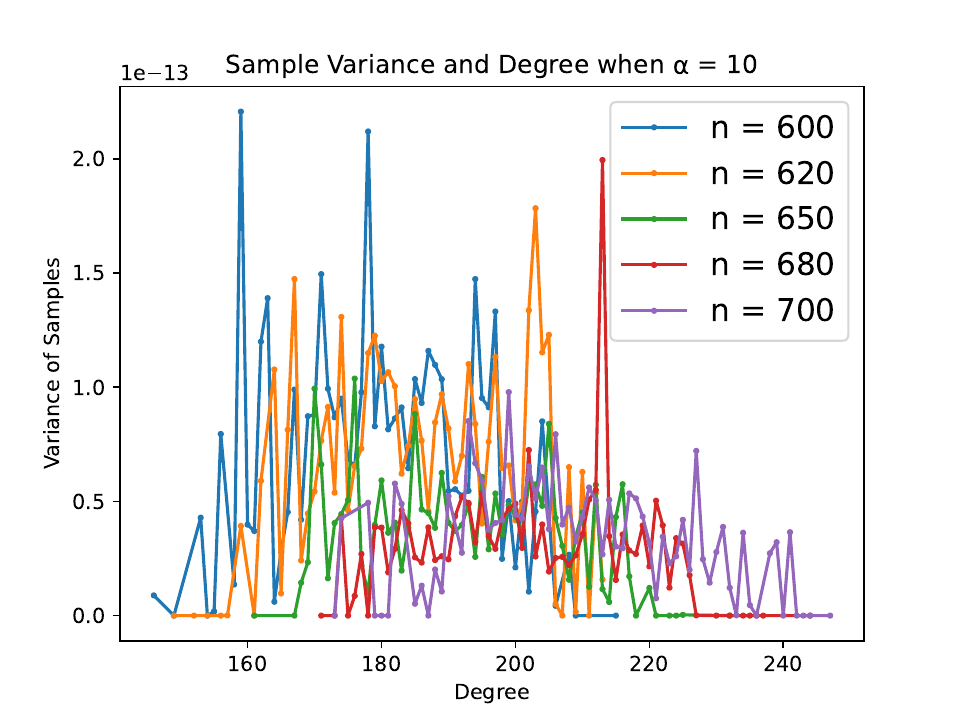}}
    \caption{Relationship between sample variance and degree of ER for different $\alpha$.}
    \label{ALL_ER_V}
\end{figure}

For real-world networks, Figure \ref{ALL_RN_S}, \ref{ALL_RN8_S}, and Figure \ref{ALL_RN_V} show the average sample number and variance for different $\alpha$. 
\begin{figure}[htbp]
    \centering
    \subfigure[$\alpha = 0.01$]{ \includegraphics[width = 0.188\linewidth]{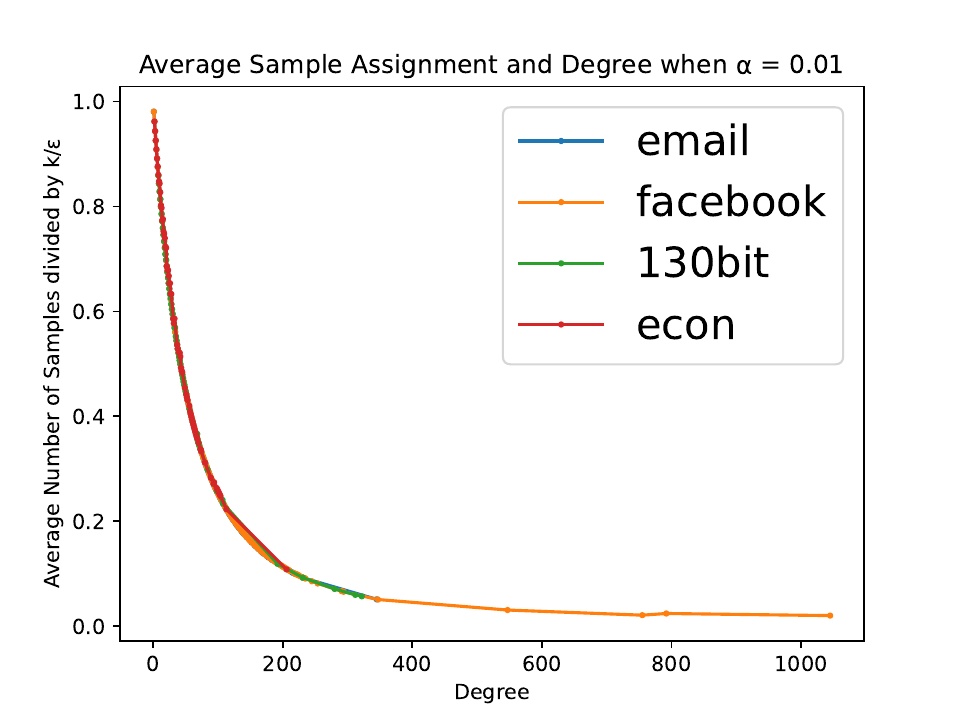}}
    \subfigure[$\alpha = 0.1$]{  \includegraphics[width = 0.188\linewidth]{figures/Degree_and_TSC/Real/avg_0.1.pdf}}
    \subfigure[$\alpha = 1$]{  \includegraphics[width = 0.188\linewidth]{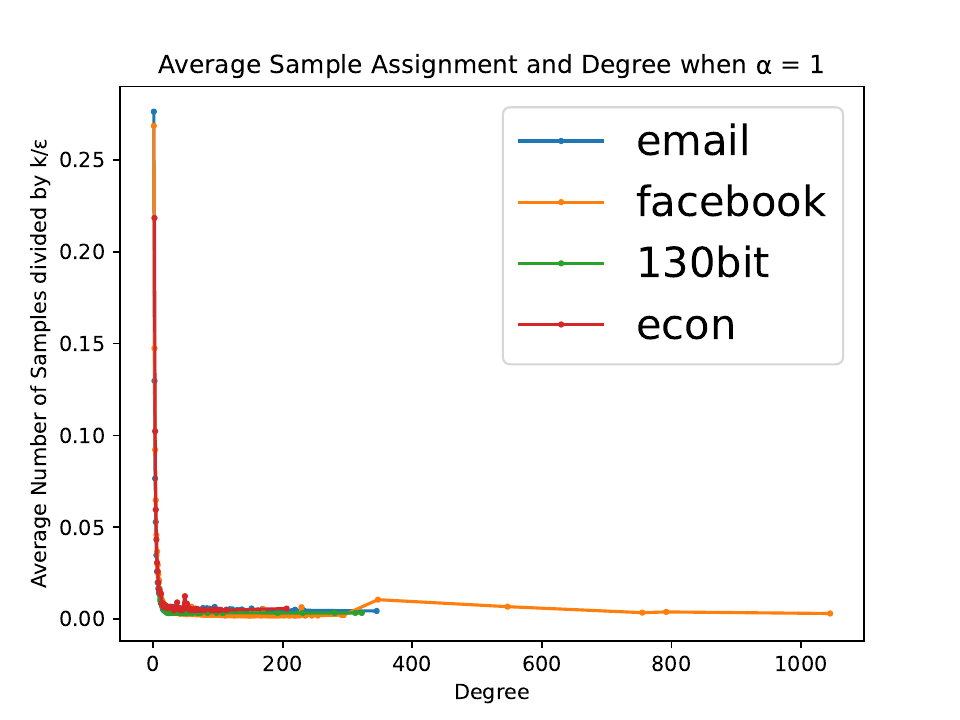}}
    \subfigure[$\alpha = 5$]{ \includegraphics[width = 0.188\linewidth]{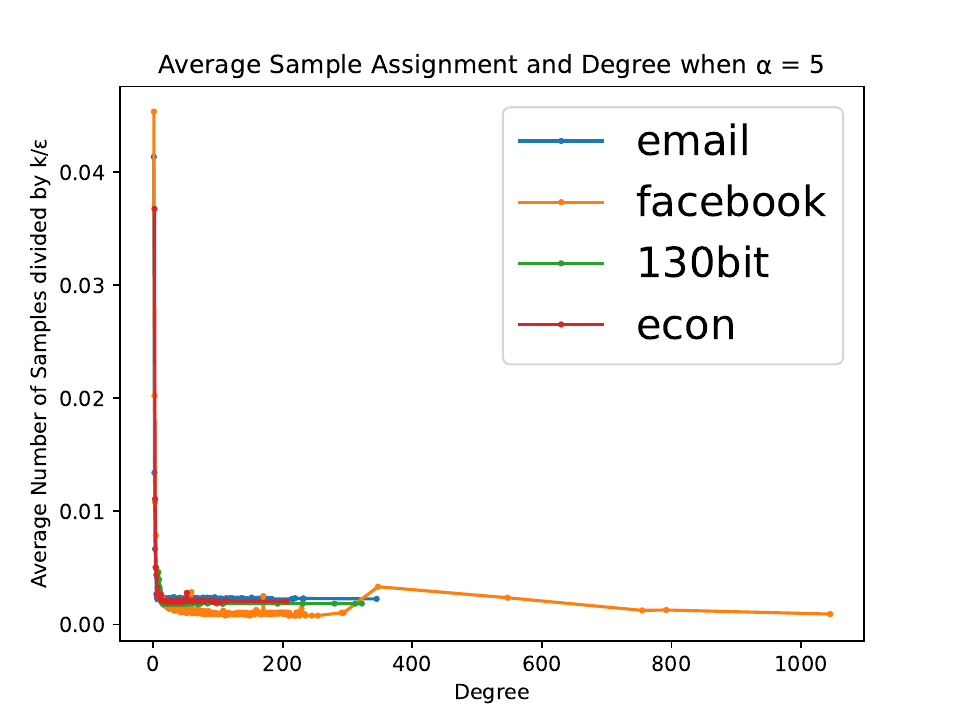}}
    \subfigure[$\alpha = 10$]{ \includegraphics[width = 0.188\linewidth]{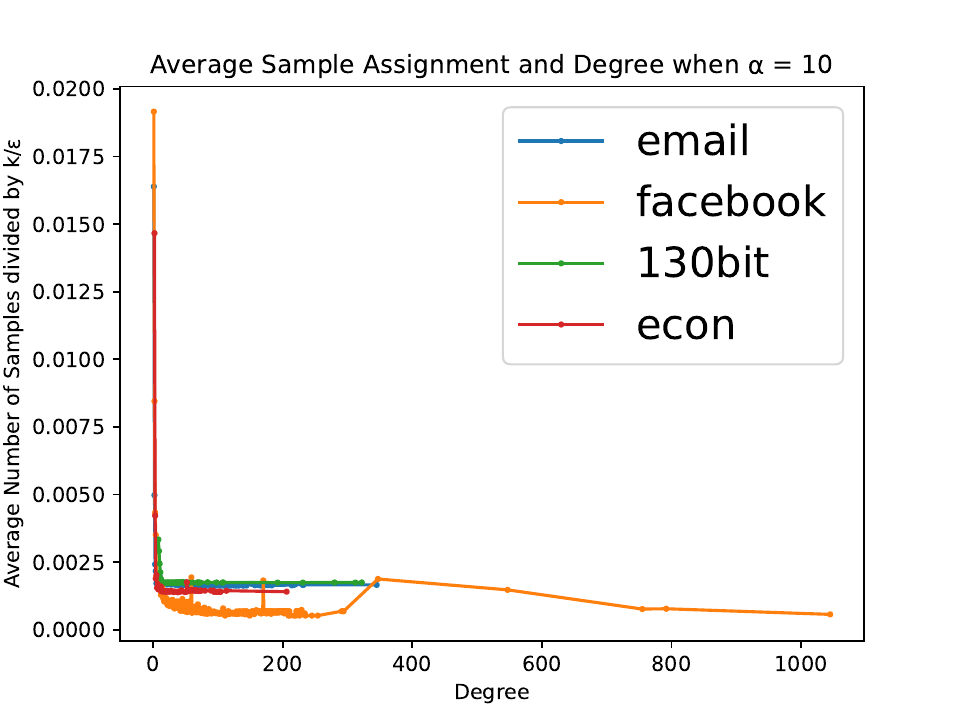}}
    \caption{Relationship between average samples and degree of RN for different $\alpha$.}
    \label{ALL_RN_S}
\end{figure}

\begin{figure}[htbp]
    \centering
    \subfigure[$\alpha = 0.01$]{ \includegraphics[width = 0.188\linewidth]{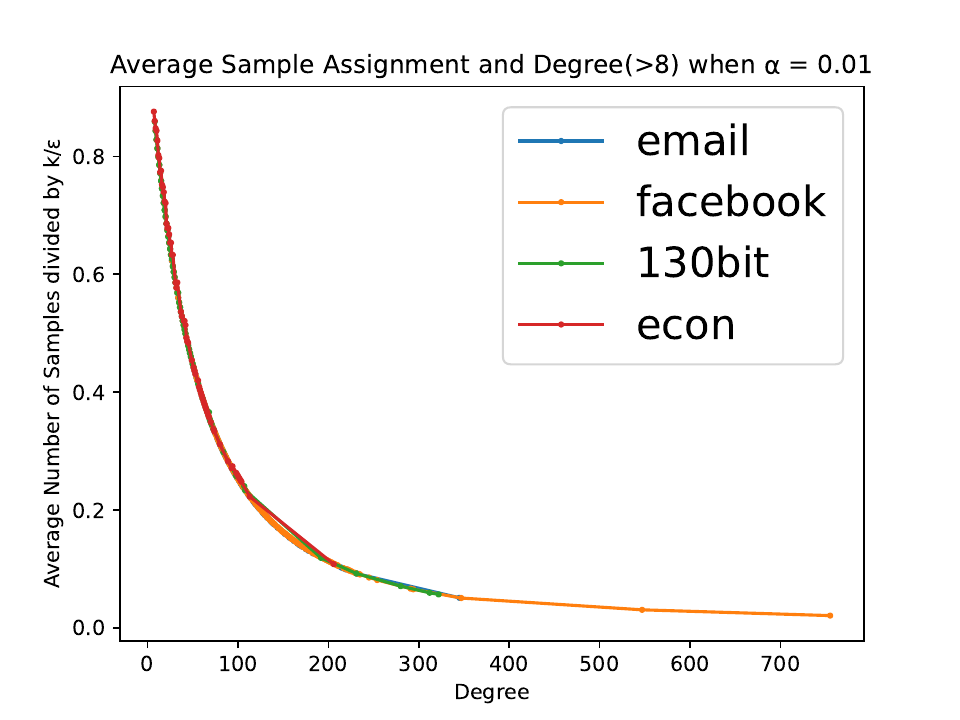}}
    \subfigure[$\alpha = 0.1$]{  \includegraphics[width = 0.188\linewidth]{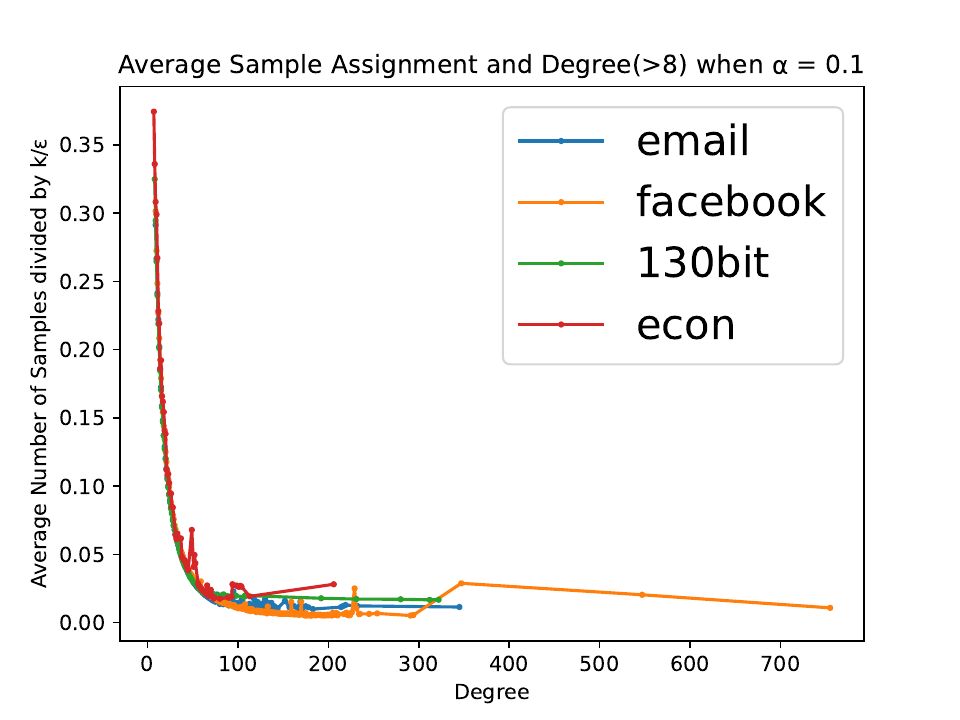}}
    \subfigure[$\alpha = 1$]{  \includegraphics[width = 0.188\linewidth]{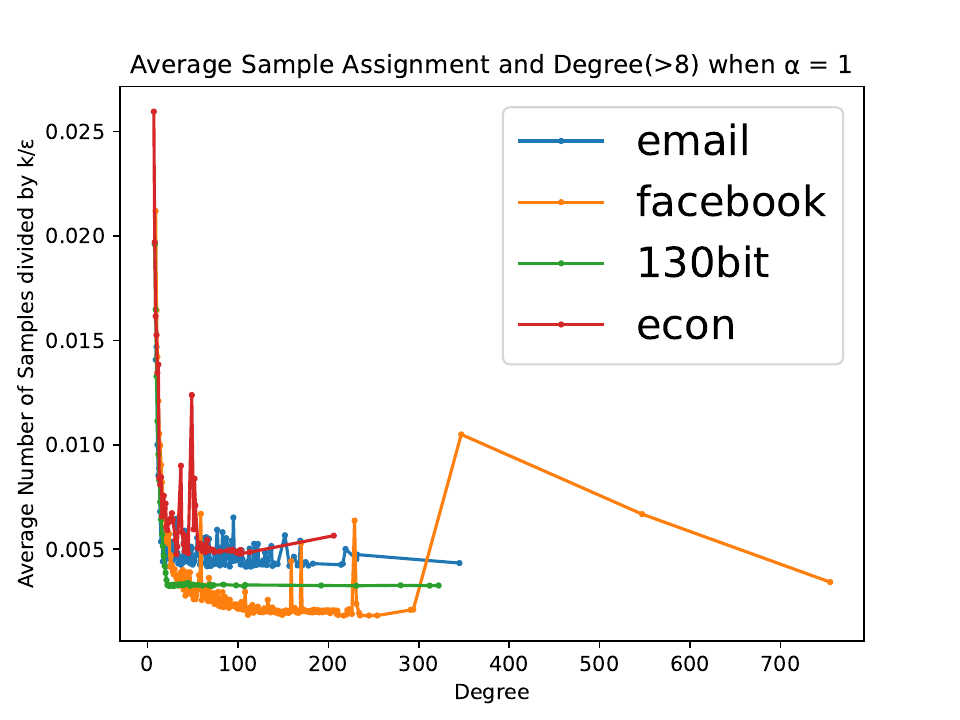}}
    \subfigure[$\alpha = 5$]{ \includegraphics[width = 0.188\linewidth]{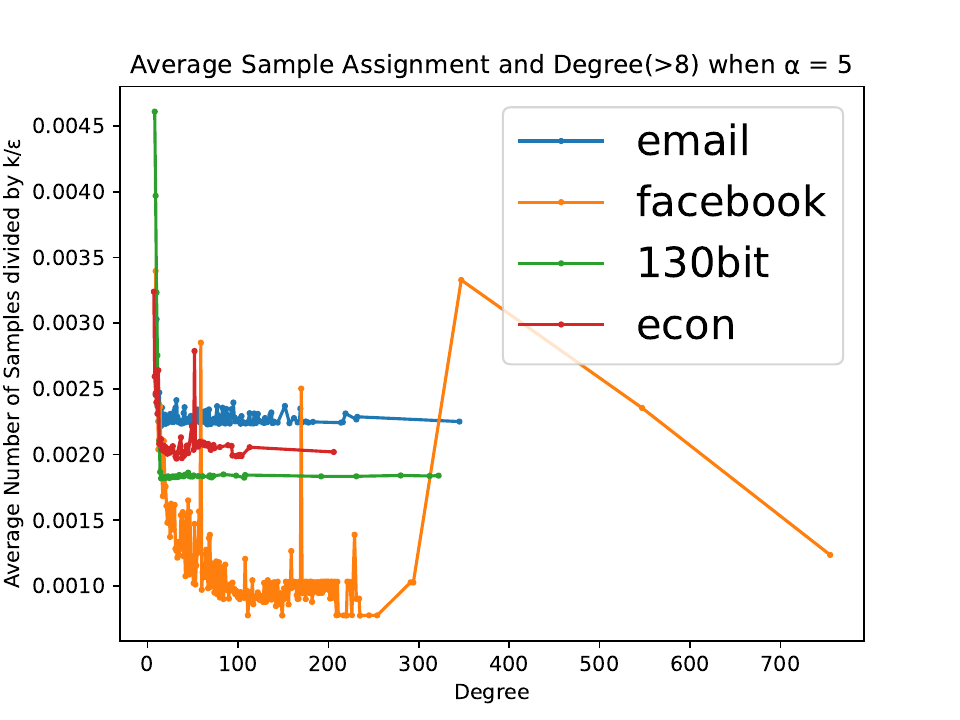}}
    \subfigure[$\alpha = 10$]{ \includegraphics[width = 0.188\linewidth]{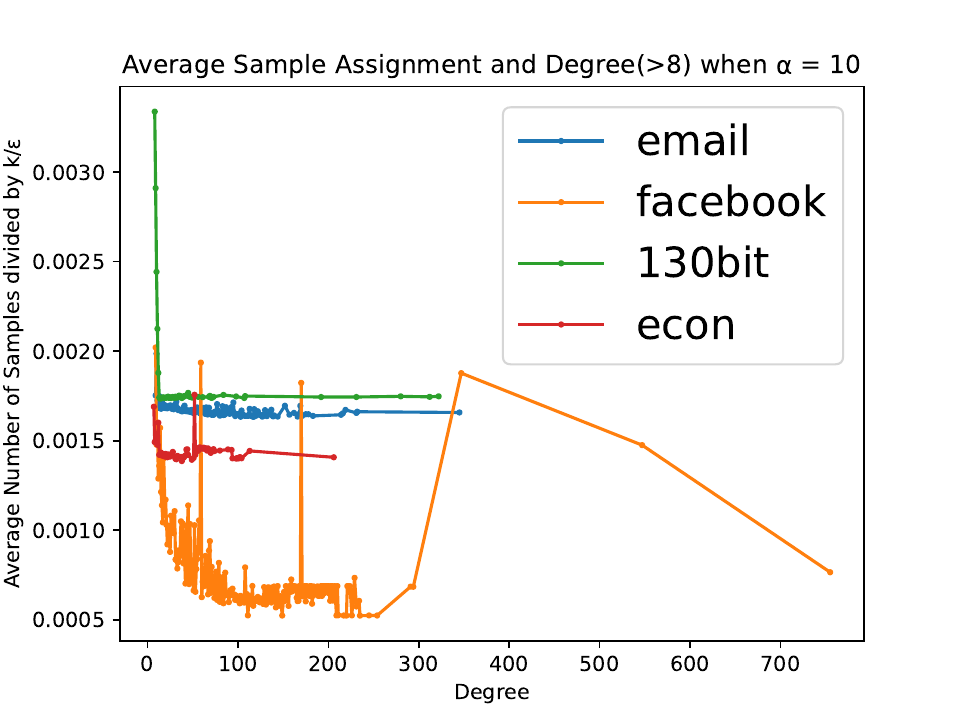}}
    \caption{Relationship between average samples and degree of RN for different $\alpha$ (when $d\ge 8$).}
    \label{ALL_RN8_S}
\end{figure}

\begin{figure}[htbp]
    \centering
    \subfigure[$\alpha = 0.01$]{ \includegraphics[width = 0.188\linewidth]{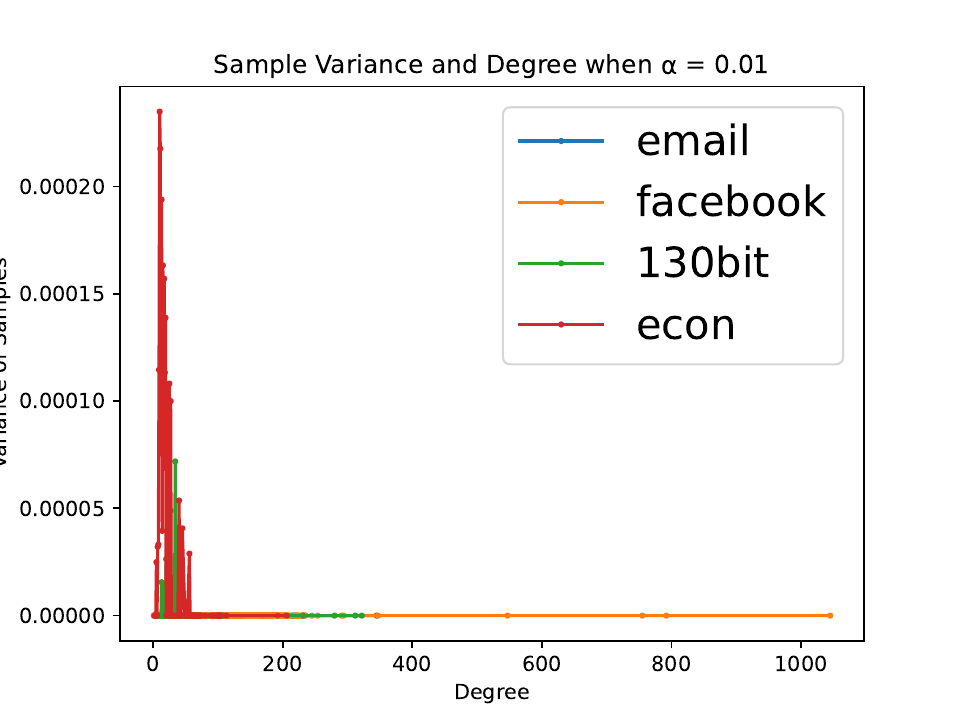}}
    \subfigure[$\alpha = 0.1$]{  \includegraphics[width = 0.188\linewidth]{figures/Degree_and_TSC/Real/var_0.1.pdf}}
    \subfigure[$\alpha = 1$]{  \includegraphics[width = 0.188\linewidth]{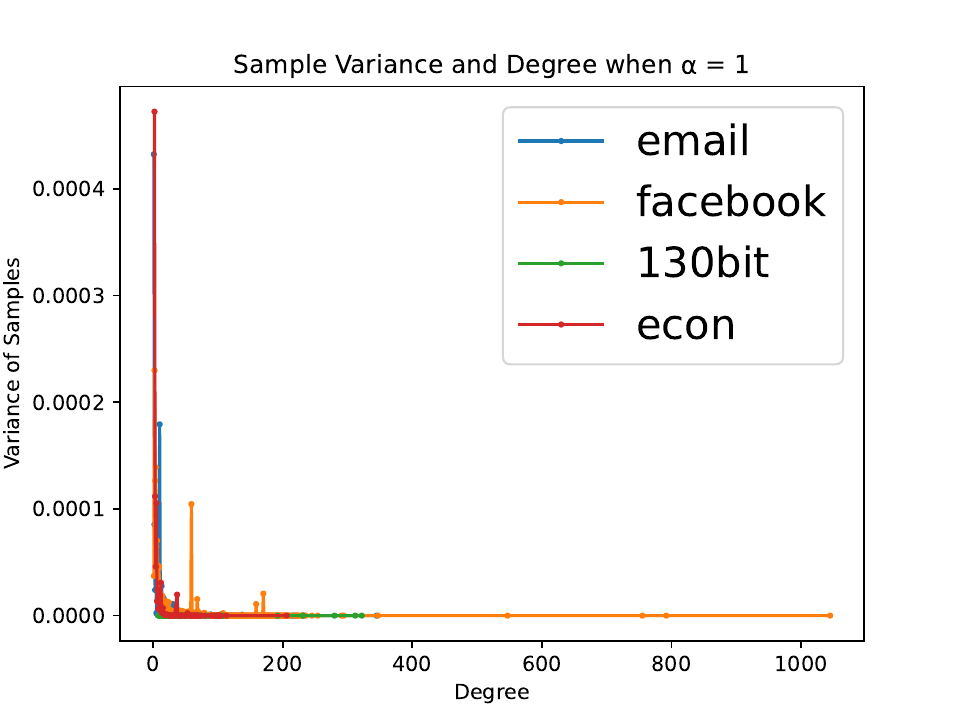}}
    \subfigure[$\alpha = 5$]{ \includegraphics[width = 0.188\linewidth]{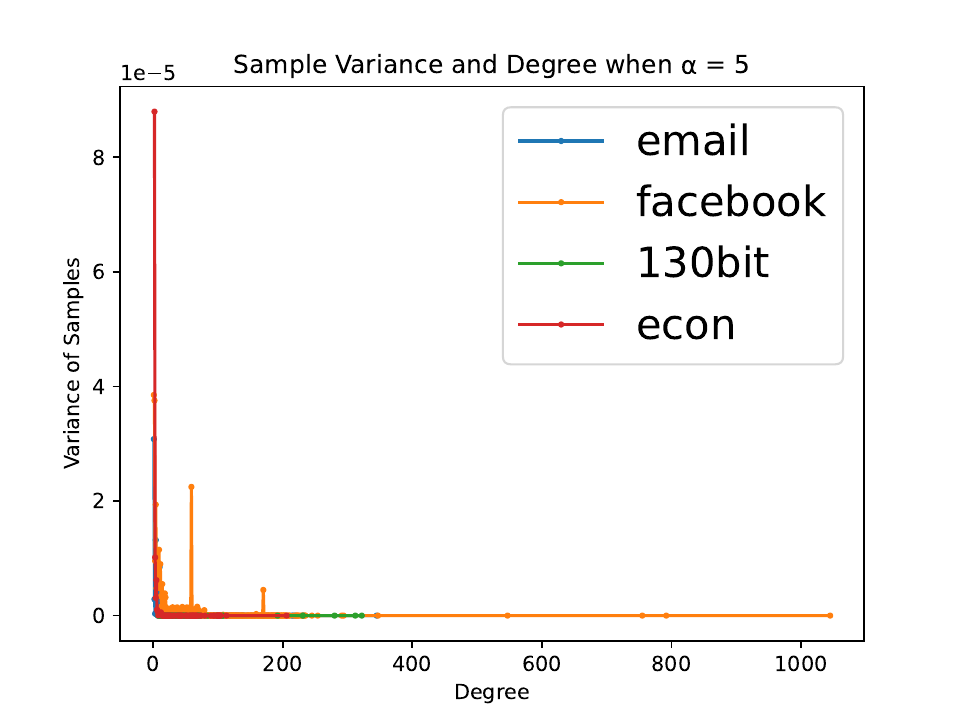}}
    \subfigure[$\alpha = 10$]{ \includegraphics[width = 0.188\linewidth]{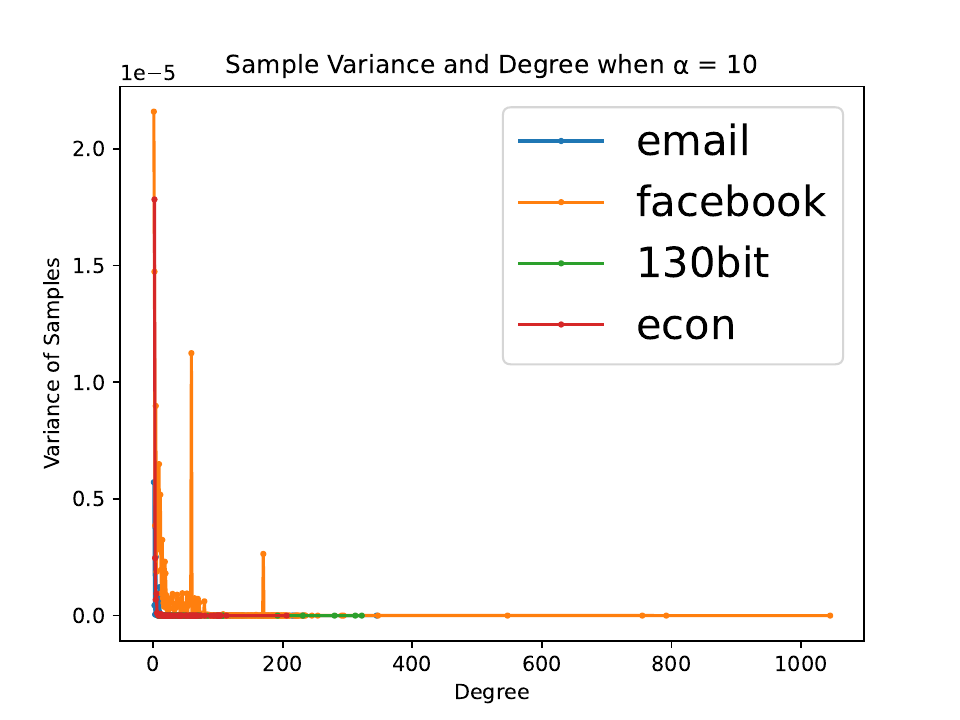}}
    \caption{Relationship between sample variance and degree of RN for different $\alpha$}
    \label{ALL_RN_V}
\end{figure}

Based on the above results, low-degree agents tend to have more samples, and agents with the same degree tend to have the same number of samples. When $\alpha$ is small, the sample assignment seems to have relatively good monotonicity with $d_i$ but as $\alpha$ gets larger, the trend for high degrees becomes more unpredictable. 

Another observation is that as $n$ increases, the samples assigned to each degree decrease when $\alpha$ gets larger. Our theory can not explain this well now, and it is interesting to investigate it more in future work.

\subsection{Tightness of bounds} \label{subsec:tightness bounds}
In this section, we test the tightness of our upper and lower bound in Theorem \ref{thm: general tight bound}. We also study the impact of the magnitude of influences factors on the performance of bounds. The influence factors $v_{ij}$s are uniformly generated at random from $[0,0.01], [0,0.1], [0,1], [0,5], [0,10]$.


Our synthesized graphs are generated for $n = 100, 200, 400, 600, 800, 1000$. For each $n$, we generate $25$  networks with different structures and $v_{ij}$s (The detailed process will be discussed later). For all these networks, we calculate their optimal solution, the upper bound and the lower bound indicated in Theorem \ref{thm: general tight bound}.  We measure the performance of bounds in Theorem \ref{thm: general tight bound} by relative error (i.e $\frac{|U \text{ or } L - \sum_{i \in [n]} m_i^*|}{\sum_{i \in [n]} m_i^*}$ where $U$ is the upper in Theorem \ref{thm: general tight bound}) and $L$ is the max of 1 and the lower bound in Theorem \ref{thm: general tight bound}. In the following plots, the x-axis is the relative errors of upper/lower bounds and the y-axis is the frequency (number of networks in each bin).

For scale-free networks, for each $n$, we randomly generate $15$ networks with attachment parameter $2$ and $10$ networks with attachment parameter $3$. The influence factors of these networks are also randomly generated in a given range. The error plot is given in Figure \ref{ALL_BOUND_SF}.

\begin{figure}[htbp]
    \centering
    \subfigure[$v_{ij} \in \text{$[0, 0.01]$}$]{ \includegraphics[width = 0.188\linewidth]{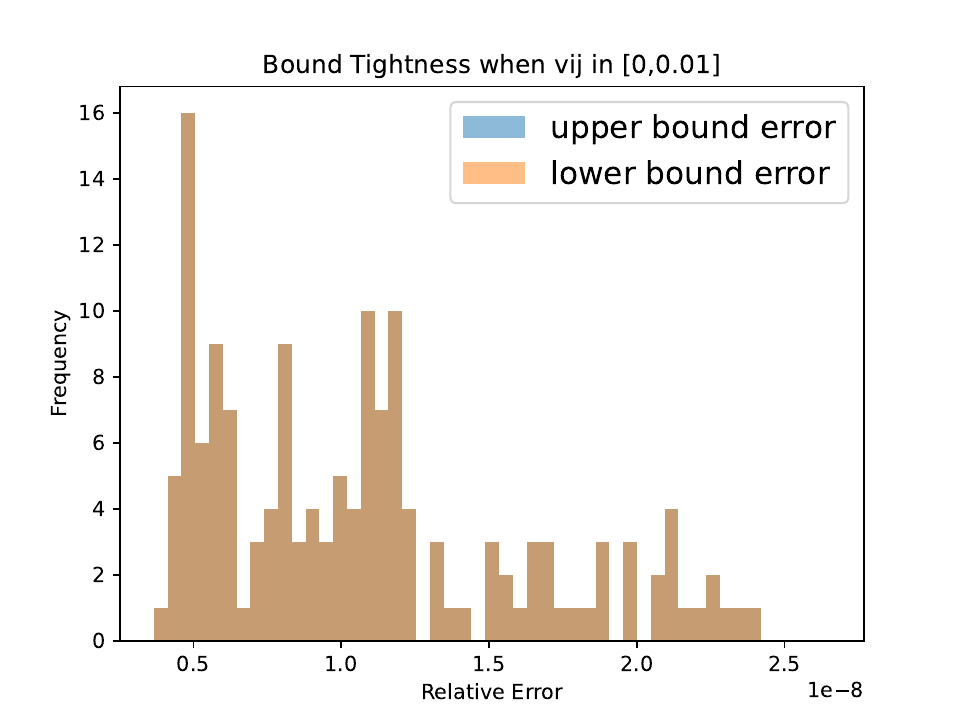}}
    \subfigure[$v_{ij} \in \text{$[0, 0.1]$}$]{  \includegraphics[width = 0.188\linewidth]{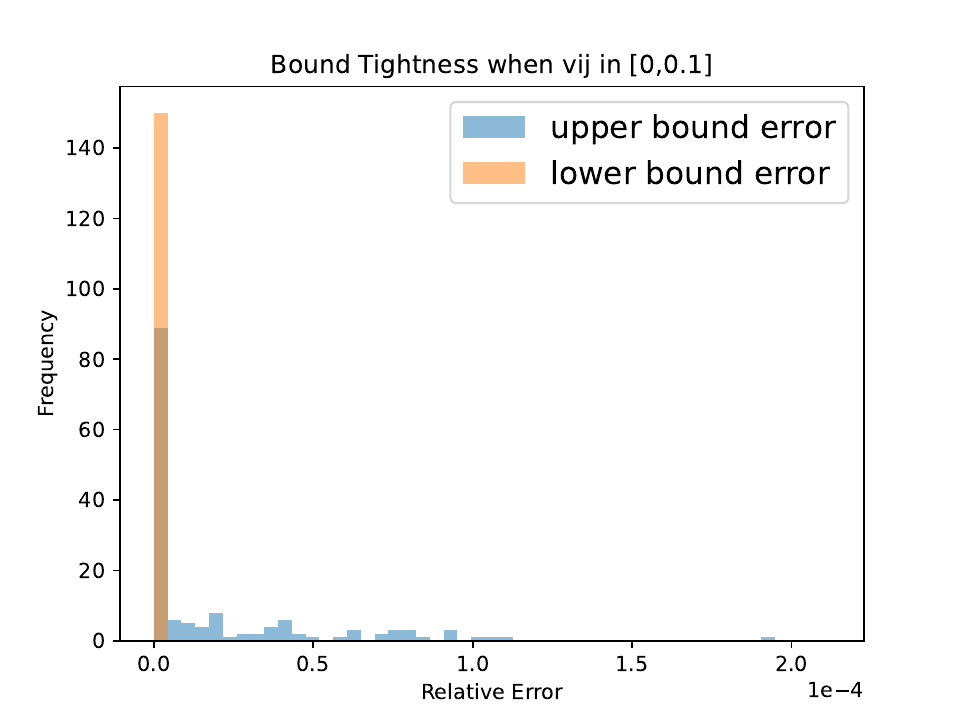}}
    \subfigure[$v_{ij} \in \text{$[0, 1]$}$]{  \includegraphics[width = 0.188\linewidth]{figures/weighted_tightness/PL/IF=1.pdf}}
    \subfigure[$v_{ij} \in \text{$[0, 5]$}$]{ \includegraphics[width = 0.188\linewidth]{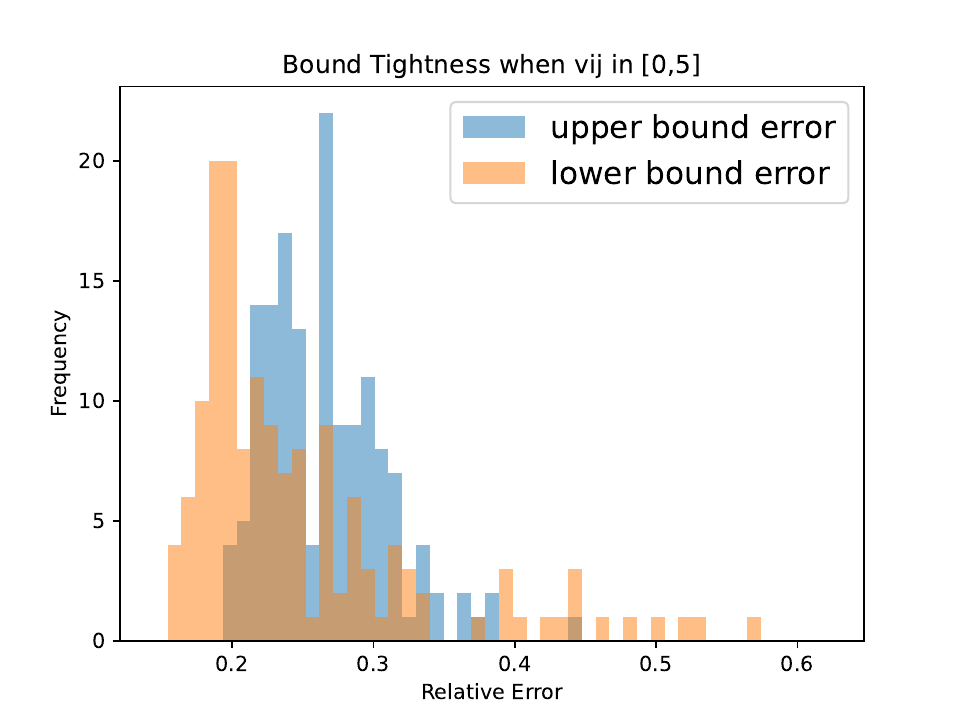}}
    \subfigure[$v_{ij} \in \text{$[0, 10]$}$]{ \includegraphics[width = 0.188\linewidth]{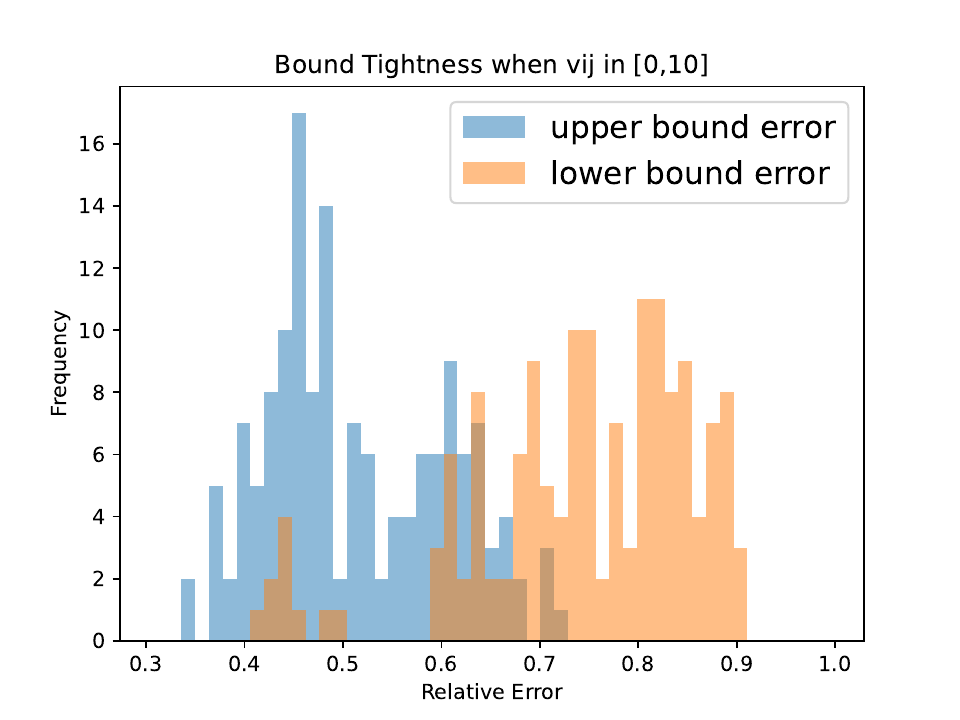}}
    \caption{Bounds of SF for different influence factors $v_{ij}$s}
    \label{ALL_BOUND_SF}
\end{figure}


For random regular graphs, for each $n$, we generate graphs with degree roughly equal to $n^{0.25 + 0.1k}$ where $k = \{0,1,2,3,4\}$. For each degree, we generate $5$ graphs. The influence factors of these networks are also randomly generated in a given range. The error is given in Figure \ref{ALL_BOUND_RR}.

\begin{figure}[htbp]
    \centering
    \subfigure[$v_{ij} \in \text{$[0, 0.01]$}$]{ \includegraphics[width = 0.188\linewidth]{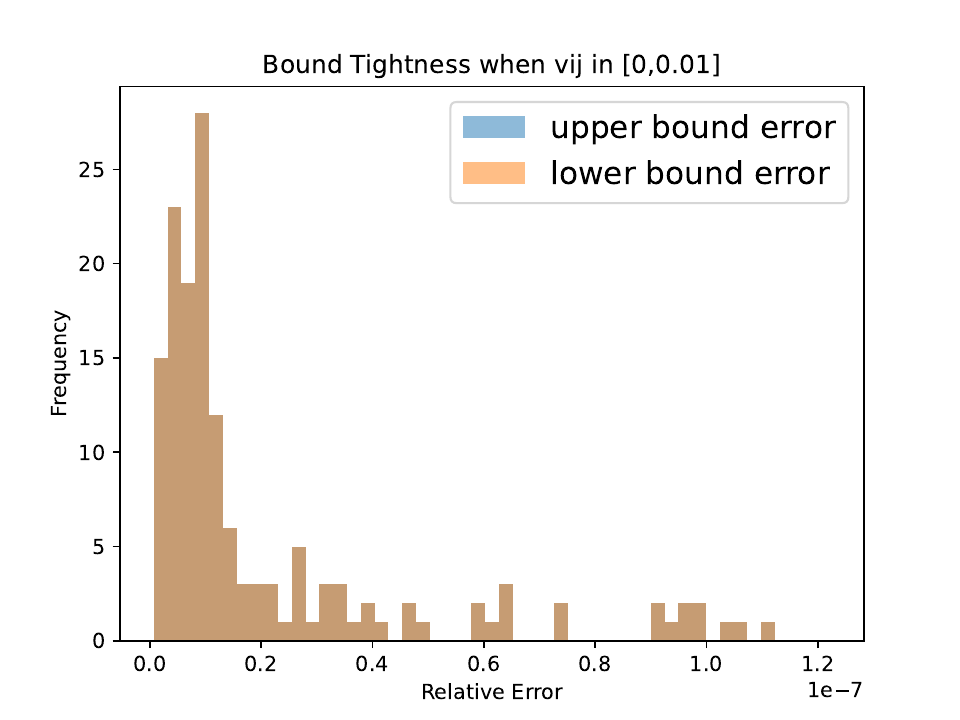}}
    \subfigure[$v_{ij} \in \text{$[0, 0.1]$}$]{  \includegraphics[width = 0.188\linewidth]{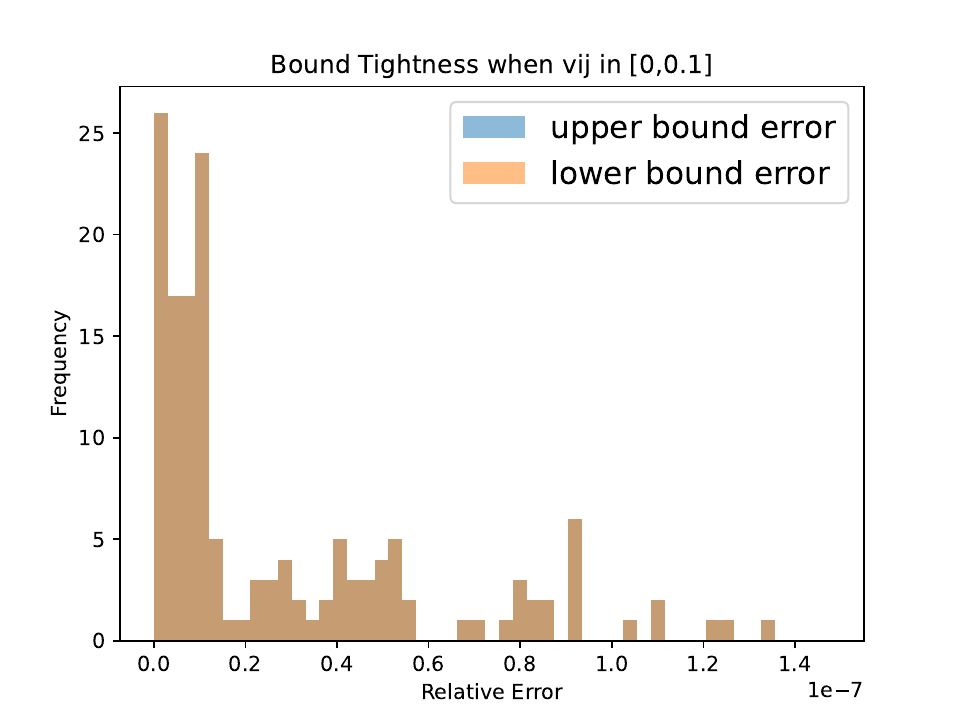}}
    \subfigure[$v_{ij} \in \text{$[0, 1]$}$]{  \includegraphics[width = 0.188\linewidth]{figures/weighted_tightness/RR/IF=1.pdf}}
    \subfigure[$v_{ij} \in \text{$[0, 5]$}$]{ \includegraphics[width = 0.188\linewidth]{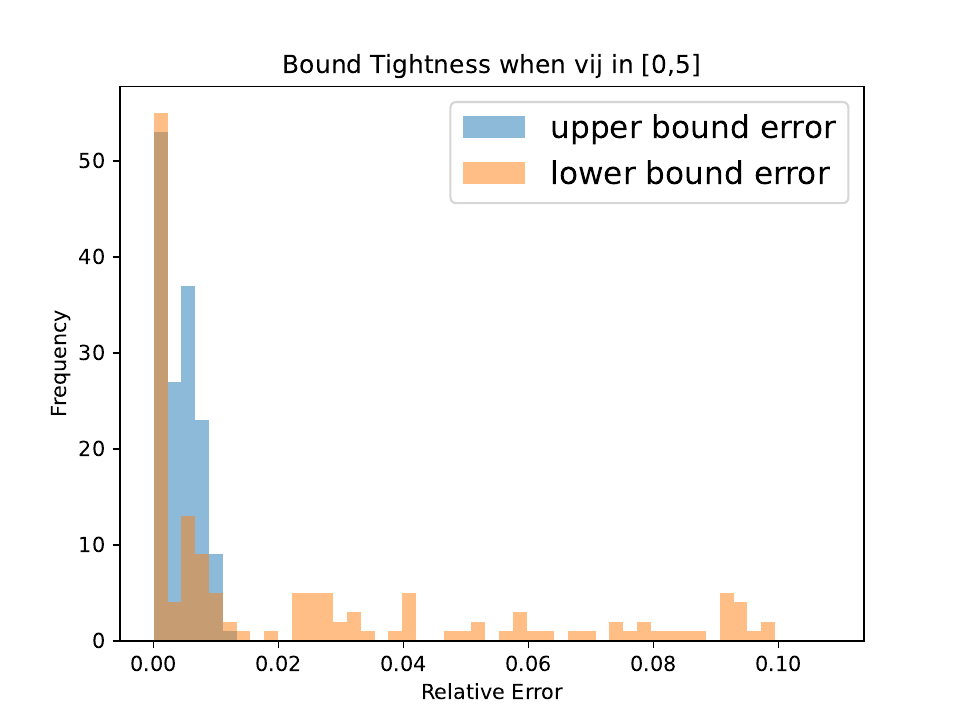}}
    \subfigure[$v_{ij} \in \text{$[0, 10]$}$]{ \includegraphics[width = 0.188\linewidth]{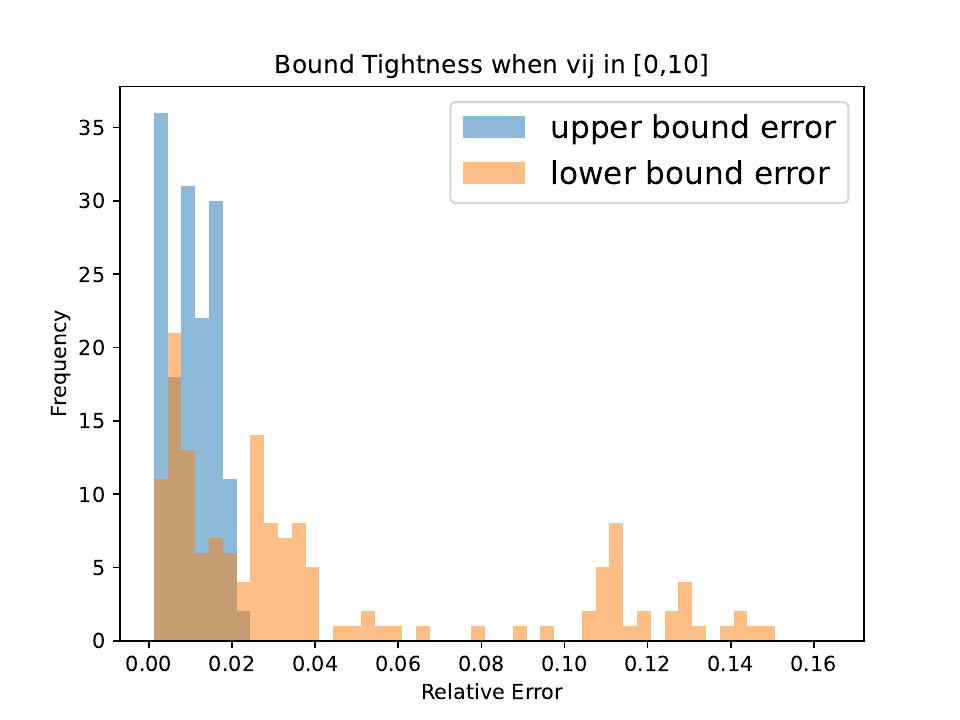}}
    \caption{Bounds of RR for different influence factors $v_{ij}$s}
    \label{ALL_BOUND_RR}
\end{figure}


For Erdos-Renyi graphs, for each $n$, we generate $25$ graphs with edge-connection probability $0.25$.  The influence factors of these networks are also randomly generated in a given range. The error plot is given in Figure \ref{ALL_BOUND_ER}.

\begin{figure}[htbp]
    \centering
    \subfigure[$v_{ij} \in \text{$[0, 0.01]$}$]{ \includegraphics[width = 0.188\linewidth]{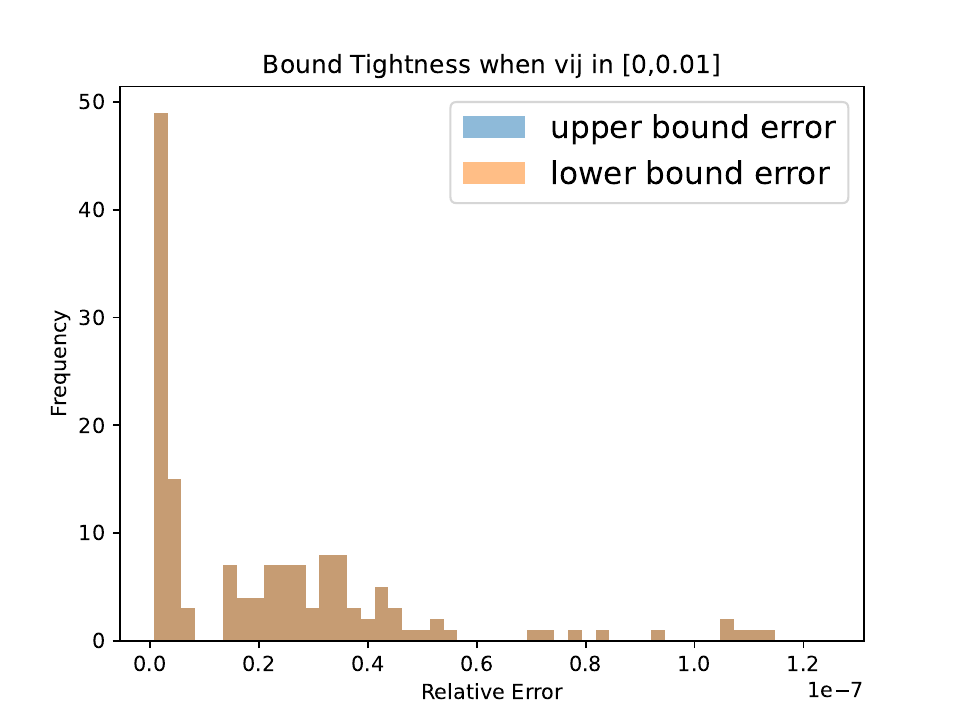}}
    \subfigure[$v_{ij} \in \text{$[0, 0.1]$}$]{  \includegraphics[width = 0.188\linewidth]{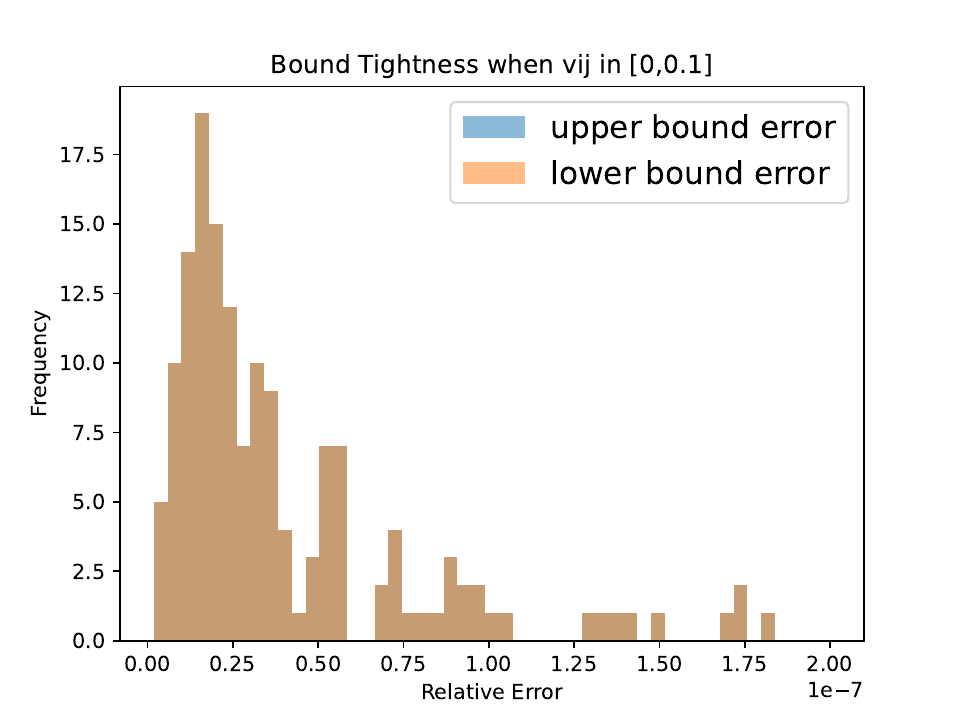}}
    \subfigure[$v_{ij} \in \text{$[0, 1]$}$]{  \includegraphics[width = 0.188\linewidth]{figures/weighted_tightness/ER/IF=1.pdf}}
    \subfigure[$v_{ij} \in \text{$[0, 5]$}$]{ \includegraphics[width = 0.188\linewidth]{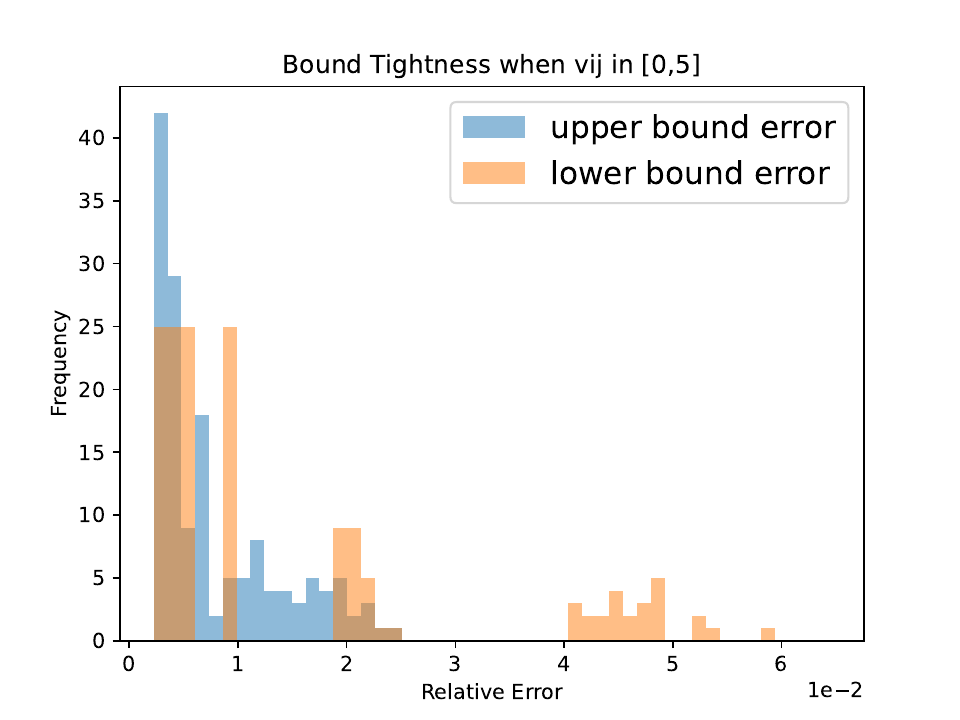}}
    \subfigure[$v_{ij} \in \text{$[0, 10]$}$]{ \includegraphics[width = 0.188\linewidth]{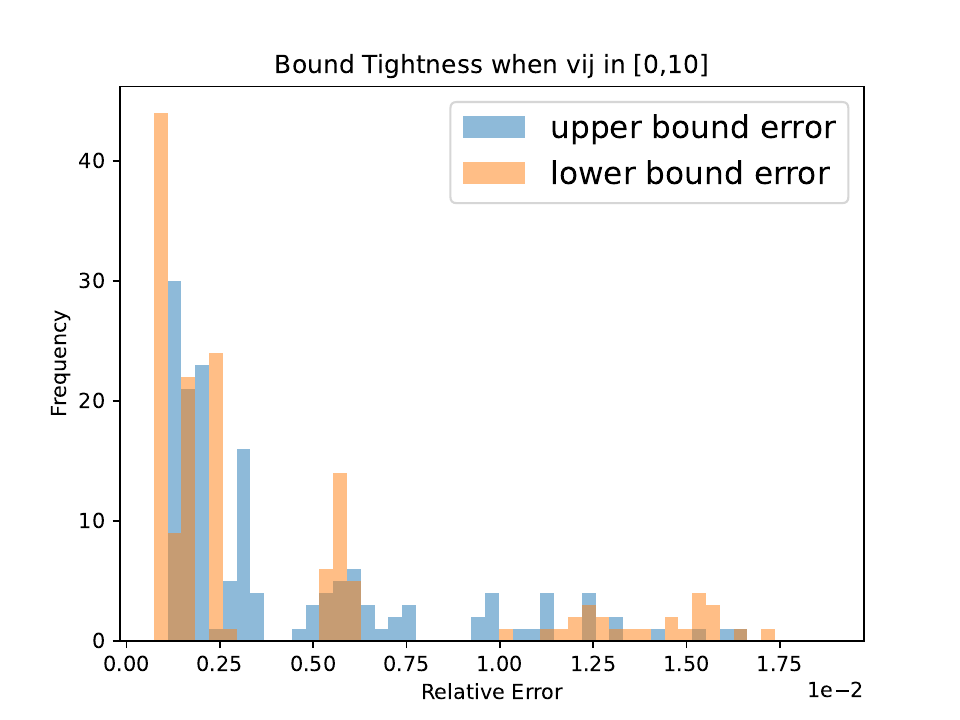}}
    \caption{Bounds of ER for different influence factors $v_{ij}$s}
    \label{ALL_BOUND_ER}
\end{figure}


In Table \ref{ALL_BOUND_REAL}, we show the relative error for bounds of different real-world networks for different $v_{ij}$s. Let $err_u$ and $err_l$ be the maximum relative error of the upper and lower bound in Theorem \ref{thm: general tight bound} respectively among these random experiments. 

\begin{table}[H]
\centering
\begin{tabular}{| p{6em} | p{7em} | p{5em}| p{5em}|}
    \hline
    \textbf{Network} & Upper bound of $v_{ij}$ & \textbf{$err_u$} & \textbf{$err_l$}  \\ \hline
    ego-Facebook & 0.01 & $2.2\mathrm{e}{-6}\%$ & $6.6\mathrm{e}{-7}\%$ \\ \hline
    ego-Facebook & 0.1 & $0.4\%$ & $3\mathrm{e}{-3}\%$ \\ \hline
    ego-Facebook & 1 & $27\%$ & $24\%$ \\ \hline
    ego-Facebook & 5 & $89\%$ & $74\%$ \\ \hline
    ego-Facebook & 10 & $98\%$ & $70\%$ \\ \hline
    
    Econ & 0.01 & $1\mathrm{e}{-6}\%$ & $1\mathrm{e}{-14}\%$ \\ \hline
    Econ & 0.1 & $0.1\%$ & $3\mathrm{e}{-5}\%$ \\ \hline
    Econ & 1 & $6.6\%$ & $1.3\%$ \\ \hline
    Econ & 5 & $64\%$ & $88\%$ \\ \hline
    Econ & 10 & $112\%$ & $75\%$  \\ \hline
    
    Email-Eu & 0.01 & $3\mathrm{e}{-7}\%$  & $1\mathrm{e}{-13}\%$ \\ \hline
    Email-Eu & 0.1 & $0.23\%$ & $4.3\mathrm{e}{-5}\%$  \\ \hline
    Email-Eu & 1 & $30\%$ & $21\%$ \\ \hline
    Email-Eu & 5 & $117\%$ & $76\%$ \\ \hline
    Email-Eu & 10 & $151\%$ & $60\%$  \\ \hline
    
    130bit & 0.01 & $2\mathrm{e}{-7}\%$ & $3\mathrm{e}{-13}\%$   \\ \hline
    130bit & 0.1 & $0.8\%$ & $8\mathrm{e}{-3}\%$ \\ \hline
    130bit & 1 & $41\%$ & $36\%$  \\ \hline
    130bit & 5 & $65\%$ & $11\%$ \\ \hline
    130bit & 10 & $42\%$ & $3.4\%$  \\ \hline
\end{tabular}
\caption{Bounds for real-world networks for different influence factors $v_{ij}$s}
\label{ALL_BOUND_REAL}
\end{table}

    
    
    

From all experiments above, we can observe that the performance of the bounds in Theorem \ref{thm: general tight bound} is excellent when the influence factors are small but get worse when the influence factors become larger.

The reason why these two bounds get good performance for small influence factors has the following intuitive explanation. Consider the original optimization Equation \ref{eq:opt}, when all influence factors are zero, the optimal value is achieved when we set all inequality of the first constraint to equality. When influence factors are small, the optimal solution will also tend to be the value derived in this way, which is exactly the method we get our upper and lower bound.

Through all these experiments, an obvious result is that when $\alpha$ gets large, the total sample complexity decreases. However, our upper bound in Theorem \ref{thm: degree bound}, by assigning each agent $\frac{\alpha + 1}{\alpha d_i + 1}\cd\frac{k}{\epsilon}$, does not capture this trend. It seems that we can remove the $\alpha$ in the numerator to improve this bound and we leave it as future work.
